\documentclass[journal]{IEEEtran}
\IEEEoverridecommandlockouts
\usepackage[noadjust]{cite}
\usepackage{amsmath,amssymb,amsfonts}
\usepackage{algorithmic}
\usepackage{algorithm}
\usepackage{graphicx, array}
\usepackage{textcomp,dsfont}
\usepackage{xcolor,comment}
\usepackage{nicefrac,xfrac}
\usepackage{subcaption}
\usepackage{flushend}
\usepackage{tikz}%
\usetikzlibrary{arrows.meta,positioning,calc}%

\newtheorem{theorem}{Theorem}
\newtheorem{corollary}{Corollary}
\newtheorem{lemma}{Lemma}
\newtheorem{identity}{Identity}
\newtheorem{definition}{Definition}
\newtheorem{remark}{Remark}
\newtheorem{example}{Example}%

\def\BibTeX{{\rm B\kern-.05em{\sc i\kern-.025em b}\kern-.08em
    T\kern-.1667em\lower.7ex\hbox{E}\kern-.125emX}}

\begin{document}

\title{Zak-OTFS: A Predictable Physical Layer \\ for Communications and Sensing}

\author{Sandesh Rao Mattu$^*$, Nishant Mehrotra$^*$, Venkatesh Khammammetti, and Robert Calderbank~\IEEEmembership{Life Fellow,~IEEE}
\thanks{This work is supported by the National Science Foundation under grants 2342690 and 2148212, in part by funds from federal agency and industry partners as specified in the Resilient \& Intelligent NextG Systems (RINGS) program, and in part by the Air Force Office of Scientific Research under grants FA 8750-20-2-0504 and FA 9550-23-1-0249. \\ The authors are with the Department of Electrical and Computer Engineering, Duke University, Durham, NC, 27708, USA (email: sandesh.mattu@duke.edu, nishantm\allowbreak@alumni.rice.edu,~\allowbreak \{venkatesh.\allowbreak khammammetti\allowbreak,~\allowbreak robert.calderbank\}\allowbreak@duke\allowbreak.edu). \\ $*$ denotes equal contribution.}
}


\maketitle
\begin{abstract}
This tutorial derives the mathematical foundations of what it means for a carrier waveform to be predictable and non-selective. We focus on Zak-OTFS, where each carrier waveform is a pulse in the delay-Doppler (DD) domain, formally a quasi-periodic localized function with specific periods along delay and Doppler. Viewed in the time domain, the Zak-OTFS carrier is realized as a pulse train modulated by a tone (termed a pulsone).

We start by providing physical intuition, describing what it means for the Zak-OTFS carrier waveforms to be geometric modes of the Heisenberg-Weyl (HW) group of discrete delay and Doppler shifts that define the discrete-time communication model. In fact, we show that these geometric modes are common eigenvectors of a maximal commutative subgroup of our discrete HW group.

When the channel delay spread is less than the delay period, and the channel Doppler spread is less than the Doppler period, we show that the Zak-OTFS input-output (I/O) relation is predictable and non-selective. Given the I/O response at one DD point in a frame, it is possible to predict the I/O response at all other points, without recourse to some mathematical model of the channel. While it may be intuitive that geometric modes of the HW group are predictable and non-selective wireless carriers, this is not a requirement. We provide a necessary and sufficient condition that depends on the ambiguity properties of the basis of carrier waveforms. In fact, we show that the structure of a pulse train modulated by a Hadamard matrix is common to several families of waveforms proposed for 6G, including Zak-OTFS, AFDM, OTSM and ODDM.

Given bandwidth $B$ and time duration $T$, we describe the discrete-time communications model that uses $BT$ orthonormal Zak-OTFS carrier waveforms. This is Nyquist signaling, and Faster-than-Nyquist signaling requires that we sacrifice orthogonality. The standard approach is to reduce the symbol time and resolve the inter-symbol interference. We describe an alternative approach using mutually unbiased bases that retains the original symbol period.

We repurpose the discrete-time communications model and describe how the Zak-OTFS carrier waveform can be used for radar sensing. We show that by choosing a radar waveform that is a common eigenvector of a maximal commutative subgroup of the discrete HW group, we are able to reduce the complexity of calculating the radar cross-ambiguity.

The peaks in the I/O relation are associated with dominant scatterers and so the I/O relation for Zak-OTFS evolves slowly from frame to frame. We describe how inter-frame predictability makes differential communication possible. Data detected in one frame can be used to generate a fresh channel estimate that in turn supports data detection in the next frame. Inter-frame predictability makes it possible to increase spectral efficiency by using data as pilots.

We describe how to optimize filter design in the DD domain to balance localization (sensing) and prevention of inter-symbol interference on the information lattice (communication). We describe how the choice of filter enables massive random access on a wireless uplink subject to mobility and delay spread. Successive interference cancellation plays an essential role in coded random access, where users transmit packets in multiple slots. We show that Gaussian filters enable accurate channel estimation across slots, making it possible to perform successive interference cancellation.

Delay-Doppler methods (such as Zak-OTFS) and time-frequency methods (such as OFDM) provide different ways of balancing channel estimation and equalization. In OFDM, equalization is relatively simple, whereas acquiring the I/O relation can be challenging. In Zak-OTFS, the reverse is true, and we describe how to reduce complexity by performing equalization in the frequency domain. We also describe a neural receiver with the potential to reduce the complexity of equalization dramatically.

Finally, since seeing is believing, we also describe over-the-air demonstrations of Zak-OTFS at sub-6 GHz, at mmWave, and at sub-THz.
\end{abstract}


\section{Introduction}
\label{sec:intro}

\IEEEPARstart{}{}

\begin{figure*}
    \centering
    \includegraphics[width=0.8\linewidth]{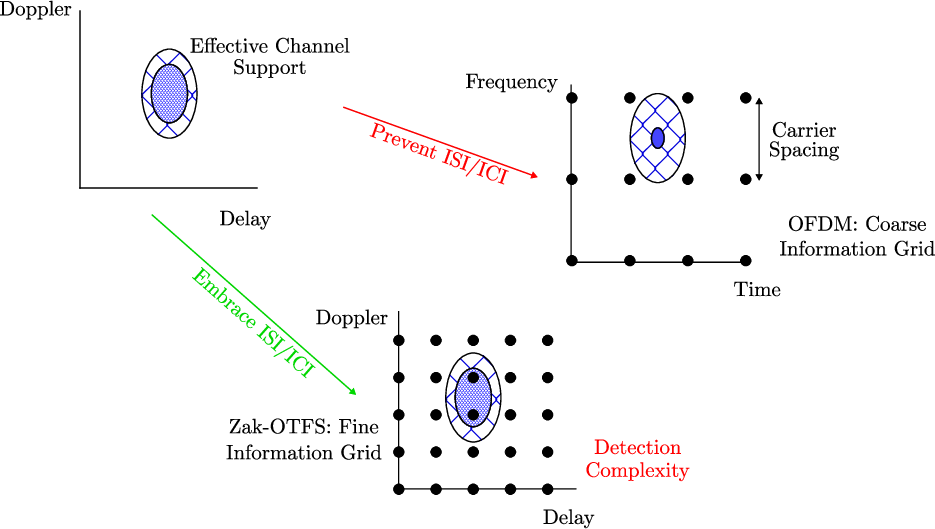}
    \caption{A pictorial representation of the difference between OFDM and Zak-OTFS. OFDM prevents ISI/ICI, that is the subcarrier spacing is adjusted to avoid interference, while Zak-OTFS embraces ISI/ICI. This results in poor resource utilization in OFDM but low equalization complexity while resource utilization is good in Zak-OTFS at the cost of equalization complexity.}
    \label{fig:prevent_vs_embrace}
\end{figure*}

Every ten years there is a new wireless standard, and each new standard is an opportunity to think about fundamentals of wireless communication. Voice was the dominant application in 1-3G and circuit switching was used to support a modest number of persistent but relatively low-rate connections. As data supplanted voice, the air interface needed to change in order to quickly distribute resources according to widely and rapidly varying traffic demands. 
Link adaptation based on approximate channel state information at the transmitter made it possible to satisfy user demand for high-speed data in 4G and led to the transition from CDMA to OFDMA.
This tutorial explores the value to 6G of delay-Doppler signal processing in light of higher speeds, higher frequencies, the balance between uplink and downlink, the integration of sensing and communication, and the introduction of artificial intelligence (AI) and machine learning (ML). Each Section concludes with a list of research problems that identify possible research directions.

Section~\ref{sec:foundations} provides the mathematical foundation for signal processing in the delay-Doppler domain. In the 1960s, Slepian, Landau and Pollack~\cite{pollak1961_pswf1,pollak1961_pswf2} developed the mathematical foundation of what it means for a signal to be limited to time $T$ and bandwidth $B$. We consider a discrete time model where $T$ is partitioned into $N$ equal subintervals and $B$ is partitioned into $M$ equal subintervals. The number of degrees of freedom is the time-bandwidth product $BT (=MN)$ and we consider bases of $BT$ carrier waveforms in the context of doubly selective channels. In both the time domain (TD) and the delay-Doppler (DD) domain, we consider how the basis of carrier waveforms interacts with the discrete delay and Doppler shifts forming the Heisenberg-Weyl (HW) group.
 
We address the fundamental question of what attributes we should seek in a basis of carrier waveforms, and we focus on predictability and non-selectivity because these are the characteristics that improve performance of ML applications. The Zak-OTFS basis is an example where predictability and non-selectivity are intuitive because the HW group permutes the carrier waveforms. We say the Zak-OTFS carrier waveforms are geometric modes of the HW group. We prove that in this case, the carrier waveforms are common eigenvectors of a maximal commutative subgroup of the HW group. We describe how to apply a discrete affine Fourier transform (DAFT) to construct new waveforms that are also predictable and non-selective. We provide examples of waveforms constructed in this way that were proposed for 6G. Conjugation by DAFT preserves the HW group, so these new waveforms are also common eigenvectors of maximal commutative subgroups of HW. However, predictability and non-selectivity do not require that the carrier waveforms be common eigenvectors of a maximal commutative subgroup. We provide a necessary and sufficient condition that depends on the ambiguity functions of the individual carrier waveforms. We provide new examples, not equivalent to Zak-OTFS, and again, some of these waveforms were proposed for 6G.

Section~\ref{sec:zak_intro} introduces Zak-OTFS. We begin by describing how the Zak-transform acts as a unitary transform between the Hilbert space of sequences in the TD with period $MN$ and the Hilbert space of $M \times N$ quasi-periodic arrays in the DD domain. A particular instance is the correspondence between pulses in the DD domain and modulated pulse trains in the TD. The Zak-OTFS basis is parametrized by a delay period and a Doppler period. We describe how to modulate information onto carriers in the DD domain, then apply a transmit filter that restricts the signal to time $T$ and bandwidth $B$, then form a TD signal by applying the inverse Zak transform. We describe how to apply the Zak transform to the received signal, then apply a matched filter, then recover the information in the DD domain. We derive the crystallization condition which describes how predictability depends on the interaction between channel support and the delay and Doppler periods. 

Delay-Doppler methods (Zak-OTFS) and time-frequency methods (OFDM) provide different ways of balancing channel estimation and equalization. In OFDM, once the input / output (I/O) relation is known, equalization is relatively simple, at least when there is no inter-carrier-interference (ICI). However, acquisition of the I/O relation is non-trivial and model dependent. In contrast, equalization is more involved in OTFS due to inter-symbol-interference (ISI), however acquisition of the I/O relation is simple and model free. Acquisition becomes more critical when Doppler spreads measured in kHz make it more and more difficult to estimate channels. The most challenging situation is the combination of unresolvable paths and high channel spreads.
In this scenario, a model-free approach is essential, and Zak OTFS in the crystallization regime maintains performance.

Section~\ref{sec:filter_design} describes a framework for optimizing filter design in the DD domain. A pulse shaping filter is characterized by three metrics: (i) localization to enable accurate I/O relation estimation (sensing), (ii) orthogonality on the information lattice to prevent inter-symbol interference, and (iii) no time and bandwidth expansion beyond $B$ and $T$ to achieve full spectral efficiency. A filter simultaneously meeting all three objectives is ideal for both sensing and communication. We describe each metric in detail and present examples.
 
Section~\ref{sec:predictability} describes how inter-frame predictability improves spectral efficiency. Zak-OTFS provides a compact and sparse DD domain parameterization of the wireless channel, where the parameters map directly to physical attributes of the reflectors that comprise the scattering environment. The effective channel changes only at the speed with which the dominant reflectors change, that is at the speed of physics rather than the speed of wireless. 

We describe how inter-frame predictability makes it possible to design a differential communication scheme for Zak-OTFS systems that alleviates the need for periodic pilot transmission. We describe how data detected in one frame can be used to generate a fresh channel estimate that in turn supports data detection in the next frame, and we propagate these steps forward. The first advantage is that it allows the data symbols to enjoy higher energy since the energy that would otherwise be required for pilot symbols can also be allocated to data symbols. The second advantage is that it allows for full spectral efficiency compared to systems that employ pilots to estimate channels.

Section~\ref{sec:multipleaccess} describes how Zak-OTFS enables massive random access on a wireless uplink subject to mobility and delay spread. An emerging challenge in 6G is support for a massive number of devices that sporadically and unpredictably transmit short data packets to a central base station, typically under constraints on energy consumption and latency. Users upload replicas of their data packet into multiple uplink slots, and when a packet is detected, it is removed from the other slots where it is present. Successive interference cancellation (SIC) is the essential enabling technology, and cancellation requires inter-frame predictability.

Polyanskiy~\cite{Polyanskiy2017:Perspective} proposed a theoretical framework to model grant-free random access for power-constrained additive white Gaussian noise (AWGN) channels, which is now popularly referred to as unsourced random access. He derived an achievable probability of error for random coding as a function of signal-to-noise ratio (SNR), which has served as a benchmark for grant-free random access coding schemes. Polyanskiy demonstrated that the performance of naive random-access schemes such as slotted ALOHA (SA) is significantly far away from the random coding bound. Later works showed that coded random access (CRA) schemes can be optimized for the AWGN channel and that they offer significant performance improvement over naive SA schemes, particularly for a large number of users. Section~\ref{sec:multipleaccess} describes how the combination of coded random access and Zak-OTFS significantly outperforms OFDM-based schemes in channels subject to mobility and delay spread, thereby enabling efficient and reliable grant-free random access for massive machine-type communication (mMTC).

Given bandwidth $B$ and time duration $T$, the number of orthogonal degrees of freedom is the time-bandwidth product $BT$. This is Nyquist signaling, and increasing the rate (Faster-than-Nyquist signaling) requires that we sacrifice orthogonality. The standard approach is to reduce the symbol period and manage the inter-symbol interference. Section VII describes an alternative approach using mutually unbiased bases that retains the original symbol period.

Section~\ref{sec:equalization} explores the choice between avoiding or embracing inter-carrier interference / inter-symbol interference (ICI / ISI); see Fig.~\ref{fig:prevent_vs_embrace}. In OFDM we configure the system so that ICI is of the same magnitude as the noise and we use LDPC codes designed for Gaussian channels to recover the transmitted information. We operate at realistic SNRs by allowing 0.1 block error rate (BLER) and we accept that retransmission will lead to a latency hit. On the positive side we have a one-tap equalizer, on the negative side channel estimation is difficult when Doppler is high.  In Zak-OTFS we can read off the effective channel from the response to a single pilot, but equalization complexity is a challenge. We describe how to reduce complexity by performing equalization in the frequency domain.

Section~\ref{sec:neural_rx} describes a neural receiver with the potential to reduce the complexity of equalization dramatically. Offline AI algorithms have revolutionized image and natural language processing, but the speed of wireless makes it extremely challenging to tailor these algorithms to the wireless domain. Cellular networks are interference-limited rather than noise-limited, there is an extraordinary number of possible interference scenarios, these scenarios change at sub-millisecond speeds, and adaptation relies on a very limited number of over-the-air (OTA) training samples. We start from convolution, which governs transmit and receive signal processing of any signal in any part of the wireless spectrum, and we describe how to design a universal neural receiver. This receiver is designed to invert convolution, and we separate the question of which convolution to invert from the actual deconvolution. The neural network that performs deconvolution is very simple, and we describe how to configure this network by setting weights based on domain knowledge. 
 
Section~\ref{sec:radar} repurposes the discrete-time systems model for Zak-OTFS communications for radar. We start with a target delay resolution of $\nicefrac{1}{B}$ and a Doppler resolution of $\nicefrac{1}{T}$. We define $M$ delay bins and $N$ Doppler bins then work modulo $MN$. We describe how discrete radar with a sensing waveform that is a common eigenvector of a maximal commutative subgroup of HW leads to a bed of nails ambiguity function. Given some knowledge of the scattering environment, we describe how to choose the subgroup. When we apply the DAFT to a Zak-OTFS carrier we obtain a noise-like waveform with low PAPR that is a good fit to radar applications. We describe how we can use multiple waveforms to construct waveform libraries. We then describe a computational benefit of using a common eigenvector of a maximal commutative subgroup. Shifts from the subgroup act on the eigenvector as multiplication by a complex phase and this reduces the complexity of calculating the cross-ambiguity to $O(MN \log N)$ from $O(M^2N^2)$. We conclude by describing how two mutually unbiased spread carriers can be used to implement radar polarimetry.

In 1953, Philip Woodward~\cite{Woodward1953} described how to think of radar in information theoretic terms. He suggested that we view the radar scene as an unknown operator parametrized by delay and Doppler, and that we view radar waveforms as questions that we ask the operator. He measured the goodness of a question in terms of lack of ambiguity in the answer, and he sought questions with good localization in delay and Doppler. The radar waveform he proposed (a train of narrow TD Gaussian pulses modulated with a broad Gaussian envelope) is strikingly similar to the Zak-OTFS carrier waveform (a train of narrow impulses modulated by a sinusoid). The theory of discrete radar described in Section~\ref{sec:radar} runs parallel to the theory of continuous radar developed by Auslander and Tolmieri~\cite{Auslander1985}.

\begin{table}
    \centering
    \caption{Power-delay profile of Veh-A channel model}
    \resizebox{0.9\linewidth}{!}{
    \begin{tabular}{|c|c|c|c|c|c|c|}
         \hline
         Path index $i$ & 1 & 2 & 3 & 4 & 5 & 6 \\
         \hline
         Delay $\tau_i (\mu s)$ & 0 & 0.31 & 0.71 & 1.09 & 1.73 & 2.51 \\
         \hline
         Relative power (dB) & 0 & -1 & -9 & -10 & -15 & -20 \\
         \hline
    \end{tabular}}
    \label{tab:veh_a}
\end{table}

Throughout this tutorial we present simulations for the 3GPP compliant Veh-A channel model~\cite{veh_a} which consists of six channel paths. This channel is representative of real propagation environments since it considers fractional delay and Doppler shifts – the path delays in Table~\ref{tab:veh_a} being non-integer multiples of the delay resolution $1/B$, and the Doppler shifts $\nu_i = \nu_{\max} \cos(\theta_i)$ being non-integer multiples of the Doppler resolution $1/T$, where $\theta_i$ is uniformly distributed in the interval $[0, 2\pi)$.

\emph{Remark:} There are many versions of OTFS, for example, MC-OTFS~\cite{Hadani2017_mcotfs} also referred to as two-step OTFS, DZT-OTFS~\cite{lampel2022onotfs, yogesh2023onthe} also referred to as one-step OTFS among many more. Recent papers, for example~\cite{huang2023channel, halder2025channel}, that report comparison with OTFS use one of these variants. These are different from the Zak-OTFS presented in this paper (For a detailed discussion please refer~\cite{bitspaper1, bitspaper2}).

\textit{Notation:} $x$ denotes a complex scalar, $\mathbf{x}$ denotes a vector with $n$th entry $\mathbf{x}[n]$, and $\mathbf{X}$ denotes a matrix with $(n,m)$th entry $\mathbf{X}[n,m]$. $(\cdot)^{\ast}$ denotes complex conjugate, $(\cdot)^{\top}$ denotes transpose, $(\cdot)^{\mathsf{H}}$ denotes complex conjugate transpose and $\langle \mathbf{x}, \mathbf{y} \rangle = \sum_{n} \mathbf{x}[n] \mathbf{y}^{\ast}[n]$ denotes the inner product. Calligraphic font $\mathcal{X}$ denotes operators or sets, with usage clear from context. $\emptyset$ denotes the empty set. $\mathbb{Z}$ denotes the set of integers and $\mathbb{Z}_{N}$ the set of integers modulo $N$. $\lfloor \cdot \rfloor$ and $\lceil \cdot \rceil$ denote the floor and ceiling functions. $a \cdot b$ and $(a,b)$ respectively denote the bitwise dot product and greatest common divisor of two integers $a,b$. $(\cdot)_{{}_{N}}$ denotes the value modulo $N$ and $(\cdot)^{-1}_{{}_{N}}$ denotes the inverse modulo $N$. $\delta(\cdot)$ denotes the delta function, $\delta[\cdot]$ denotes the Kronecker delta function, $\mathds{1}{\{\cdot\}}$ denotes the indicator function, and $\mathbf{I}_{N}$ denotes the $N \times N$ identity matrix. $\mathbb{E}$ denotes expectation and $\equiv$ denotes congruence. 

\section{Mathematical Foundations}
\label{sec:foundations}



In this Section, we develop the mathematical foundations for communication over doubly-selective channels using signals with finite bandwidth and duration. In Section~\ref{subsec:foundations_identities}, we outline preliminaries that find use throughout the paper, such as certain useful definitions and identities. Section~\ref{subsec:foundation_sysmodel} develops the basic discrete time-domain system model for communicating over doubly-selective channels that is used in the remainder of the paper. Section~\ref{subsec:foundation_pred_sel} describes desirable properties of modulation schemes for reliable, high-rate communication over doubly-selective channels. Section~\ref{subsec:family_complex_hadamard} describes a general family of modulation schemes satisfying the properties outlined in Section~\ref{subsec:foundation_pred_sel}. In Section~\ref{subsec:foundation_hwgroup}, we provide a group-theoretic framework for analyzing a subset of the modulation schemes described in Section~\ref{subsec:family_complex_hadamard}. Section~\ref{sec:zak_intro} subsequently describes equivalent signal representations in the delay-Doppler domain, and overviews the Zak-OTFS (orthogonal time frequency space) modulation.


\subsection{Preliminaries}
\label{subsec:foundations_identities}


The following identities from classical number theory are used repeatedly throughout the paper.

\begin{identity}[\cite{iwaniec2021_numtheorybook}, pg.18]
    \label{idty:sumrootsofunity}
    The sum of $N$th roots of unity is:
    \begin{align*}
        \sum_{n \in \mathbb{Z}_{N}}e^{\frac{j2\pi}{N}kn} = \begin{cases}
        N \quad \text{if } \ k \equiv 0 \bmod{N} \\
        0 \quad \ \text{otherwise}
        \end{cases}.
    \end{align*}
\end{identity}

\begin{identity}[\cite{murty2017evaluation,iwaniec2021_numtheorybook}]
    \label{idty:gauss_sum}
    The quadratic Gauss sum is:
    \begin{align*}
        \sum_{n \in \mathbb{Z}_{N}} e^{\frac{j2\pi}{N}(an^2+bn)} = \epsilon_{N}\sqrt{N}\left(\frac{a}{N}\right)_J e^{-\frac{j2\pi}{N}(4a)^{-1}_{N} b^2},
    \end{align*}
    assuming odd $N$, $(a, N) = 1$, $a, b \in \mathbb{Z}_+$, and $\epsilon_{N}$ defined as:
    \begin{align*}
        \epsilon_{N} = \begin{cases}
            1~~~\text{if $N \equiv 1 \bmod{4}$} \\
            j~~~\text{if $N \equiv 3 \bmod{4}$}
        \end{cases}.
    \end{align*}
\end{identity}

An important definition used throughout the paper is that of the \emph{cross-ambiguity function}, which is a generalization of the inner product. Formally, for sequences in the complex-valued Hilbert space $\mathcal{H}$ of $MN$-periodic sequences, the cross-ambiguity is defined as follows.

\begin{definition}[\cite{benedetto_phasecoded}]
    \label{def:amb_fun}
    The \emph{cross-ambiguity function} of two unit-norm $MN$-periodic sequences $\mathbf{x},\mathbf{y} \in \mathcal{H}$ is defined as:
    \begin{align*}
        \mathbf{A}_{\mathbf{x},\mathbf{y}}[k,l] = \sum_{n \in \mathbb{Z}_{MN}} \mathbf{x}[n] \mathbf{y}^{*}[(n-k)_{{}_{MN}}]e^{-\frac{j2\pi}{MN}l(n-k)}, 
    \end{align*}
    where $k,l \in \mathbb{Z}_{MN}$, the integers modulo $MN$. When $\mathbf{y} = \mathbf{x}$, $\mathbf{A}_{\mathbf{x},\mathbf{x}}[k, l]$ is referred to as the \emph{self-ambiguity function} of $\mathbf{x}$, and is abbreviated to $\mathbf{A}_{\mathbf{x}}[k, l]$ for brevity.
\end{definition}

Moyal's Identity fixes the volume under the surface of the self-ambiguity function, making an ideal ``thumbtack'' ambiguity function impossible to achieve.

\begin{identity}[Moyal's Identity~\cite{Moyal1949,Auslander1985,Miller1991,Moran2001_mathofradar,Moran2004_grouptheory_radar,Howard2006_HW}]
    \label{idty:moyal}
    The self-ambiguities of unit-norm $MN$-periodic sequences $\mathbf{x},\mathbf{y} \in \mathcal{H}$ satisfy:
    \begin{align*}
        \frac{1}{MN} \sum_{k, l \in \mathbb{Z}_{MN}} \mathbf{A}_{\mathbf{x}}^{*}[k, l] \mathbf{A}_{\mathbf{y}}[k, l] = \big|\big\langle \mathbf{x}, \mathbf{y} \big\rangle\big|^{2}.
    \end{align*}

    A special case of the above relation corresponds to $\mathbf{y} = \mathbf{x}$:
    \begin{align*}
        \frac{1}{MN} \sum_{k, l \in \mathbb{Z}_{MN}} \big|\mathbf{A}_{\mathbf{x}}[k, l]\big|^{2} = 1.
    \end{align*}
\end{identity}

\subsection{Communicating over Doubly-Selective Channels}
\label{subsec:foundation_sysmodel}

Here, we develop a discrete time-domain system model for communication over doubly-selective channels using finite bandwidth and duration signals. In continuous time, the linear time varying system model for communication over a doubly-selective channel is given by~\cite{Bello1963_ltv,bitspaper1,bitspaper2,otfs_book}:
\begin{align}
    \label{eq:prelim1}
    y(t) &= \iint h(\tau,\nu) x(t-\tau) e^{j2\pi\nu(t-\tau)} d\tau d\nu + w(t),
\end{align}
where $x(t)$ (resp. $y(t)$) denotes the transmit (resp. receive) waveform in continuous time, $w(t)$ denotes the additive noise at the receiver, and $h(\tau,\nu)$ represents the channel spreading function in delay $\tau$ and Doppler $\nu$. 

We assume communication occurs over a finite bandwidth $B$ and time interval $T$, where we assume the time-bandwidth product $BT$ is an integer. In this paper, we assume the time-bandwidth product $BT$ can be expressed as the product of two integers, i.e., $BT = MN$, e.g., when the bandwidth is proportional to $M$, i.e., $B = M \Delta f$ for subcarrier spacing $\Delta f$, and the time interval is proportional to $N$, i.e., $T = \nicefrac{N}{\Delta f}$. 

Given the above assumptions, the discrete time version of~\eqref{eq:prelim1} corresponds to sampling the transmit and receive waveforms at integer multiples of the delay resolution $\nicefrac{1}{B}$ and limiting the waveforms to duration $T$~\cite{Mehrotra2025_EURASIP,Mattu2025_npj,ComparisonArxiv2025}:
\begin{align}
    \label{eq:prelim2}
    \mathbf{y}[n]\!&=\!\sum_{k,l \in \mathbb{Z}_{MN}}\!\mathbf{h}_{\mathsf{eff}}[k,l] \mathbf{x}[(n-k)_{{}_{MN}}] e^{\frac{j2\pi}{MN}l(n-k)}\!+\!\mathbf{w}[n],
\end{align}
where $n \in \mathbb{Z}_{MN}$ denotes the sampling index ($n = \lfloor Bt \rfloor$ for $0 \leq t \leq T$), $\mathbf{w}$ denotes noise samples, and $\mathbf{h}_{\mathsf{eff}}[k,l] = h\big(\nicefrac{k}{B},\nicefrac{l}{T}\big)$ denotes the effective channel spreading function sampled at integer multiples of the delay / Doppler resolution.

\textit{Note:} The channel spreading function $h(\tau,\nu)$ and its sampled version $\mathbf{h}_{\mathsf{eff}}[k,l]$ take into account both pulse shaping and the physical scattering environment. Pulse shaping limits the time and bandwidth of the waveform to $T$ and $B$ respectively (see Section~\ref{sec:filter_design} for more details). The physical scattering environment may comprise paths with \emph{fractional} delay and Doppler values as long as $\mathbf{h}_{\mathsf{eff}}[k,l]$ is sampled at multiples of the delay / Doppler resolution. For more details, see Section~\ref{sec:zak_intro}.

Observe that both the transmit and receive waveforms $\mathbf{x}$ and $\mathbf{y}$ in~\eqref{eq:prelim2} are $MN$-periodic sequences, with the channel acting on the Hilbert space $\mathcal{H}$ of $MN$-periodic sequences via scaled delay and Doppler shifts. Hence, $MN$ information symbols can be transmitted by mounting on an appropriate $MN$-dimensional orthonormal basis as:
\begin{align}
    \label{eq:prelim3}
    \mathbf{x}[n] &= \sum_{i \in \mathbb{Z}_{MN}} \mathbf{s}[i] \boldsymbol{\phi}_{i}[n],
\end{align}
where $\mathbf{s}$ denotes the $MN$-length vector of information symbols and $\boldsymbol{\phi}$ is an orthonormal basis with $MN$ elements, each of length $MN$. In Section~\ref{subsec:prelim_mod}, we present examples of modulation schemes that can be modeled using~\eqref{eq:prelim3}. One of the primary goals in communication engineering is to design a basis $\boldsymbol{\phi}$ in~\eqref{eq:prelim3} appropriately to enable \emph{reliable, high-rate} communication that \emph{maximizes spectral efficiency}. In Section~\ref{subsec:foundation_pred_sel}, we describe two key characteristics that a basis $\boldsymbol{\phi}$ must possess to achieve the above objectives. Section~\ref{subsec:family_complex_hadamard} defines a family of waveforms satisfying the characteristics described in Section~\ref{subsec:foundation_pred_sel}. Section~\ref{subsec:foundation_hwgroup} analyzes subsets of the waveform family from a group theoretic perspective.

\subsection{Examples of Modulation Schemes for Double Selectivity}
\label{subsec:prelim_mod}


\subsubsection{OFDM}
\label{subsubsec:prelim_ofdm}

The basis element in OFDM (orthogonal frequency division multiplexing) is~\cite{Ebert1971_ofdm,Bingham1990_ofdm}:
\begin{align}
    \label{eq:ofdm1}
    \boldsymbol{\phi}_{i}[n] = \frac{1}{\sqrt{M}} e^{\frac{j2\pi}{M}in} \mathds{1}\big\{\lfloor\nicefrac{n}{M}\rfloor = \lfloor\nicefrac{i}{M}\rfloor\big\},
\end{align}
where $i = 0, 1, \cdots, MN-1$.

\subsubsection{AFDM}
\label{subsubsec:prelim_afdm}

The basis element in AFDM (affine frequency division multiplexing) is~\cite{Bemani2023_afdm,Cho2025_dftpfdma,Hanzo2025_afdm_gen}:
\begin{align}
    \label{eq:afdm1}
    \boldsymbol{\phi}_{i}[n] = \frac{1}{\sqrt{MN}}e^{j2\pi\big(c_1n^2+c_2i^2+\frac{ni}{MN}\big)},
\end{align}
where $c_1, c_2 \in \mathbb{Z}$. The AFDM basis specializes to OCDM (orthogonal chirp division multiplexing)~\cite{Bemani2023_afdm} when $c_1 = c_2 = \nicefrac{1}{2MN}$ and to DFT-p-FDMA (discrete Fourier transform phase rotated \& permuted frequency division multiple access)~\cite{Cho2025_dftpfdma} when $c_1 = c_2 = \nicefrac{\Delta}{MN}$, where $(\Delta,MN) = 1$.

\subsubsection{ODDM}
\label{subsubsec:prelim_oddm}

The basis element in ODDM (orthogonal delay-Doppler division multiplexing) is~\cite{Yuan2022_oddm,Yuan2024_oddm_gen}:
\begin{align}
    \label{eq:oddm1}
    \boldsymbol{\phi}_{i}[n] = \frac{1}{\sqrt{N}} e^{\frac{j2\pi}{N}{\lfloor\nicefrac{i}{M}\rfloor} \lfloor\nicefrac{n}{M}\rfloor} \mathds{1}\big\{n \equiv i \bmod{M}\big\}.
\end{align}

\subsubsection{OTSM}
\label{subsubsec:prelim_otsm}

The basis element in OTSM (orthogonal time sequency division multiplexing) is~\cite{Viterbo2021_otsm,Hanzo2024_otsm_amp}:
\begin{align}
    \label{eq:otsm1}
    \boldsymbol{\phi}_{i}[n] = \frac{1}{\sqrt{N}} (-1)^{\lfloor\nicefrac{i}{M}\rfloor \odot \lfloor\nicefrac{n}{M}\rfloor} \mathds{1}\big\{n \equiv i \bmod{M}\big\},
\end{align}
where $\odot$ denotes the bitwise dot product.

\subsubsection{Zak-OTFS}
\label{subsubsec:prelim_zak}

The basis element in Zak-OTFS (Zak-transform based orthogonal time frequency space modulation) is~\cite{bitspaper1,bitspaper2,otfs_book}:
\begin{align}
    \label{eq:zakotfs1}
    \boldsymbol{\phi}_{i}[n] &= \frac{1}{\sqrt{N}} \sum_{d \in \mathbb{Z}} e^{j\frac{2\pi}{N} d \lfloor \nicefrac{i}{M} \rfloor} \delta[n-(i)_{{}_{M}}-dM] \nonumber \\
    &\textcolor{black}{= \frac{1}{\sqrt{N}} e^{\frac{j2\pi}{N}{\lfloor\nicefrac{i}{M}\rfloor} \lfloor\nicefrac{n}{M}\rfloor} \mathds{1}\big\{n \equiv i \bmod{M}\big\},}
\end{align}
termed \emph{pulsone} due to its structure of a pulse train modulated by a tone, \textcolor{black}{and coincides with the ODDM basis element in~\eqref{eq:oddm1}.}

\subsection{Choosing an Appropriate Basis $\boldsymbol{\phi}$}
\label{subsec:foundation_pred_sel}

We identify two key characteristics -- \emph{predictability} and \emph{non-selectivity} -- that a modulation basis (we refer to the basis elements outlined in Sec.~\ref{subsec:prelim_mod} as modulation basis) must possess to achieve reliability and maximize spectral efficiency. Predictability corresponds to whether the response of the channel to a modulation carrier can be predicted from the response to a different carrier. Predictable modulations have small pilot overhead since only a single pilot transmission suffices to estimate the channel for the full frame. Predictable modulations are also non-selective, i.e., have no variation in received signal energy across modulation carriers. 

In order to formalize the above notions, we first express the system model in~\eqref{eq:prelim2} in terms of the basis $\boldsymbol{\phi}$ as~\cite{Mattu2025_npj,ComparisonArxiv2025}:
\begin{align}
    \label{eq:prelim4}
    \mathbf{r}[f] = &\sum_{n \in \mathbb{Z}_{MN}} \boldsymbol{\phi}_{f}^{*}[n] \mathbf{y}[n] \nonumber \\
    = &\sum_{i \in \mathbb{Z}_{MN}} \mathbf{s}[i] \bigg(\sum_{k,l \in \mathbb{Z}_{MN}} e^{-\frac{j2\pi}{MN}lk} \mathbf{h}_{\mathsf{eff}}[k,l] \nonumber \\ \times &\sum_{n \in \mathbb{Z}_{MN}} \boldsymbol{\phi}_{f}^{*}[n] \boldsymbol{\phi}_{i}[(n-k)_{{}_{MN}}] e^{\frac{j2\pi}{MN}ln}\bigg) + \mathbf{v}[f] \nonumber \\
    = &\sum_{i \in \mathbb{Z}_{MN}} \mathbf{s}[i] \mathbf{H}[f,i] + \mathbf{v}[f],
\end{align}
where $\mathbf{v}[f] = \sum_{n \in \mathbb{Z}_{MN}} \boldsymbol{\phi}_{f}^{*}[n] \mathbf{w}[n]$ denotes the projection of the noise samples on the basis $\boldsymbol{\phi}$. On vectorizing~\eqref{eq:prelim4}:
\begin{align}
    \label{eq:prelim5}
    \mathbf{r} &= \mathbf{H} \mathbf{s} + \mathbf{v},
\end{align}
where $\mathbf{H}$ denotes the equivalent $MN \times MN$ channel matrix.

Recovering the transmitted information symbols $\mathbf{s}$ requires knowledge of the channel matrix $\mathbf{H}$, or equivalently, the sampled channel spreading function $\mathbf{h}_{\mathsf{eff}}[k,l]$. The sampled channel spreading function can be estimated by transmitting a known pilot symbol and computing the cross-ambiguity (c.f. Definition~\ref{def:amb_fun}) between the received and transmitted waveforms: 
\begin{align}
    \label{eq:prelim6}
    \widehat{\mathbf{h}}_{\mathsf{eff}}[k,l] &= \mathbf{A}_{\mathbf{y},\mathbf{x}}[k,l],
\end{align}
which has been shown to be the maximum likelihood estimate in~\cite{Calderbank2025_isac}. Subsequently, the matrix $\mathbf{H}$ is estimated using~\eqref{eq:prelim4} and used to recover the information symbols $\mathbf{s}$, e.g., via the minimum mean squared error (MMSE) estimator~\cite{otfs_book}.

We are now ready to formalize the notion of predictability.

\begin{definition}[\cite{ComparisonArxiv2025}]
    \label{def:predict}
    A modulation is predictable if all its constituent modulation carriers estimate the same channel spreading function via the cross-ambiguity operation in~\eqref{eq:prelim6}:
    \begin{align*}
        \widehat{\mathbf{h}}_{\mathsf{eff}}[k,l] = \mathbf{A}_{\mathbf{y}_{i},\boldsymbol{\phi}_{i}}[k,l] = \mathbf{A}_{\mathbf{y}_{i},\boldsymbol{\phi}_{j}}[k,l],
    \end{align*}
    for all $i,j \in \mathbb{Z}_{MN}$, where $\mathbf{y}_{i}$ denotes the received sequence when information is modulated only on carrier $i$. 
\end{definition}

We term the above condition predictability since it implies that the channel response to any basis element $i\in\mathbb{Z}_{MN}$ can be exactly predicted from the (noiseless) measurements obtained for a given basis element. In other words, the estimate of a channel obtained from a pilot is independent of where the pilot was transmitted in the frame. The following Lemma from~\cite{ComparisonArxiv2025} derives the condition on the channel spreading function under which a modulation basis $\boldsymbol{\phi}$ is predictable.

\begin{lemma}[\cite{ComparisonArxiv2025}]
    \label{lmm:pred}
    Let $\mathcal{S} = \big\{(k,l):|\mathbf{h}[k,l]|\neq0\big\}$ denote the channel support corresponding to spreading function $\mathbf{h}[k,l]$. Define $\mathcal{K}_{\mathcal{S}} = \big\{(k',l'):\mathcal{S}\cap\big(\mathcal{S}+(k',l')\big)\neq\emptyset\big\}$ as the set of DD locations where translates of the channel support overlap with $\mathcal{S}$. Then, channel estimates via~\eqref{eq:prelim6} always equal $\mathbf{h}[k,l]$ for all $(k,l)\in\mathcal{S},i\in\mathbb{Z}_{MN}$ in a modulation basis with:
    \begin{align*}
        \mathbf{A}_{\boldsymbol{\phi}_{i}}[0,0] = 1,\mathbf{A}_{\boldsymbol{\phi}_{i}}[k',l'] = 0~\text{for}~(k',l')\in\mathcal{K}_{\mathcal{S}}\setminus\{(0,0)\},
    \end{align*}
    where $\mathbf{A}_{\mathbf{x}}[k,l]$ denotes the self-ambiguity of $\mathbf{x}$.
\end{lemma}

Fig.~\ref{fig:cryst} geometrically illustrates the condition in Lemma~\ref{lmm:pred} for a simple example of a rectangular support $\mathcal{S}$.

\begin{figure}
    \centering
    \includegraphics[width=0.92\linewidth]{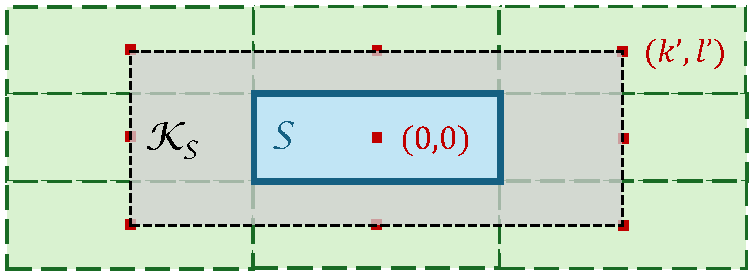}
    \caption{\textcolor{black}{For predictability (Lemma~\ref{lmm:pred}) with a rectangular channel support $\mathcal{S}$, the self-ambiguity $\mathbf{A}_{\boldsymbol{\phi}_{i}}[k',l']$ should have only $1$ non-zero value at $(k',l')=(0,0)$ within the region $\mathcal{K}_{\mathcal{S}}$ (gray).}}
    \label{fig:cryst}
\end{figure}

A consequence of predictability is non-selectivity.

\begin{corollary}[\cite{ComparisonArxiv2025}]
    \label{corr:nonsel}
    When predictability as per Lemma~\ref{lmm:pred} holds, then the constituent carriers in the modulation basis do not undergo symbol fading\footnote{\textcolor{black}{Assuming $\mathbf{s} = c \cdot \mathbf{e}_{i}$ for unit-energy symbol $c$ on standard basis vector / carrier $i$ and calculating $\mathbf{r}^{\mathsf{H}} \mathbf{r} = \mathbf{s}^{\mathsf{H}} \mathbf{G}^{\mathsf{H}}\mathbf{G} \mathbf{s} = |c|^{2} \mathbf{e}_{i}^{\mathsf{H}} \mathbf{G}^{\mathsf{H}}\mathbf{G} \mathbf{e}_{i} = (\mathbf{G}^{\mathsf{H}}\mathbf{G})[i,i]$.}}, i.e., for all $i,j\in\mathbb{Z}_{MN}$:
    \begin{align*}
        (\mathbf{H}^{\mathsf{H}}\mathbf{H})[i,i] = (\mathbf{H}^{\mathsf{H}}\mathbf{H})[j,j] = \sum_{k,l\in\mathcal{S}}|\mathbf{h}[k,l]|^{2}.
    \end{align*}
\end{corollary}

\subsection{A Family of Predictable / Non-Selective Modulations}
\label{subsec:family_complex_hadamard}

We now describe a family of waveforms that are predictable and non-selective as per Lemma~\ref{lmm:pred} / Corollary~\ref{corr:nonsel}. Consider the following family of $MN$-periodic waveforms:
\begin{align}
    \label{eq:sys1}
    \Psi = \bigg\{\mathbf{U} \boldsymbol{\phi}_{i}:~&\boldsymbol{\phi}_{i}[n] = \frac{1}{\sqrt{N}} \mathbf{G}_{(i)_{{}_{M}}}[\lfloor \nicefrac{n}{M} \rfloor, \lfloor \nicefrac{i}{M} \rfloor] \nonumber \\ &\times \mathds{1}\big\{n \equiv i \bmod{M}\big\}, i,n \in \mathbb{Z}_{MN} \bigg\},
\end{align}
where $\mathbf{U}$ is an arbitrary $MN \times MN$ unitary matrix, and $\mathbf{G}_{x}$ denotes an $N \times N$ \emph{complex Hadamard matrix} parameterized by index $x \in \mathbb{Z}_{M}$ and with $(p,q)$th entry $\mathbf{G}_{x}[p,q]$.

\begin{definition}[\cite{Butson1962_hadamard,Tadej2006_hadamard,Banica2021_hadamard1,Banica2024_hadamard2}]
    \label{def:complex_Hadamard}
    An $N \times N$ complex Hadamard matrix $\mathbf{G}$ is a square matrix with unimodular entries: $|\mathbf{G}[p,q]| = 1$ for all $p,q \in \mathbb{Z}_{N}$, and pairwise orthogonal rows \& columns: $\mathbf{G}^{\mathsf{H}} \mathbf{G} = \mathbf{G} \mathbf{G}^{\mathsf{H}} = N \mathbf{I}_{N}$.
\end{definition}

Examples of complex Hadamard matrices include the discrete Fourier transform (DFT) matrix and its inverse (IDFT), the Walsh matrix, etc. Different choices of the unitary matrix $\mathbf{U}$ and the complex Hadamard matrix $\mathbf{G}_{(i)_{{}_{M}}}$ in~\eqref{eq:sys1} result in different modulations (except OFDM) from Section~\ref{subsec:prelim_mod}.

\subsubsection{Zak-OTFS \& ODDM}
\label{subsubsec:family_zak}

Choose $\mathbf{U} = \mathbf{I}_{MN}$ and $\mathbf{G}_{(i)_{{}_{M}}}$ as the IDFT matrix with entry $\mathbf{G}_{(i)_{{}_{M}}}[p,q] = e^{j\frac{2\pi}{N}pq}$.

\subsubsection{OTSM}
\label{subsubsec:family_otsm}

Choose $\mathbf{U} = \mathbf{I}_{MN}$ and $\mathbf{G}_{(i)_{{}_{M}}}$ as the Walsh matrix with entry $\mathbf{G}_{(i)_{{}_{M}}}[p, q] = (-1)^{p \odot q}$.

\subsubsection{AFDM}
\label{subsubsec:family_afdm}

Choose $\mathbf{U}$ as the generalized discrete affine Fourier transform (GDAFT) matrix from~\cite{Mehrotra2025_WCLSpread,Mehrotra2025_EURASIP} and $\mathbf{G}_{(i)_{{}_{M}}}$ as the IDFT matrix with entry $\mathbf{G}_{(i)_{{}_{M}}}[p,q] = e^{j\frac{2\pi}{N}pq}$.

All modulations in $\Psi$ are predictable\footnote{MC-OTFS is not a member of $\Psi$ and therefore is not predictable~\cite{ComparisonArxiv2025}.} and non-selective\footnote{Here by non-selectivity, we mean that each information symbol in the transmit grid sees the same fade, irrespective of its location.}, as shown in the Theorem below. Note that for simplicity, we present the result for the case $\mathbf{U} = \mathbf{I}_{MN}$, and analogous results can be similarly derived for specific choices of $\mathbf{U}$, e.g., see~\cite{Mehrotra2025_WCLSpread,Mehrotra2025_EURASIP} for an example with $\mathbf{U}$ as the GDAFT matrix.

\begin{theorem}[\cite{ComparisonArxiv2025}]
    \label{thm:family_pred}
    For $\mathbf{U} = \mathbf{I}_{MN}$, all modulation schemes in $\Psi$ are predictable and non-selective when the channel support $\mathcal{S}$ satisfies $\mathcal{S}\cap\bigg(\bigcup_{(a,b)\neq(0,0)}\big(\mathcal{S}+(aM,bN)\big)\bigg)=\emptyset$.
\end{theorem}

In other words, Theorem~\ref{thm:family_pred} implies that the channel response for each modulation scheme in $\Psi$ is contained within the $M\times N$ frame. The (quasi-) periodic repetitions of the channel support do not interfere with one another. This is also referred to as the crystalline condition in Zak-OTFS literature~\cite{bitspaper1, bitspaper2} (see Sec.~\ref{subsec:crystallization}).

One of the advantages of modulations in $\Psi$ that are predictable and non-selective is that they achieve full diversity~\cite{ComparisonArxiv2025}. This is because every symbol in the transmit frame sees the same channel, which is a combination of all the paths in the channel.

\subsection{A Group Theoretic Framework to Analyze Subsets of $\Psi$}
\label{subsec:foundation_hwgroup}

We now describe a group-theoretic framework to explore special subsets of the waveform family $\Psi$ in~\eqref{eq:sys1}, whose self-ambiguity functions satisfy Lemma~\ref{lmm:pred} in addition to being \emph{unimodular} on exactly $MN$ delay-Doppler locations\footnote{The theory presented here is utilized in Sec.~\ref{sec:radar} for generating radar waveforms with good self-ambiguity properties.}.

Let $\mathcal{H}$ denote the Hilbert space of $MN$-periodic sequences. Following~\eqref{eq:prelim2}, we may represent the channel action for a given $k,l \in \mathbb{Z}_{MN}$ (delay / time and Doppler / frequency shift) on the transmitted sequence via a unitary operator acting on $\mathcal{H}$.

\begin{definition}[\cite{Mehrotra2025_EURASIP}]
    \label{def:heis_op}
    \emph{Heisenberg operators} $\mathcal{D}_{(k,l)}$ act on sequences ${\mathbf{x} \in \mathcal{H}}$ as:
    \begin{equation*}
        \begin{gathered}
        \mathcal{D}_{(k,l)}: \mathcal{H} \rightarrow \mathcal{H}, k,l \in \mathbb{Z}_{MN}, \\
        (\mathcal{D}_{(k,l)} \mathbf{x})[n] = \mathbf{x}[(n-k)_{{}_{MN}}]e^{\frac{j2\pi}{MN}l(n-k)},
        \end{gathered}
    \end{equation*}
    for all $n \in \mathbb{Z}_{MN}$. 
\end{definition}

Heisenberg operators are unitary, and the collection of all (phase-scaled) Heisenberg operators forms a \emph{group}.

\begin{figure*}[t]
\centering
\begin{subfigure}{0.32\linewidth}
    \includegraphics[width=\textwidth]{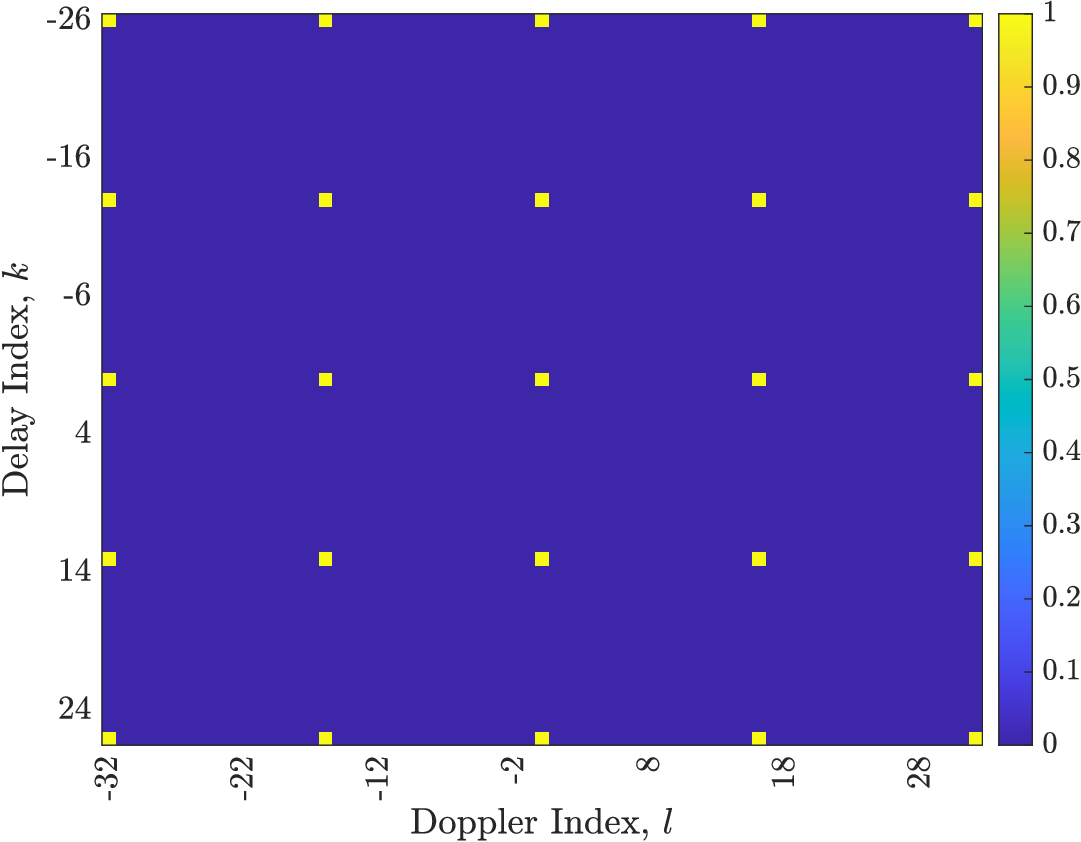}
\caption{Zak-OTFS pulsone (Example~\ref{ex:comm_subgrp_ex1})}
    \label{fig:wvf_lib1}
\end{subfigure}
\begin{subfigure}{0.32\linewidth}
    \includegraphics[width=\textwidth]{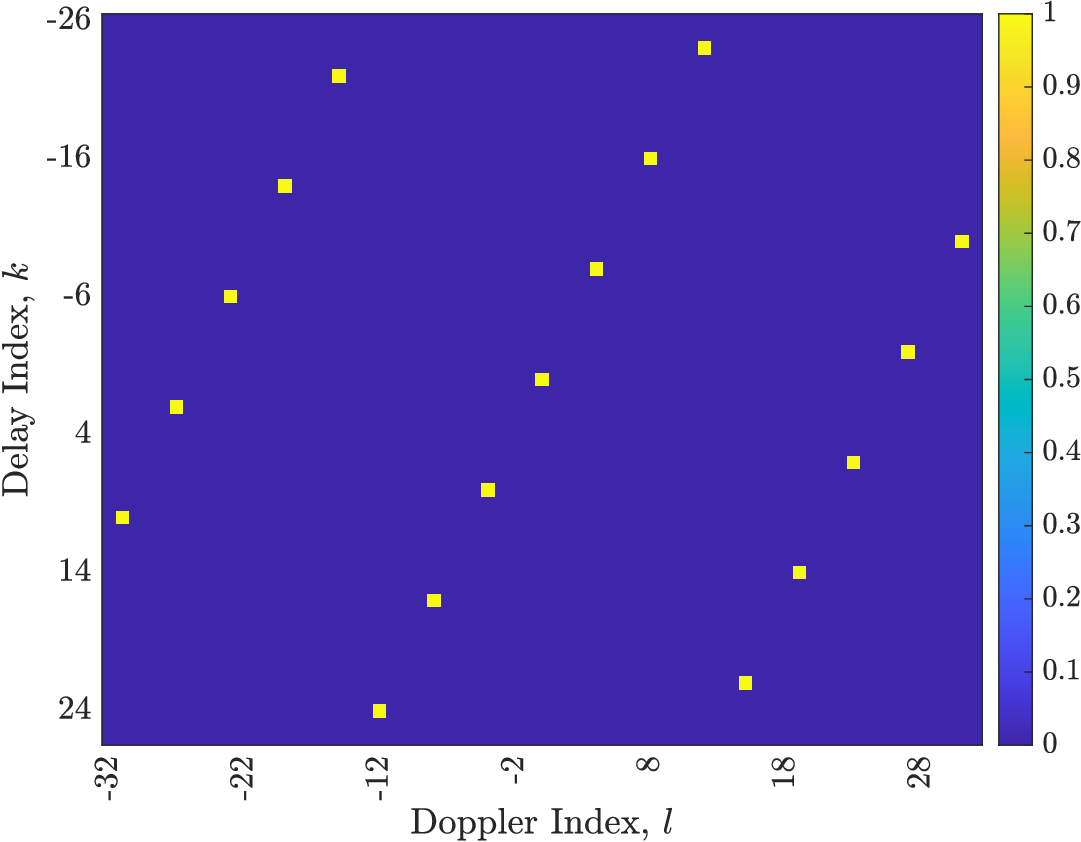}
\caption{AFDM chirp (Example~\ref{ex:comm_subgrp_ex2})}
    \label{fig:wvf_lib2}
\end{subfigure}
\begin{subfigure}{0.32\linewidth}
    \includegraphics[width=\textwidth]{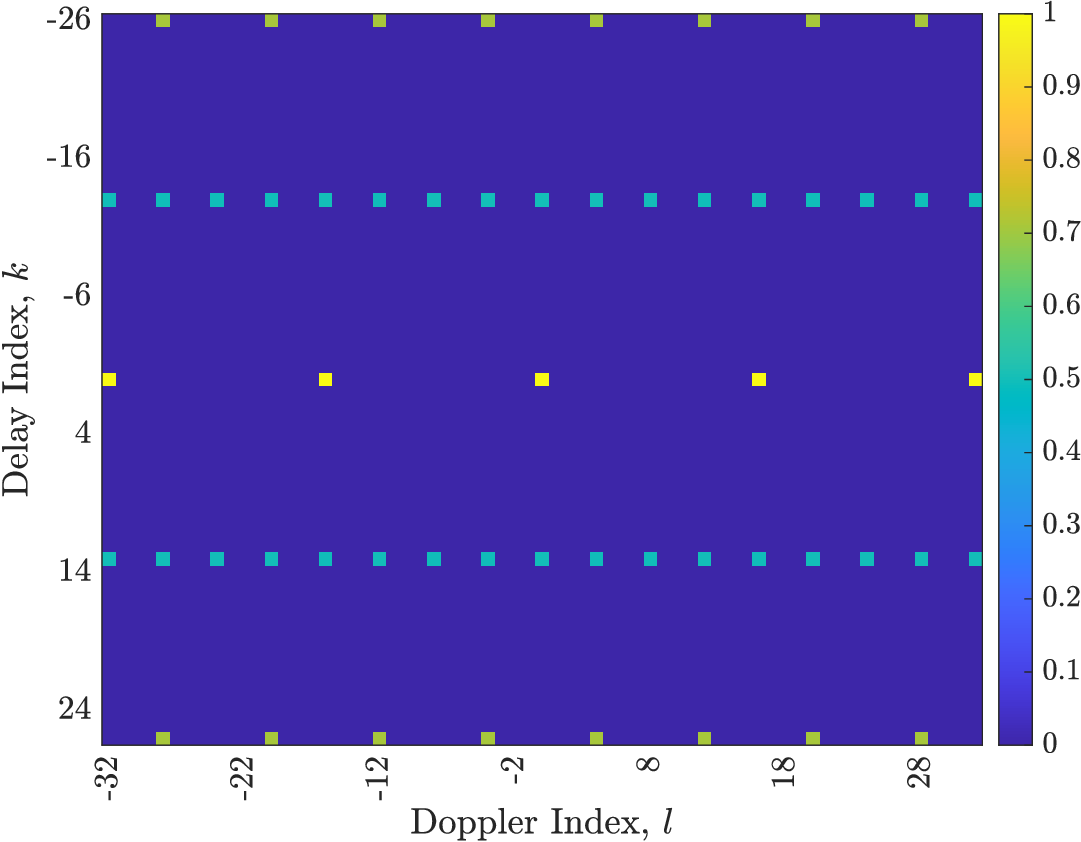}
\caption{OTSM basis element}
    \label{fig:wvf_lib3}
\end{subfigure}
\caption{Illustration of self-ambiguity function magnitudes corresponding to different modulation schemes in the waveform family $\Psi$ for parameters $M = 13$, $N = 16$. The Zak-OTFS pulsone in (a) and the AFDM chirp in (b) are eigenvectors of maximal commutative subgroups (Examples~\ref{ex:comm_subgrp_ex1} and~\ref{ex:comm_subgrp_ex2}), whereas the OTSM basis element in (c) does not belong to a maximal commutative subgroup. All three modulations satisfy the properties in Lemma~\ref{lmm:pred} / Corollary~\ref{corr:nonsel} and are predictable / non-selective.}
\vspace{-5mm}
    \label{fig:wvf_lib}
\end{figure*}

\begin{theorem}[\cite{Mehrotra2025_EURASIP}]
    \label{thm:hw_group}
    Let $\mathcal{H}_{MN}$ denote the collection of phase-scaled Heisenberg operators $\mathcal{D}_{(k,l)}$:
    \begin{align*}
        \mathcal{H}_{MN} &= \big\{e^{\frac{j2\pi}{MN}m} \mathcal{D}_{(k,l)} \big| k,l,m \in \mathbb{Z}_{MN}\big\}. 
    \end{align*}

    The set $\mathcal{H}_{MN}$ under operator composition $\circ$ forms the \emph{Heisenberg-Weyl group}.
\end{theorem}

Let us now understand the action of the Heisenberg-Weyl group on an $MN$-length basis element $\boldsymbol{\phi}_{i}$. Let $\mathcal{V}_{i} \subset \mathcal{H}$ denote the $1$-dimensional subspace comprising phase scalings of $\boldsymbol{\phi}_{i}$:
\begin{align}
    \label{eq:stab1}
    \mathcal{V}_{i} = \big\{e^{\frac{j2\pi}{MN}m} \boldsymbol{\phi}_{i} \big| m \in \mathbb{Z}_{MN} \big\}.
\end{align}

We compute the \emph{stabilizer} of $\mathcal{V}_{i}$, which is defined as follows.

\begin{definition}
    \label{def:stab}
    Given a subspace $\mathcal{V} \subset \mathcal{H}$ and a group $\mathcal{G}$, the stabilizer of $\mathcal{V}$ is defined as:
    \begin{align*}
        \mathcal{G}_{\mathcal{V}} = \big\{g \in \mathcal{G} \big| (g\mathbf{v}) \in \mathcal{V},~\text{for all}~\mathbf{v} \in \mathcal{V} \big\}.
    \end{align*}
\end{definition}

Applying Definition~\ref{def:stab} to~\eqref{eq:stab1} and the Heisenberg-Weyl group, the stabilizer of $\mathcal{V}_{i}$ is the subset of the Heisenberg-Weyl group that acts on basis elements $\boldsymbol{\phi}_{i}$ via phase scaling: 
\begin{align}
    \label{eq:stab2}
    \mathcal{G}_{\mathcal{V}_{i}} &= \big\{e^{\frac{j2\pi}{MN}m} \mathcal{D}_{(k,l)} \in \mathcal{H}_{MN} \big| (\mathcal{D}_{(k,l)}\boldsymbol{\phi}_{i}) = e^{\frac{j2\pi}{MN}m'} \boldsymbol{\phi}_{i} \big\}.
\end{align}

It is well-known from~\cite{Mehrotra2025_EURASIP,Howard2006_HW} that such subsets of the Heisenberg-Weyl group form \emph{maximal commutative subgroups} (a.k.a. isotropy subgroups), i.e., commutative subgroups of order $MN$, with the property that no group element outside of the subgroup commutes with the subgroup. Maximal commutative subgroups act on their \emph{eigenvectors} via phase scaling. 


\begin{definition}[\cite{Mehrotra2025_EURASIP}]
    \label{def:max_comm_subgrp}
    A subgroup $\mathcal{T}_{MN} \subset \mathcal{H}_{MN}$ of the Heisenberg-Weyl group is \emph{commutative} if any two elements $t_1 \in \mathcal{T}_{MN}$ and $t_2 \in \mathcal{T}_{MN}$ commute, i.e., $(t_1 \circ t_2) = (t_2 \circ t_1)$. A commutative subgroup is \emph{maximal} if and only if no $h \in \mathcal{H}_{MN}, h \not\in \mathcal{T}_{MN}$ exists such that $h$ commutes with $\mathcal{T}_{MN}$. 
\end{definition}


\begin{lemma}[\cite{Mehrotra2025_EURASIP}]
    \label{lmm:eigenbasis}
    For every maximal commutative subgroup $\mathcal{T}_{MN} \subset \mathcal{H}_{MN}$, there exists a simultaneous orthonormal eigenbasis $\big\{\mathbf{v}_{i}\big\}_{i=1}^{MN}$ for all elements $t \in \mathcal{T}_{MN}$, with group elements $h \not\in \mathcal{T}_{MN}$ not in the subgroup acting orthogonally on the eigenvectors:
    \begin{gather*}
        t \mathbf{v}_{i} = \lambda_{t,i} \mathbf{v}_{i},~t \in \mathcal{T}_{MN}, |\lambda_{t,i}| = 1, \\
        \big\langle \mathbf{v}_{i},h \mathbf{v}_{i} \big\rangle = 0,~h \not\in \mathcal{T}_{MN}, h \in \mathcal{H}_{MN}.
    \end{gather*}
\end{lemma}

The operators belonging to a maximal commutative subgroup are \emph{simultaneously diagonalizable}. Maximal commutative subgroups are also \emph{normal subgroups} of the Heisenberg-Weyl group, hence have special ambiguity function properties.

\begin{corollary}[\cite{Mehrotra2025_EURASIP}]
    \label{corr:max_comm_ambg}
    Eigenvectors $\big\{\mathbf{v}_{i}\big\}_{i=1}^{MN}$ of maximal commutative subgroups $\mathcal{T}_{MN} \subset \mathcal{H}_{MN}$ have self-ambiguity:
    \begin{align*}
        \mathbf{A}_{\mathbf{v}_{i}}[k, l] = \begin{cases}
            e^{\frac{j2\pi}{MN}m} \lambda_{t,i}^{*}, & t = e^{\frac{j2\pi}{MN}m} \mathcal{D}_{(k,l)} \in \mathcal{T}_{MN}, \\ 
            0, & h = e^{\frac{j2\pi}{MN}m} \mathcal{D}_{(k,l)} \not\in \mathcal{T}_{MN},
        \end{cases}
    \end{align*}
    i.e., the self-ambiguity function is unimodular over exactly $MN$ delay-Doppler indices $(k,l)$, $k \in \mathbb{Z}_{MN}$, $l \in \mathbb{Z}_{MN}$.
\end{corollary}

It is clear that Corollary~\ref{corr:max_comm_ambg} corresponds to a special case of Lemma~\ref{lmm:pred}, where the self-ambiguity function is unimodular over exactly $MN$ delay-Doppler indices. Since Lemma~\ref{lmm:pred} is satisfied, eigenvectors of maximal commutative subgroups are predictable / non-selective, hence span a subset of $\Psi$. Below, we present two examples of maximal commutative subgroups.


\begin{example}[\cite{Mehrotra2025_EURASIP}]
    \label{ex:comm_subgrp_ex1}
    The maximal commutative subgroup: 
    \begin{align*}
        \mathcal{T}_{MN} = \big\{e^{\frac{j2\pi}{MN}m} \mathcal{D}_{(aM,bN)} \big| m \in \mathbb{Z}_{MN}, a \in \mathbb{Z}_{{}_{N}}, b \in \mathbb{Z}_{{}_{M}}\big\},
    \end{align*}
    represents a \emph{rectangular grid} spaced at multiples of $M$ along delay and multiples of $N$ along Doppler. The corresponding eigenvectors are the Zak-OTFS pulsones from~\eqref{eq:zakotfs1}:
    \begin{align*}
        \mathbf{v}_{i}[n] = \frac{1}{\sqrt{N}} \sum_{d \in \mathbb{Z}} e^{\frac{j2\pi}{N} d l_0} \delta[n-k_0-dM],~i=k_0+l_0M,
    \end{align*}
    with eigenvalues $\lambda_{t,i} = e^{\frac{j2\pi}{MN}m} e^{\frac{j2\pi}{M}bk_0} e^{-\frac{j2\pi}{N} al_0}$.
\end{example}


Fig.~\ref{fig:wvf_lib}(\subref{fig:wvf_lib1}) illustrates how the self-ambiguity of a Zak-OTFS pulsone is supported on a rectangular grid as per Example~\ref{ex:comm_subgrp_ex1}.

\begin{example}[\cite{Mehrotra2025_EURASIP}]
    \label{ex:comm_subgrp_ex2}
    The maximal commutative subgroup: 
    \begin{align*}
        \mathcal{T}_{MN} = \big\{e^{\frac{j2\pi}{MN}m} \mathcal{D}_{(k,l)} \big|&2\alpha k - l \equiv 0 \bmod{MN}, \nonumber \\ &(\alpha,MN) = 1, k,l,m \in \mathbb{Z}_{MN} \big\},
    \end{align*}
    represents a \emph{line with slope $2 \alpha$} in the delay-Doppler grid. The corresponding eigenvectors are the AFDM chirps from~\eqref{eq:afdm1}:
    \begin{align*}
        \mathbf{v}_{i}[n] = \frac{1}{\sqrt{MN}} e^{\frac{j2\pi}{MN} [\alpha n^2 + \beta n + \gamma]},~i = f(\alpha,\beta,\gamma),
    \end{align*}
    with eigenvalues $\lambda_{t,i} = e^{\frac{j2\pi}{MN}m} e^{-\frac{j2\pi}{MN}(\alpha k^2 + \beta k)}$, where $f$ is an affine function.
\end{example}


Fig.~\ref{fig:wvf_lib}(\subref{fig:wvf_lib2}) illustrates how the self-ambiguity of an AFDM chirp is supported on a line as per Example~\ref{ex:comm_subgrp_ex2}. 

We also plot the self-ambiguity of an OTSM basis element in Fig.~\ref{fig:wvf_lib}(\subref{fig:wvf_lib3}), which is a member of the waveform family $\Psi$, however does not constitute a maximal commutative subgroup.



We now define a special class of unitary operators called \emph{symplectic transformations} that act by conjugation on the Heisenberg-Weyl group, defining automorphisms that preserve the group structure and map between commutative subgroups.

\begin{definition}[\cite{Mehrotra2025_EURASIP}]
    \label{def:sympl_transform}
    A \emph{symplectic transformation} of the Heisenberg-Weyl group is a unitary operator $\mathcal{W}(g): \mathcal{H} \rightarrow \mathcal{H}$ associated with an element $g \in \mathcal{G}$ of the special linear group of order $2$ over $\mathbb{Z}_{MN}$:
    \begin{align*}
        \mathcal{G} &= \bigg\{\begin{bmatrix}
            a & b \\ c & d
        \end{bmatrix} \bigg| a, b, c, d \in \mathbb{Z}_{MN}, ad - bc = 1 \bigg\},
    \end{align*}
    that acts on the Heisenberg-Weyl group $\mathcal{H}_{MN}$ via conjugation:
    \begin{gather*}
        \mathcal{W}(g \cdot h) = \mathcal{W}(g) \circ \mathcal{W}(h), \\
        \mathcal{W}(g) \circ \big(e^{\frac{j2\pi}{MN}m}\mathcal{D}_{(k,l)}\big) \circ \mathcal{W}^{-1}(g) = e^{\frac{j2\pi}{MN}m'}\mathcal{D}_{g \cdot (k,l)},
    \end{gather*}
    for all $g,h \in \mathcal{G}$, and $e^{\frac{j2\pi}{MN}m}\mathcal{D}_{(k,l)}, e^{\frac{j2\pi}{MN}m'}\mathcal{D}_{g \cdot (k,l)} \in \mathcal{H}_{MN}$.
\end{definition}

Symplectic transformations satisfy the following properties.

\begin{lemma}[\cite{Mehrotra2025_EURASIP}]
    \label{lmm:weil_prop}
    Symplectic transformations map commutative subgroups to commutative subgroups and rotate the cross-ambiguity function of waveforms:
    \begin{align*}
        \big|\mathbf{A}_{\mathbf{x},\mathbf{y}}[k,l]\big| &= \big|\mathbf{A}_{\mathcal{W}(g) \mathbf{x},\mathcal{W}(g) \mathbf{y}}[g \cdot (k,l)]\big|.
    \end{align*}
\end{lemma}

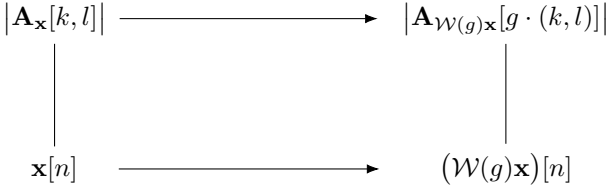
\begin{figure}
  \centering
  \begin{tikzpicture}[>=Latex]
    \node (eqA) at (0, 2) {$\big|\mathbf{A}_{\mathbf{x}}[k,l]\big|$};
    \node (eqB) at (6, 2) {$\big|\mathbf{A}_{\mathcal{W}(g) \mathbf{x}}[g \cdot (k,l)]\big|$};
    \node (eqC) at (0, 0) {$\mathbf{x}[n]$};
    \node (eqD) at (6, 0) {$\big(\mathcal{W}(g) \mathbf{x}\big)[n]$};

    \coordinate (Aleft) at ($(eqA.east)+(0.05,0)$);
    \coordinate (Bleft) at ($(eqB.west)-(0.2,0)$);
    \coordinate (Cleft) at ($(eqC.east)+(0.42,0)$);
    \coordinate (Dleft) at ($(eqD.west)-(0.67,0)$);

    \draw[->] (Aleft) -- (Bleft);
    \draw[->] (Cleft) -- (Dleft);

    \draw[-] ($(eqA.south)+(0,0)$) -- ($(eqC.north)+(0,0)$);
    \draw[-] ($(eqB.south)-(0,0)$) -- ($(eqD.north)-(0,0)$);
  \end{tikzpicture}
  \caption{The action of symplectic transformations on waveforms and their ambiguity functions.}
  \label{fig:weil_ambg_fun}
\end{figure}

Fig.~\ref{fig:weil_ambg_fun} illustrates Lemma~\ref{lmm:weil_prop}. The ability to rotate ambiguity functions via symplectic transformations aids the construction of \emph{radar waveform libraries}; see Section~\ref{sec:radar} for more details.

We now present an example of a symplectic transformation.

\begin{definition}[\cite{Mehrotra2025_WCLSpread,Mehrotra2025_EURASIP}]
    \label{def:gdaft}
    The \emph{generalized discrete affine Fourier transform} (GDAFT) of an $MN$-length sequence $\mathbf{x}$ is:
    \begin{align*}
        \mathcal{F}_a\mathbf{x}[n] = \frac{1}{\sqrt{MN}} \sum_{m \in \mathbb{Z}_{MN}} e^{j\frac{\pi}{MN} b^{-1}_{{}_{MN}} (d n^2 - 2 nm + a m^2)} \mathbf{x}[m],
    \end{align*}
    where $n\in\mathbb{Z}_{MN}$, $b$ is co-prime to $MN$, $(b,MN) = 1$, with parameters $a,b,c,d$ defined as in Definition~\ref{def:sympl_transform}.
\end{definition}

The GDAFT rotates / chirps its input waveform.

\begin{lemma}[\cite{Mehrotra2025_EURASIP}]
    \label{lmm:heis_op_gdaft}
    Let $\tilde{\mathbf{x}}[n] = (\mathcal{F}_a \mathbf{x})[n]$. Conjugation by the GDAFT as in Definition~\ref{def:sympl_transform} rotates / chirps waveforms:
    \begin{align*}
        \big(\mathcal{F}_a \circ \mathcal{D}_{(k,l)} \circ \mathcal{F}_a^{-1}\tilde{\mathbf{x}}\big)[n] &= e^{j\frac{\pi}{MN} [ack^2 + bdl^2 + 2bclk]} \nonumber \\ &\times (\mathcal{D}_{g\cdot(k,l)} \tilde{\mathbf{x}})[n],
    \end{align*}
    which rotates the ambiguity function by $g = \begin{bmatrix}
        a & b \\ c & d
    \end{bmatrix} \in \mathcal{G}$:
    \begin{align*}
        \big|\mathbf{A}_{\mathbf{x},\mathbf{y}}[k,l]\big| &= \big|\mathbf{A}_{\mathcal{F}_a \mathbf{x},\mathcal{F}_a \mathbf{y}}[g \cdot (k,l)]\big|.
    \end{align*}
\end{lemma}

We have shown in~\cite{Mehrotra2025_WCLSpread,Mehrotra2025_EURASIP,Mattu2025_npj} that the GDAFT maps the Zak-OTFS pulsone basis from Example~\ref{ex:comm_subgrp_ex1} to the AFDM chirp basis from Example~\ref{ex:comm_subgrp_ex2}, hence establishing equivalence between Zak-OTFS and AFDM. 


\section{Introduction to Zak-OTFS}
\label{sec:zak_intro}

\begin{figure}
    \centering
    \includegraphics[width=\linewidth]{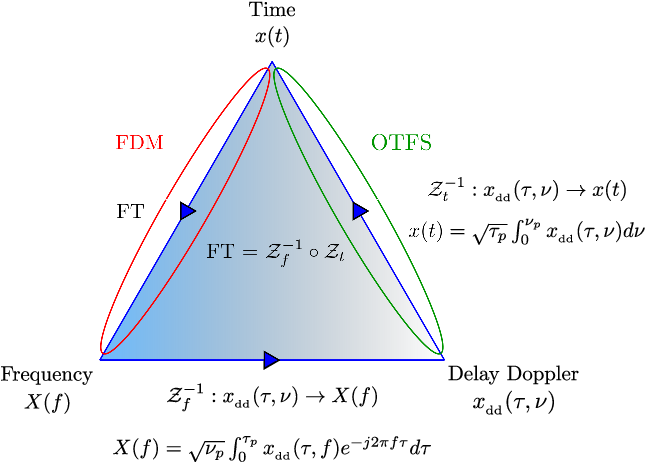}
    \caption{Unitary maps between various signal domains.}
    \label{fig:signaldomains}
\end{figure}

\begin{figure*}
    \centering
    \includegraphics[width=\linewidth]{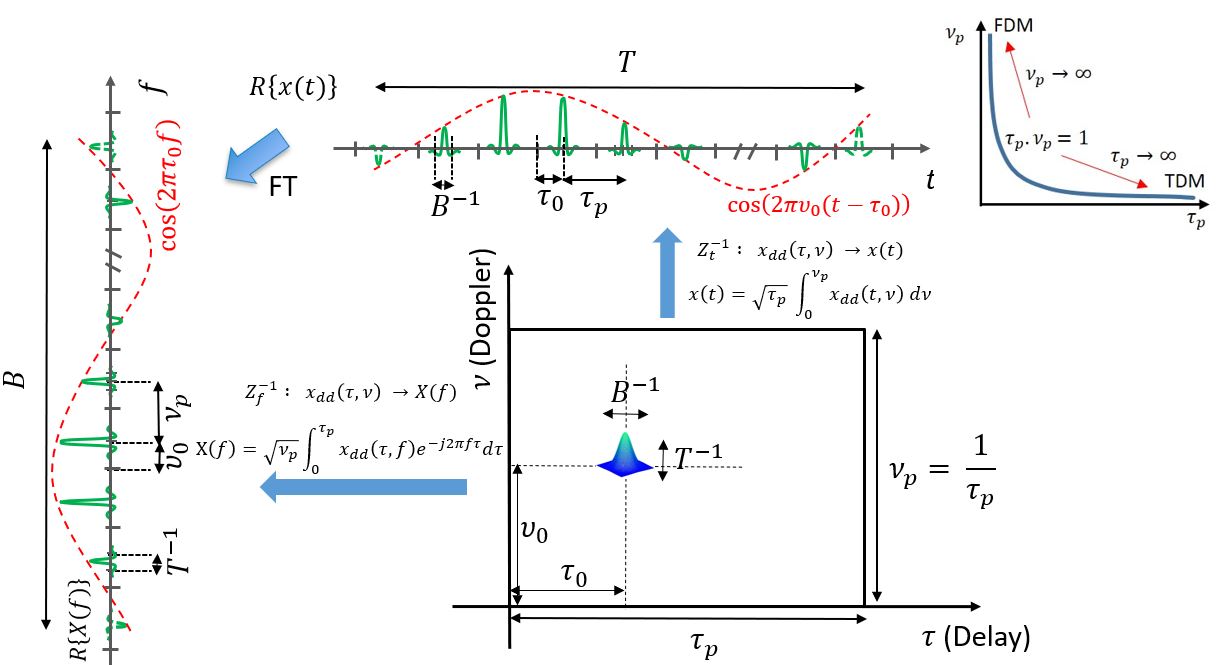}
    \caption{A DD domain pulse and its TD/FD realizations referred to as TD/FD pulsone. The TD pulsone comprises of a finite duration pulse train modulated by a TD tone. The FD pulsone comprises of a finite bandwidth pulse train modulated by a FD tone. The location of the pulses in the TD/FD pulse train and the frequency of the modulated TD/FD tone is determined by the location of the DD domain pulse $(\tau_0,\nu_0)$. The time duration and bandwidth of a pulsone are inversely proportional to the characteristic width of the DD domain pulse along the Doppler axis and the delay axis, respectively. The number of non-overlapping DD pulses, each spread over an area $B^{-1}T^{-1}$, inside the fundamental period ${\mathcal D}_0$ (which has unit area) is equal to the time-bandwidth product $BT$ and the corresponding pulsones are orthogonal to one another, rendering OTFS an orthogonal modulation that achieves the Nyquist rate.  As $\nu_p \rightarrow \infty$, the FD pulsone approaches a single FD pulse which is the FDM carrier. Similarly, as $\tau_p \rightarrow \infty$, the TD pulsone approaches a single TD pulse which is the TDM carrier. OTFS is therefore a family of modulations parameterized by $\tau_p$ that interpolates between TDM and FDM.}
    \label{fig:sec3fig1}
\end{figure*}

\begin{figure}
    \centering
    \includegraphics[width=0.99\linewidth]{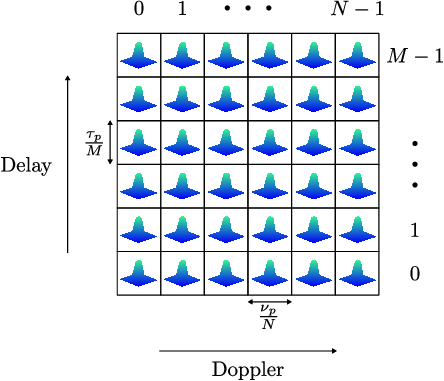}
    \caption{DD domain information grid. There are $M$ delay bins and $N$ Doppler bins. The spacing along the Doppler axis is $\nicefrac{\nu_p}{N}$ and along delay axis is $\nicefrac{\tau_p}{M}$. The DD grid can carry $MN$ information symbols.}
    \label{fig:sec3fig2}
\end{figure}

\subsection{Equivalent Signal Representations in Other Domains}
\label{subsec:zak_intro_dd_mapping}

In Section~\ref{sec:foundations}, we analyzed modulations for doubly-selective channels in the time domain. Here, we describe equivalent signal representations in domains such as the delay-Doppler (DD) domain, as illustrated in Fig.~\ref{fig:signaldomains}. To that end, we define various transforms that map time-domain signals to DD.

\begin{definition}[\cite{dzt}]
    \label{def:dzt}
    The \emph{discrete Zak transform} (DZT) maps $MN$-periodic sequences $\mathbf{x}$ to $M \times N$ quasi-periodic arrays $\mathbf{X}$:
    \begin{align*}
        \mathbf{X}[k, l] = \frac{1}{\sqrt{N}} \sum_{p \in \mathbb{Z}_{N}}\mathbf{x}[k+pM]e^{-\frac{j2\pi}{N} pl},
    \end{align*}
    where $k \in \mathbb{Z}_{M}$ is the delay and $l \in \mathbb{Z}_{N}$ is the Doppler index. The $MN$-periodic sequences $\mathbf{x}$ (resp. $M \times N$ quasi-periodic arrays $\mathbf{X}$) form a Hilbert space, and inner products can be calculated on any $MN$ length interval (resp. $M \times N$ subarray).
\end{definition}

\begin{lemma}
    \label{lmm:dzt_unitary}
    The DZT in Definition~\ref{def:dzt} is a unitary transform that preserves inner products.
\end{lemma}

\begin{IEEEproof}
Consider the following basis for the Hilbert space $\mathcal{H}$ of $MN$-periodic sequences:
\begin{align}
    \label{eq:timebasis}
    \mathbf{v}_{r,s}[n] = \begin{cases}
        \frac{1}{\sqrt{M}} e^{\frac{j2\pi}{M} sn}, \quad \text{if } rM \leq n < (r+1)M \\ 
        0, \quad \text{otherwise},
    \end{cases}
\end{align}
where $\mathbf{v}_{r,s}[n]$ is repeated periodically with period $MN$, which closely resembles the discrete time version of the pulsone basis element presented in Example~\ref{ex:comm_subgrp_ex1}. The basis is orthonormal:
\begin{align*}
    \big\langle \mathbf{v}_{r,s}, \mathbf{v}_{r',s'} \big\rangle &= \frac{1}{M} \sum_{n \in \mathbb{Z}_{MN}} \mathds{1}_{\{rM \leq n < (r+1)M\}} e^{\frac{j2\pi}{M} sn} \nonumber \\ &~~~~~~~~~~~~\times \mathds{1}_{\{r'M \leq n < (r'+1)M\}} e^{-\frac{j2\pi}{M} s'n} \nonumber \\
    &= \delta[r-r'] \frac{1}{M} \sum_{n = rM}^{(r+1)M-1} e^{\frac{j2\pi}{M} (s-s')n} \nonumber \\
    &= \delta[r-r']  \delta[s-s'],
\end{align*}
where the last expression follows from Identity~\ref{idty:sumrootsofunity}. Per Definition~\ref{def:dzt}, the DZT of~\eqref{eq:timebasis} is:
\begin{align}
    \label{eq:timebasis_dzt}
    \mathbf{V}_{r,s}[k,l]\!&=\!\frac{1}{\sqrt{N}} \sum_{p \in \mathbb{Z}_{N}}\mathbf{v}_{r,s}[k+pM] e^{-\frac{j2\pi}{N} pl} \nonumber \\
    &=\!\frac{1}{\sqrt{MN}}\hspace{-1mm} \sum_{p \in \mathbb{Z}_{N}}\hspace{-2mm} e^{j2\pi\big[\frac{s(k+pM)}{M}-\frac{pl}{N}\big]} \mathds{1}_{\{rM \leq k+pM < (r+1)M\}} \nonumber \\
    &=\!\frac{1}{\sqrt{MN}} \sum_{p \in \mathbb{Z}_{N}} e^{j2\pi\big[\frac{s(k+pM)}{M}-\frac{pl}{N}\big]} \delta\bigg[p-r+\bigg\lfloor\frac{k}{M}\bigg\rfloor\bigg] \nonumber \\
    &=\!\frac{1}{\sqrt{MN}} e^{\frac{j2\pi}{M}s(k+(r-\lfloor\frac{k}{M}\rfloor)M)} e^{-\frac{j2\pi}{N} (r-\lfloor\frac{k}{M}\rfloor)l} \nonumber \\ 
    &=\!\frac{1}{\sqrt{MN}} e^{\frac{j2\pi}{M}sk} e^{-\frac{j2\pi}{N} (r-\lfloor\frac{k}{M}\rfloor)l}.
\end{align} 

The obtained basis is clearly orthonormal since:
\begin{align*}
   \big\langle \mathbf{V}_{r,s}, \mathbf{V}_{r',s'} \big\rangle &= \frac{1}{MN} \sum_{k = 0}^{M-1} \sum_{l = 0}^{N-1} e^{j\frac{2\pi}{M}(s-s')k} e^{-j\frac{2\pi}{N}(r-r')l} \nonumber \\ 
   &= \delta[r-r'] \delta[s-s'],
\end{align*}
where the final result follows from Identity~\ref{idty:sumrootsofunity}. Hence, the DZT preserves inner products since it maps the orthonormal basis $\mathbf{v}_{r,s}$ to the orthonormal basis $\mathbf{V}_{r,s}$.
\end{IEEEproof}

The DZT maps shifts in time-frequency to delay-Doppler.

\begin{lemma}[\cite{Mehrotra2025_EURASIP}]
    \label{lmm:heis_op_dzt}
    Let $\mathcal{Z}$ denote the operator representation of the DZT. Conjugation by the DZT as in Definition~\ref{def:sympl_transform} maps time-frequency shifts per Definition~\ref{def:heis_op} to delay-Doppler shifts:
    \begin{align*}
        \big(\mathcal{Z} \circ \mathcal{D}_{(k,l)} \circ \mathcal{Z}^{-1}\mathbf{X}\big)[k',l'] &= \mathbf{X}[(k'-k)_{{}_{M}},(l'-l)_{{}_{N}}] \nonumber \\ &\times e^{\frac{j2\pi}{N}(l'-l)\big\lfloor \frac{k'-k}{M} \big\rfloor} e^{\frac{j2\pi}{MN}l(k'-k)}.
    \end{align*}
\end{lemma}

Hence, an equivalent description of the time-domain Heisenberg-Weyl group in Section~\ref{subsec:foundation_hwgroup} is possible in the DD domain. Let $\mathcal{H}_{\mathsf{DD}}$ denote the Hilbert space of all complex-valued $M \times N$ quasi-periodic arrays formed by the DZT of $MN$-periodic sequences in $\mathcal{H}$, and let $\mathcal{D}_{(k,l)}^{\mathsf{DD}}$ denote the delay-Doppler shift operator in Lemma~\ref{lmm:heis_op_dzt}, which is the DD equivalent of the Heisenberg operator from Definition~\ref{def:heis_op}. The DD equivalent of the time-domain Heisenberg-Weyl group from Theorem~\ref{thm:hw_group} may be defined as follows.

\begin{theorem}[\cite{Mehrotra2025_EURASIP}]
    \label{thm:hw_group_dd}
    Let $\mathcal{H}_{MN}^{\mathsf{DD}}$ denote the collection of phase-scaled DD Heisenberg operators $\mathcal{D}_{(k,l)}^{\mathsf{DD}}$:
    \begin{align*}
        \mathcal{H}_{MN}^{\mathsf{DD}} &= \big\{e^{\frac{j2\pi}{MN}m} \mathcal{D}_{(k,l)}^{\mathsf{DD}} \big| k,l,m \in \mathbb{Z}_{MN} \big\}. 
    \end{align*}

    The set $\mathcal{H}_{MN}^{\mathsf{DD}}$ under operator composition $\circ$ forms the Heisenberg-Weyl group in the delay-Doppler domain, which is unitarily equivalent to the Heisenberg-Weyl group representation $\mathcal{H}_{MN}$ in Theorem~\ref{thm:hw_group} over the Hilbert space of $MN$-periodic sequences.
\end{theorem}

The above discussion has assumed Nyquist rate sampling of all signals, as indicated in~\eqref{eq:prelim2} in the time domain. In continuous time, the DZT is replaced by the \emph{continuous Zak transform}. Below, we describe the continuous DD domain processing underlying the above discussion.

\subsection{Introduction to Zak-OTFS}
\label{subsec:zak_intro_to_zak}

Zak-OTFS~\cite{bitspaper1, bitspaper2} is a two-dimensional DD modulation scheme determined by two parameters: the delay period $\tau_p$ and the Doppler period $\nu_p$, with $\nu_p\tau_p = 1$. Figure~\ref{fig:sec3fig1} shows the realization of Zak-OTFS. Information symbols are mounted in DD domain using pulses. A DD pulse is a quasi-periodic localized function defined by the delay period $\tau_p$ and the Doppler period $\nu_p=\nicefrac{1}{\tau_p}$. The fundamental period in the DD domain is defined as:
\begin{align}
    \{(\tau,\nu): 0\leq\tau<\tau_p, 0\leq\nu<\nu_p\},
\end{align}
where $\tau$ and $\nu$ represent the delay and Doppler variables, respectively. The DD-domain pulse has spread $\nicefrac{1}{B}$ in the delay domain and $\nicefrac{1}{T}$ in the Doppler domain. The time domain realization of the DD domain pulse is obtained via the inverse-time Zak transform~\cite{otfs_book} ($\mathcal{Z}_t^{-1}$). In the time domain (TD), the carrier waveform consists of a pulse train modulated by a tone, hence referred to as the ``pulsone''. Similarly, computing the inverse frequency Zak transform\footnote{Inverse frequency-Zak transform of a DD function $a(\tau,\nu)$ is defined as $\mathcal{Z}_{f}^{-1}(a(\tau,\nu)) = \sqrt{\nu_{\mathrm p}} \int_0^{\tau_{\mathrm p}} a(\tau,f)e^{-j2\pi f\tau} d\tau$.}~\cite{otfs_book} ($\mathcal{Z}_f^{-1}$) of the DD domain pulse also results in a pulsone in the frequency domain (FD). When the delay period is selected to be very large, the Zak-OTFS system becomes equivalent to a time division multiplexing (TDM) system and when the Doppler period is selected to be very large, the Zak-OTFS system becomes equivalent to a frequency division multiplexing (FDM) system. As shown in the right corner of Fig.~\ref{fig:sec3fig1}, the Zak-OTFS system can interpolate between TDM and FDM via the choice of $\tau_p$ and $\nu_p$.

The Zak-OTFS grid consists of $M$ delay bins and $N$ Doppler bins. Figure~\ref{fig:sec3fig2} shows the DD grid. Bins along the delay are spaced $\nicefrac{\tau_p}{M}$ apart and bins along Doppler are spaced $\nicefrac{\nu_p}{N}$ apart. The frame occupies bandwidth $B = 1/\nicefrac{\tau_p}{M} = M\nu_p$ and time $T = 1/\nicefrac{\nu_p}{N} = N\tau_p$. Next, we derive and describe the system model for Zak-OTFS.

The fundamental period consisting of $M$ delay bins and $N$ Doppler bins can be expressed as:
\begin{align}
    \mathcal{D}_0 = \bigg\{\bigg(k\frac{\tau_{{\mathrm p}}}{M},l\frac{\nu_{{\mathrm p}}}{N}\bigg) \bigg\vert \ k=0,\ldots,M-1,l=0,\ldots,N-1\bigg\}.
\end{align}
In each frame, $MN$ information symbols drawn from a modulation alphabet ${\mathbb A}$ (e.g., 4-QAM), $\mathbf{x}[k,l]\in {\mathbb A}$, $k=0,\ldots,M-1$, $l=0,\ldots,N-1$, are multiplexed in the DD domain. The $\mathbf{x}[k,l]$s are encoded according to the following equation to obtain a quasi-periodic extension of the signal in the discrete DD domain:
\begin{equation}
\mathbf{x}_{{}_\mathrm{dd}}[k+nM,l+mN] = \mathbf{x}[k,l]e^{j2\pi n \frac{l}{N}}, \ n,m\in\mathbb{Z}.
\label{eq:quasi_per}
\end{equation}
These discrete DD domain signals $\mathbf{x}_{{}_\mathrm{dd}}[k,l]$s are supported on the information lattice given by:
\begin{align}
    \Lambda_{\mathrm{dd}}=
\bigg\{\bigg(k\frac{\tau_{\mathrm p}}{M},l\frac{\nu_{\mathrm p}}{N}\bigg) \bigg\vert \ k,l\in \mathbb{Z}\big\}.    
\end{align}
The continuous DD domain information signal is given by:
\begin{equation}
x_{{}_\mathrm{dd}}(\tau,\nu)=\sum_{k,l\in \mathbb{Z}} \mathbf{x}_{{}_\mathrm{dd}}[k,l] \delta\Big(\tau-\frac{k\tau_{\mathrm p}}{M}\Big)\delta\Big(\nu-\frac{l\nu_{\mathrm p}}{N}\Big),
\end{equation}
where $\delta(\cdot)$ denotes the Dirac delta impulse function. For any $n,m\in \mathbb{Z}$, we have
$x_{{}_\mathrm{dd}}(\tau+n\tau_{\mathrm{p}},\nu+m\nu_{\mathrm{p}})=e^{j2\pi n\nu \tau_{\mathrm{p}}}x_{{}_\mathrm{dd}}(\tau,\nu)$,
so that $x_{{}_\mathrm{dd}}(\tau,\nu)$ is periodic with period $\nu_{\mathrm p}$ along the Doppler axis and quasi-periodic with period $\tau_{\mathrm p}$ along the delay axis.

Since the TD pulsones are not bounded in time and require infinite bandwidth (due to localization of pulses), the Zak-OTFS system using just the pulsones is not practical. The transmit signal needs to be filtered to restrict the bandwidth and the duration of the TD signal. This is achieved via pulse shaping. The DD domain transmit signal $x_{{}_\mathrm{dd}}^{w_{\mathrm{tx}}}(\tau,\nu)$ is given by the twisted convolution of the transmit pulse shaping filter $w_{{}_\mathrm{tx}}(\tau,\nu)$ with $x_{{}_\mathrm{dd}}(\tau,\nu)$ as:
\begin{align}
    x_{{}_\mathrm{dd}}^{w_{\mathrm{tx}}}(\tau,\nu) = w_{{}_\mathrm{tx}}(\tau,\nu)*_{\sigma}x_{{}_\mathrm{dd}}(\tau,\nu),    
\end{align}
where $*_{\sigma}$ denotes the twisted convolution\footnote{Twisted convolution of two DD functions $a(\tau,\nu)$ and $b(\tau,\nu)$ is defined as $a(\tau,\nu) \ast_\sigma b(\tau,\nu) = \iint a(\tau', \nu') b(\tau-\tau',\nu-\nu')e^{j2\pi\nu'(\tau-\tau')}d\tau'  d\nu'$.}. The transmitted TD signal $s_{{}_\mathrm{td}}(t)$ is the TD realization of $x_{{}_\mathrm{dd}}^{w_{\mathrm{tx}}}(\tau,\nu)$ obtained via the inverse time Zak transform\footnote{Inverse time-Zak transform of a DD function $a(\tau, \nu)$ is given by $\mathcal{Z}_t^{-1}(a(\tau, \nu)) = \sqrt{\tau_\mathrm{p}}\int_0^{\nu_{\mathrm{p}}}a(t, \nu) d\nu$}~\cite{otfs_book} given by
\begin{align}
    s_{{}_\mathrm{td}}(t)=Z_{t}^{-1}\left(x_{{}_\mathrm{dd}}^{w_{\mathrm{tx}}}(\tau,\nu)\right).
\end{align} 
The transmit pulse shaping filter $w_{{}_\mathrm{tx}}(\tau,\nu)$ limits the time and bandwidth of the transmitted signal $s_{\mathrm{td}}(t)$. The transmit signal $s_{{}_\mathrm{td}}(t)$ passes through a doubly-selective channel, resulting in the output signal $r_{{}_\mathrm{td}}(t)$ given by
\begin{align}
    r_{{}_\mathrm{td}}(t) = \iint h_{\mathrm{phy}}(\tau,\nu)s_{{}_{\mathrm{td}}}(t-\tau)e^{j2\pi\nu(t-\tau)} d\tau d\nu,
\end{align}
where $h_{\mathrm{phy}}(\tau,\nu)$ is the DD domain impulse response of the physical channel given by
\begin{align}
    \label{eq:hphy}
    h_{\mathrm{phy}}(\tau,\nu)=\sum_{i=1}^{P}h_{i}\delta(\tau-\tau_{i})\delta(\nu-\nu_{i}),
\end{align}
where $P$ denotes the number of (resolvable) DD paths, and the $i$th path has gain $h_{i}$, delay shift $\tau_{i}$, and Doppler shift $\nu_{i}$. 

The received TD signal $y(t)$ at the receiver is given by 
\begin{align}
    y(t)=r_{{}_\mathrm{td}}(t)+n(t)
\end{align}
where $n(t)$ is AWGN with variance $N_{0}$, i.e., $\mathbb{E}[n(t)n(t+t')]=N_{0}\delta(t')$. The TD signal $y(t)$ is converted to the corresponding DD domain signal $y_{\mathrm{dd}}(\tau,\nu)$ by applying Zak transform\footnote{Zak transform of a continuous TD signal $a(t)$ is defined as
$\mathcal{Z}_t\left(a(t)\right) = \sqrt{\tau_p} \sum_{k \in \mathbb{Z}} a(\tau + k \tau_{\mathrm p}) e^{-j2\pi\nu k\tau_{\mathrm p}}$.}, i.e.,
\begin{eqnarray}
y_{{}_\mathrm{dd}}(\tau,\nu) = \mathcal{Z}_{t}(y(t)) 
= r_{{}_\mathrm{dd}}(\tau,\nu)+n_{{}_\mathrm{dd}}(\tau,\nu),
\end{eqnarray}
where 
\begin{align}
    r_{{}_\mathrm{dd}}(\tau,\nu)=h_{\mathrm{phy}}(\tau,\nu)*_{\sigma}w_{{}_\mathrm{tx}}(\tau,\nu)*_{\sigma}x_{{}_\mathrm{dd}}(\tau,\nu)
\end{align}
is the Zak transform of $r_{{}_\mathrm{td}}(t)$, given by the twisted convolution of $x_{{}_\mathrm{dd}}(\tau,\nu)$, $w_{{}_\mathrm{tx}}(\tau,\nu)$, and $h_{\mathrm{phy}}(\tau,\nu)$,  and $n_{{}_\mathrm{dd}}(\tau,\nu)$ is the Zak transform of $n(t)$. The receive filter is matched to the transmit pulse, i.e., $w_{{}_\mathrm{rx}}(\tau,\nu) = w_{{}_\mathrm{tx}}^*(-\tau,-\nu)e^{j2\pi\nu\tau}$~\cite{Calderbank2025_isac} and acts on $y_{{}_\mathrm{dd}}(\tau,\nu)$ through twisted convolution to give the output 
\begin{align}
    y_{{}_\mathrm{dd}}^{w_{\mathrm{rx}}}(\tau,\nu) &= w_{{}_\mathrm{rx}}(\tau,\nu)*_{\sigma}y_{{}_\mathrm{dd}}(\tau,\nu) \nonumber \\ 
      &= \underbrace{w_{{}_\mathrm{rx}}(\tau,\nu)*_{\sigma}h_{\mathrm{phy}}(\tau,\nu)*_{\sigma}w_{{}_\mathrm{tx}}(\tau,\nu)}_{\ \mathbf{h}_{\mathrm{eff}}(\tau,\nu)}*_{\sigma}x_{{}_\mathrm{dd}}(\tau,\nu) \nonumber \\ 
    & ~~~~+ \underbrace{w_{{}_\mathrm{rx}}(\tau,\nu)*_{\sigma}n_{{}_\mathrm{dd}}(\tau,\nu)}_{\ n_{\mathrm{dd}}^{w_{\mathrm{rx}}}(\tau,\nu)}, 
    \label{eq:sec3eq19}
\end{align}
where $\mathbf{h}_{\mathrm{eff}}(\tau,\nu)$ denotes the ``effective channel'' consisting of the twisted convolution of $w_{{}_\mathrm{tx}}(\tau,\nu),\ h_{\mathrm{phy}}(\tau,\nu)$, and $w_{{}_\mathrm{rx}}(\tau,\nu)$, and $n_{{}_\mathrm{dd}}^{w_{\mathrm{rx}}}(\tau,\nu)$ denotes the noise filtered through the receive filter. The DD signal $y_{{}_\mathrm{dd}}^{w_{\mathrm{rx}}}(\tau,\nu)$ is sampled on the information grid $\Lambda_{\mathrm{dd}}$, resulting in the discrete quasi-periodic DD domain received signal $y_{{}_\mathrm{dd}}[k,l]$ as
\begin{align}
    \mathbf{y}_{{}_\mathrm{dd}}[k,l]&= y_{{}_\mathrm{dd}}^{w_{\mathrm{rx}}}\left(\tau=\frac{k\tau_{\mathrm p}}{M},\nu=\frac{l\nu_{\mathrm p}}{N}\right) \nonumber \\
    &= \sum_{k',l'\in\mathbb{Z}} \mathbf{h}_{\mathrm{eff}}[k-k',l-l']\mathbf{x}_{{}_\mathrm{dd}}[k',l'] e^{j2\pi\frac{k'(l-l')}{MN}} \nonumber \\
    &~~~~~+ \mathbf{n}_{{}_\mathrm{dd}}[k,l].
    \label{eq:sec3eq20}
\end{align}
Hence, the $\mathbf{y}_{{}_\mathrm{dd}}[k,l]$ samples are given by:
\begin{align}
    \mathbf{y}_{{}_\mathrm{dd}}[k,l]=\mathbf{h}_{\mathrm{eff}}[k,l]*_{\sigma_\mathrm{d}}\mathbf{x}_{{}_\mathrm{dd}}[k,l]+\mathbf{n}_{{}_\mathrm{dd}}[k,l],
    \label{eq:sec3eq21}
\end{align}
where $*_{\sigma_\mathrm{d}}$ is twisted convolution in discrete DD domain, which can be evaluated as:
\begin{align}
    \mathbf{h}_{\mathrm{eff}}[k,l] &*_{\sigma_\mathrm{d}}\mathbf{x}_{{}_\mathrm{dd}}[k,l] = \nonumber \\&\sum_{k',l'\in\mathbb{Z}}\mathbf{h}_{\mathrm{eff}}[k-k',l-l']\mathbf{x}_{{}_\mathrm{dd}}[k',l'] e^{j2\pi\frac{k'(l-l')}{MN}},
    \label{eq:sec3eq22}
\end{align}
where the effective channel filter $\mathbf{h}_{\mathrm{eff}}[k,l]$ and filtered noise samples $n_{{}_\mathrm{dd}}[k,l]$ are respectively given by
\begin{align}
    \mathbf{h}_{\mathrm{eff}}[k,l]&=h_{\mathrm{eff}}\left(\tau=\frac{k\tau_{p}}{M},\nu=\frac{l\nu_{p}}{N}\right), \label{discr2} \\ 
    \mathbf{n}_{{}_\mathrm{dd}}[k,l]&=n_{{}_\mathrm{dd}}^{w_{\mathrm{rx}}}\left(\tau=\frac{k\tau_{p}}{M},\nu=\frac{l\nu_{p}}{N}\right).
    \label{eq:sec3eq23_24}
\end{align}
Because of the quasi-periodicity in the DD domain, it is sufficient to consider the received samples $y_{{}_\mathrm{dd}}[k,l]$ within the fundamental period $\mathcal{D}_0$. Writing the $y_{{}_\mathrm{dd}}[k,l]$ samples as a vector, the received signal model can be written in matrix-vector form as \cite{bitspaper1, bitspaper2}
\begin{align}
    \mathbf{y}=\mathbf{Hx}+\mathbf{n},
    \label{eq:sec3eq25}
\end{align}
where $\mathbf{x,y,n} \in\mathbb{C}^{MN\times 1}$, such that their $(kN+l+1)$th entries are given by $\mathbf{x}[kN+l+1]=\mathbf{x}_{{}_\mathrm{dd}}[k,l]$, $\mathbf{y}[kN+l+1]=\mathbf{y}_{{}_\mathrm{dd}}[k,l]$, $\mathbf{n}[kN+l+1] = \mathbf{n}{{}_\mathrm{dd}}[k,l]$, and $\mathbf{H}\in\mathbb{C}^{MN\times MN}$ is the channel matrix with
\vspace{-1mm}
\begin{align}
    \mathbf{H}&[k'N+l'+1,kN+l+1] = \nonumber \\
    &\sum_{m,n\in\mathbb{Z}}\mathbf{h}_{\mathrm{eff}}[k'-k-nM, l'-l-mN]e^{j2\pi nl/N} \nonumber \\
    &~~~~~~~~~\times e^{j2\pi\frac{(l'-l-mN)(k+nM)}{MN}},
    \label{eq:sec3eq26}
\end{align}
where $k',k=0,\ldots,M-1$, $l',l=0,\ldots,N-1$. The transceiver operations described above are shown in Fig.~\ref{fig:sec3fig3}

Note that the end-to-end input-output (I/O) relation in \eqref{eq:sec3eq21} is determined completely through \textit{twisted convolutions}. This is different from other modulation schemes where the I/O relations involve linear convolutions and/or element-wise multiplications.

\begin{figure}
    \centering
    \includegraphics[width=0.8\linewidth]{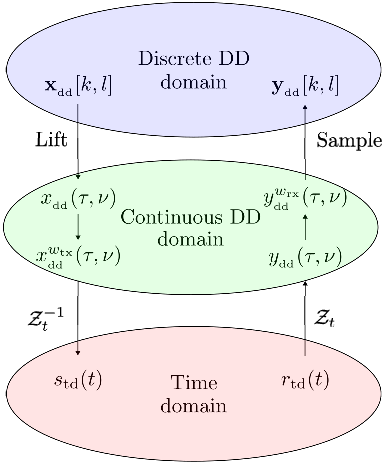}
    \caption{Zak-OTFS transceiver signal processing.}
    \label{fig:sec3fig3}
\end{figure}

\subsection{Crystallization Condition}
\label{subsec:crystallization}

Recall from~\eqref{eq:sec3eq21}, the input and output are related through twisted convolution. $\mathbf{x}_{{}_{\mathrm{dd}}}[k, l]$ is quasi-periodic (see~\eqref{eq:quasi_per}) and so is $\mathbf{y}_{{}_\mathrm{dd}}[k, l]$ because $\mathbf{y}_{{}_\mathrm{dd}}[k, l]$ is given by the twisted convolution of $\mathbf{h}_{\mathrm{eff}}[k,l]$ (which is not quasi-periodic) and $\mathbf{x}_{{}_{\mathrm{dd}}}[k, l]$. At the receiver, we need to ensure that the quasi-periodic repetitions of $\mathbf{y}_{{}_\mathrm{dd}}[k, l]$ are well separated, or in other words, the quasi-periodic repetitions do not alias. When transmitted over a channel, the DD spread of each symbol in $\mathbf{x}_{{}_{\mathrm{dd}}}[k, l]$ is governed by the DD spread in $\mathbf{h}_{\mathrm{eff}}[k,l]$, which is in turn dependent on the corresponding quantities in the physical channel. The condition for non-aliasing is that the delay and Doppler width of each quasi-periodic repetition must be greater than the delay and Doppler spread of the channel\footnote{Here the delay and Doppler spread refer to the spreads of the physical channel and not the effective channel. Hence, this theory is valid for arbitrary pulse shapes.}, respectively~\cite{bitspaper1, bitspaper2}. Let $\tau_{\max}$ and $\nu_{\max}$ respectively denote the maximum delay and Doppler spread in the physical channel, then this condition can be written as:
\begin{align}
    \tau_{\max} < \tau_p \ \mathrm{and} \ 2\vert\nu_{\max}\vert < \nu_p,
    \label{eq:sec3eq27}
\end{align}
where $2\vert\nu_{\max}\vert$ is the Doppler spread of the channel. This condition is referred to as the \textit{crystallization condition}. Since $\tau_p\nu_p = 1$, this implies that $2\tau_{\max}\vert\nu_{\max}\vert < 1$. Channels which satisfy this condition are referred to as underspread channels. Most practical channels are underspread~\cite{tse2005fundamentals}.

\subsection{Spread-Carrier Zak-OTFS}
\label{subsec:zak_spread}

Now, we describe an alternative realization of Zak-OTFS based on ``spread carriers''~\cite{Mehrotra2025_WCLSpread}. The TD representation of a DD pulse at $(k_0, l_0)$ is:
\begin{align}
    \label{eq:pulsone1}
    \mathbf{x}^{(\mathrm{p})}_{(k_0,l_0)}[n] = \frac{1}{\sqrt{N}} \sum_{d \in \mathbb{Z}} e^{\frac{j2\pi}{N} d l_0} \delta[n-k_0-dM],
\end{align}
which is also the expression of the TD pulsone described in Fig.~\ref{fig:sec3fig1}. Here, $n=0, 1, \cdots, MN-1$, $k_0 = 0, 1, \cdots, M-1, l_0 = 0, 1, \cdots, N-1$. Taking the GDAFT per Definition~\ref{def:gdaft} of the TD pulsone, we obtain~\cite{Mehrotra2025_WCLSpread}:
\begin{align}
    \label{eq:cazac1}
    \mathbf{x}^{(\mathrm{c})}_{(k_0,l_0)}[n] = \mathcal{F}_a\mathbf{x}_{p}[n] = &\frac{e^{\frac{j2\pi}{MN} (An^2+Bnk_0 +Ck_0^2)}}{\sqrt{MN}} \epsilon_N \left(\frac{CM}{N}\right)_J \nonumber \\ 
    & \times e^{-\frac{j2\pi}{N} (4CM)_N^{-1} (Bn + l_0 + 2Ck_0)^2},
\end{align}
where the simplification follows from Identity~\ref{idty:gauss_sum}. This is a constant amplitude zero autocorrelation (CAZAC) sequence~\cite{benedetto_cazac}. Information symbols can equivalently be mounted on the CAZAC sequences. Since CAZAC sequences have constant amplitude, the basis element has much better peak-to-average-power (PAPR) than the pulsone basis, which has a peaky structure (see Fig.~\ref{fig:sec3fig1}).  GDAFT is unitary, hence, the regular Zak-OTFS realized using pulsones is unitarily equivalent to the Zak-OTFS realized using the CAZAC sequences (also called spread carrier, since the energy is spread over the entire length of the CAZAC sequence).

\subsection{Over-the-Air (OTA) Demonstrations}
Below we review OTA experiments in sub-6 GHz, mmWave, and sub-THz frequency bands \cite{OTA_1, OTA_2, OTA_3} that explore the practicality of the Zak-OTFS carrier waveform. The same signal processing framework scales across all frequency bands, and the DD-domain channel representation is interpretable and approximately time-invariant, simplifying estimation and equalization. Strong compatibility with SDR platforms and high-frequency frontends demonstrate the practicality of Zak-OTFS. Spread pilot designs improve spectral efficiency and reduce PAPR~\cite{Aug2024paper} while introducing new challenges in power allocation and interference management.
\subsubsection*{Sub-6 GHz (3 GHz)}
Spread-pilot Zak-OTFS enables joint channel estimation and data detection within a single delay-Doppler frame~\cite{OTA_1}. Note that, the approach followed here does not divide DD resources between the pilot and data symbols~\cite{Aug2024paper}. The experiments demonstrate successful recovery of both channel response and data symbols, with iterative processing between channel estimation and data equalization improving the performance of both. The results highlight the critical role of the pilot-to-data power ratio (PDR) in balancing the trade-off between channel prediction accuracy and data detection reliability.

\subsubsection*{mmWave (28 GHz)} 
The mmWave experiments emphasize robustness to hardware impairments including carrier frequency offset (CFO), synchronization error, and phase noise \cite{OTA_2}. Conducted on the COSMOS testbed \cite{cosmos2025}, the results demonstrate that a short Zadoff-Chu preamble provides sufficient coarse synchronization, while the Zak-OTFS input-output relation naturally absorbs CFO and timing effects into the effective delay-Doppler channel---eliminating the need for separate impairment compensation. Root-raised-cosine filtering is shown to provide superior bit error rate (BER) performance relative to sinc filtering under mobility conditions.

\subsubsection*{Sub-THz (140/240 GHz)}
The sub-THz experiments address the most severe impairments with extreme path loss and hardware nonlinearities. The focus is integrated sensing and communications (ISAC) using both point and spread pilots \cite{OTA_3}. Spread pilots enable ISAC within a single delay-Doppler frame while reducing peak-to-average power ratio (PAPR) by more than 5 dB relative to point pilots. The delay-Doppler channel estimate remains dominated by a stable line-of-sight component across frames in the stationary setup.

\subsubsection*{BER Performance Across Frequency Bands\footnote{Note that BER measurements were obtained under different experimental conditions---sub-6 GHz used USRP X310 SDRs at 1 m separation, mmWave used COSMOS testbed at 20 m with mobility, and sub-THz used custom front-ends at 1 m stationary LOS---and path loss characteristics vary significantly across bands.}}
Even though the experimental conditions differ substantially across frequency bands, (including testbed hardware, link distances, mobility, and path loss regimes), the BER results in Fig. \ref{fig:SNR_plot}  consistently reveal the same Zak-OTFS operating principle. At sub-6 GHz (see Fig. \ref{fig:SNR_plot}a), iterative channel recovery improves BER by reducing pilot-data interference, with performance critically dependent on PDR selection. At mmWave frequencies (see Fig. \ref{fig:SNR_plot}b), BER confirms that CFO and synchronization errors are effectively absorbed into the delay-Doppler channel representation, with root-raised-cosine filtering outperforming sinc filtering. Here we use $\beta=0.5$ for the RRC filter which introduces bandwidth expansion, while the sinc filter incurs no bandwidth expansion. Notably, the BER curve exhibits an error floor at high SNR due to ill-conditioned MMSE equalization stemming from inaccurate channel estimation at these frequencies along with other hardware impairments\footnote{Note that the performance here is different from that shown in Fig.~\ref{fig:perfectcsi}. This performance difference arises from various factors such as the difference in signal processing in the practical vs the theoretic implementation, channel observed between the transmit and receive antenna vs the channel model assumed, and so on.}. Similar observations hold in the sub-THz band, where experiments at 140 GHz and 240 GHz suggest that performance degrades with increasing multiplier stages in the front-end. Collectively, these results, obtained across diverse testbeds and propagation environments, demonstrate that Zak-OTFS performance is not fundamentally constrained by carrier frequency; rather, it is governed by channel prediction quality, synchronization accuracy, pilot design effectiveness, and RF front-end implementation in the target band.

\begin{figure}[!t]
    \centering
     \begin{subfigure}{\linewidth}
        \centering
        \includegraphics[width=0.81\linewidth]{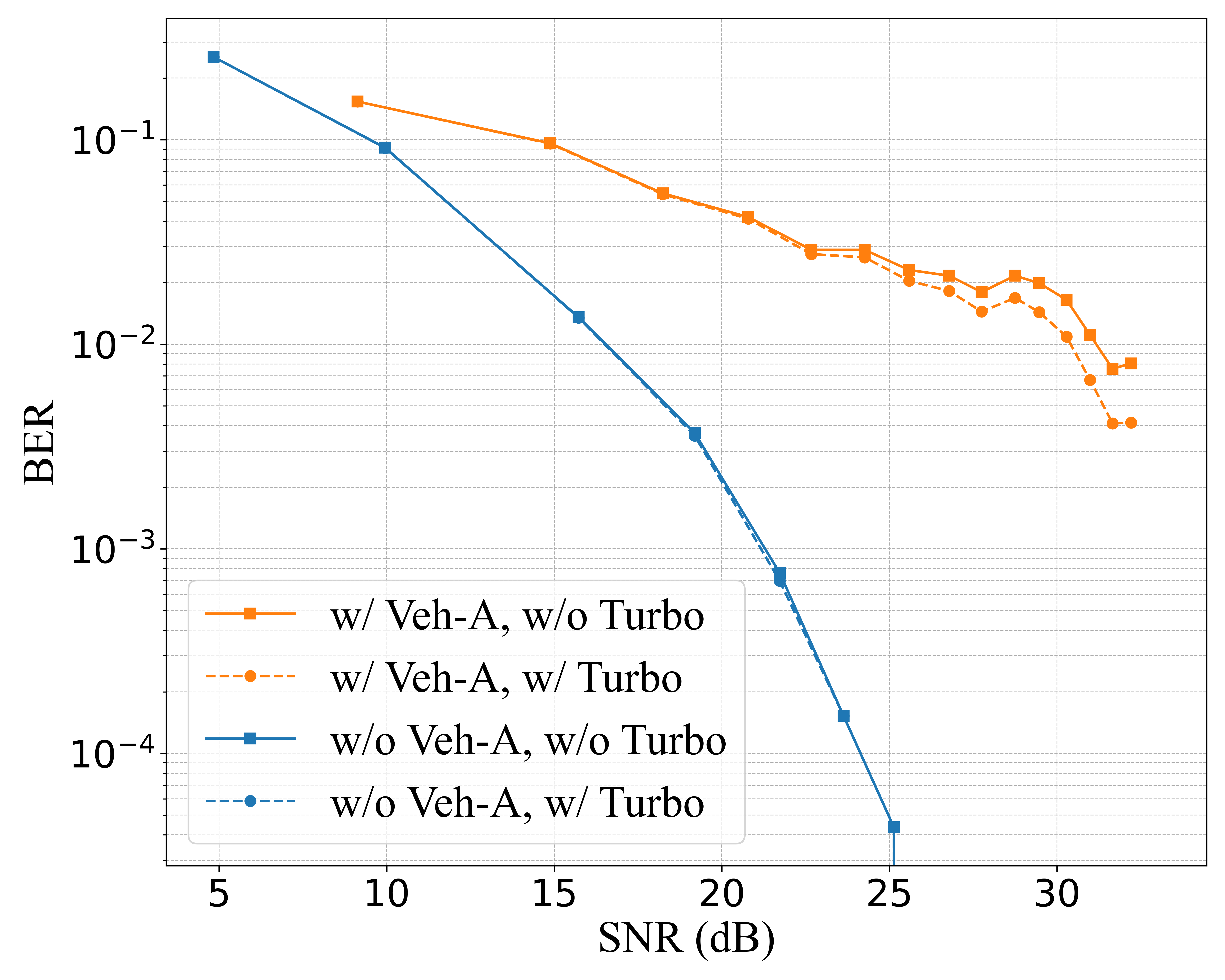}
        \caption{Sub-6 GHz \cite{OTA_1}}
        \label{fig:BER_SNR_4QAM}
    \end{subfigure}
    \\
    \begin{subfigure}{\linewidth}
         \centering
    \includegraphics[width=0.9\linewidth]{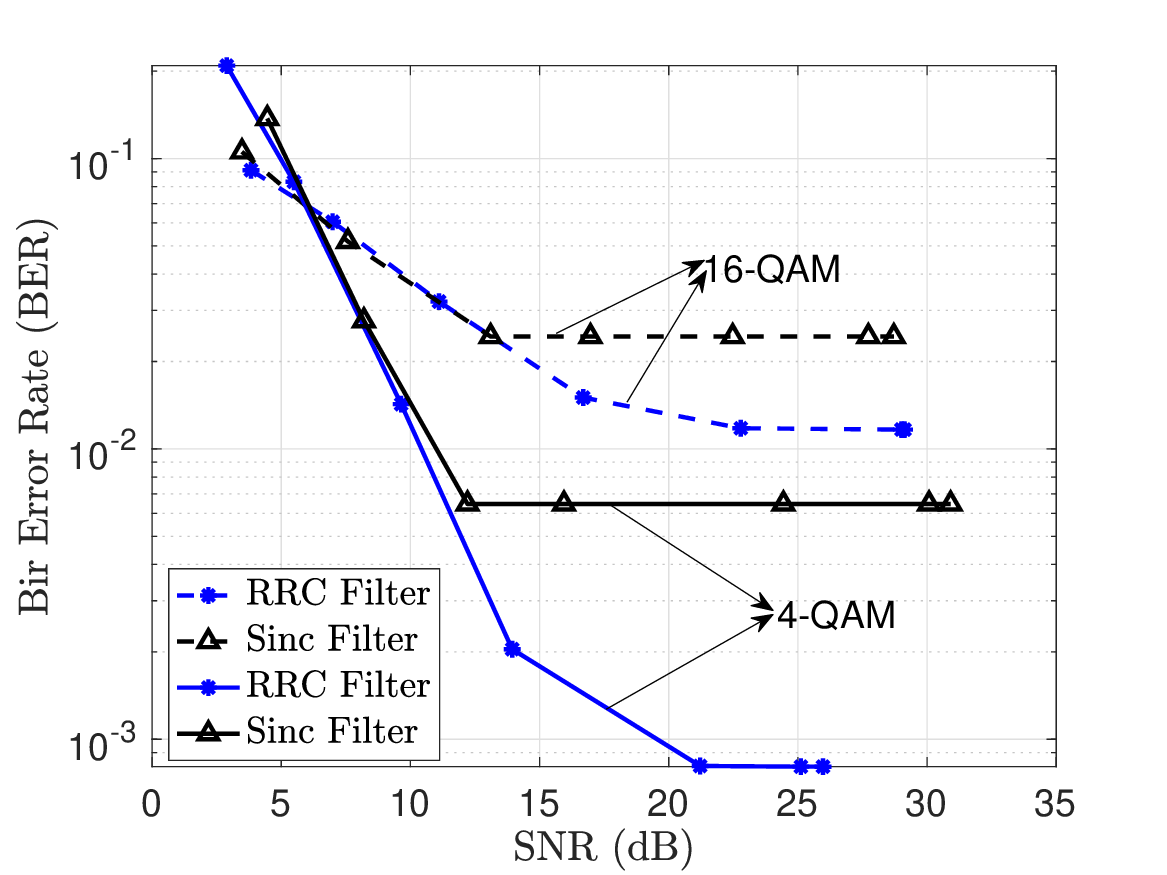}
    \caption{ mmWave \cite{OTA_3}}
    \label{fig:ber}
    \end{subfigure}
     \caption{BER versus SNR at different frequency bands} 
    \label{fig:SNR_plot}
    \vspace{-5mm}
\end{figure}

In this Section, we have described the Zak-OTFS modulation framework and in the following Sections, we consider different aspects of this framework. Once such aspect is the pulse shaping filter used at the transmitter and receiver. In addition to limiting the duration and bandwidth of the transmitted Zak-OTFS frame, the choice of pulse shaping filters also plays an important role in the equalization/channel estimation performance. In the next Section, we take a deep dive into the filters, their characteristics and trade-offs.

\subsection{Open Research Problems}
From a practical implementation perspective, several challenges remain before Zak-OTFS can be widely deployed. A key issue is the development of low-complexity and hardware-efficient receivers capable of real-time processing for large delay-Doppler grids and higher-order modulations. Efficient pilot and power allocation strategies, particularly for spread pilots, must be optimized under realistic constraints such as limited dynamic range and imperfect synchronization. Robust operation under high mobility and rapidly time-varying channels requires improved channel tracking and adaptive frame design. Additionally, tighter integration with RF hardware is needed, including accurate modeling and mitigation of phase noise, nonlinearities, and quantization effects, especially at mmWave and sub-THz frequencies. Finally, system-level challenges such as multi-user access, synchronization overhead reduction, and seamless integration with existing communication standards remain open and critical for practical deployment.

\section{Filter Design}
\label{sec:filter_design}

In Zak-OTFS the filtering is carried out in the DD domain. In the literature of Zak-OTFS, composite filters are employed, i.e., a filter in delay domain multiplied with another filter in Doppler domain. The composite filter employed at the transmitter can be written as:
\begin{align}
    \mathbf{w}(\tau, \nu) = w_{{}_\mathrm{tx}}^{(\tau)}(\tau)w_{{}_\mathrm{tx}}^{(\nu)}(\nu),
\end{align}
where $w_{{}_\mathrm{tx}}^{(\tau)}(\tau)$ and $w_{{}_\mathrm{tx}}^{(\nu)}(\nu)$ are the filters in the delay and Doppler domain, respectively. A pulse shaping filter is characterized by three metrics: $(i)$ localization to enable accurate I/O relation estimation (sensing), $(ii)$ orthogonality on the information lattice to prevent inter-symbol interference, and $(iii)$ no time and bandwidth expansion beyond $B$ and $T$ to achieve full spectral efficiency. A filter simultaneously meeting all three objectives is ideal for both sensing and communication. We describe each metric in detail in the following Subsections.

\subsection{Localization}
A filter that is well localized is ideal for accurate estimation of the I/O relation. Recall that in Zak-OTFS we estimate the effective channel, i.e., $\mathbf{h}_{\mathrm{eff}}[k, l]$. The effective channel is the twisted convolution of the pulse shaping filter at the receiver, the physical channel, and the pulse shaping filter at the transmitter. Since the physical channel introduces fractional delay and Doppler shifts (or off-grid shifts), the received symbols are also shifted by off-grid amounts in delay and Doppler. These symbols are converted to the DD domain and filtered through the pulse shaping filter that is matched to the transmit pulse shaping filter. The output is then sampled on the grid. The on-grid samples do not correspond to the actual shifts experience from the channel and this gives rise to interference. If the filter is extremely well localized, the pulse shape dies down rapidly and the interference resulting from mismatch between the on-grid sampling instants and the off-grid shifts is minimized. It is therefore crucial to have filters that are well localized, especially for accurate estimation of the I/O relation. Examples of filters that are well localized are root raised cosine (RRC)~\cite{bitspaper1,bitspaper2,otfs_book,Calderbank2025_isac,Mohammed2024_pulseshaping,Calderbank2025_interleaved,Gopalam2024_tfwindowing, preamblepaper}, Gaussian~\cite{Mohammed2024_pulseshaping}, Gaussian-sinc~\cite{Chockalingam2025_gs}, isotropic orthogonal transform algorithm (IOTA)~\cite{Mehrotra2025_IOTA}, and Hermite~\cite{Chockalingam2025_hermite}.

\subsection{Orthogonality}
At the transmitter, information symbols are filtered through the transmit pulse shaping filter. After the filtering operation, it must be ensured that all the information symbols are orthogonal and no inter-symbol interference (ISI) or inter-carrier interference (ICI) is introduced. This is possible if the filter is designed to be orthogonal when shifted by the resolution i.e., $\nicefrac{\tau_p}{M}$ along delay axis and $\nicefrac{\nu_p}{N}$ along Doppler axis. This orthogonality is crucial to obtaining good bit-error performance as the information symbols do not interfere with one another. Examples of filters that are orthogonal are sinc~\cite{bitspaper1,bitspaper2,otfs_book,Calderbank2025_isac,Mohammed2024_pulseshaping}, RRC, and IOTA.

\subsection{Time and/or Bandwidth Expansion}
The role of the pulse shaping filters is to ensure that the bandwidth and time are limited. When the bandwidth $B$ and time $T$ are limited, the Zak-OTFS system can transmit $MN=BT$ number of information symbols. In some filter designs there maybe a trade-off between the localization along delay (Doppler) and the bandwidth (time) occupied by the pulse. That is, the pulse could be made more localized at the cost of increased bandwidth and/or time. Given a bandwidth $B$ and time $T$, this would either mean that the pulse is not localized or that $MN$ information symbols cannot be transmitted. RRC is an example of a filter that requires time and/or bandwidth expansion for good localization.

The following pulse shaping filters have been considered in prior work on Zak-OTFS~\cite{bitspaper1,bitspaper2,otfs_book,Calderbank2025_isac,Mohammed2024_pulseshaping,Calderbank2025_interleaved,Gopalam2024_tfwindowing,Chockalingam2025_gs, preamblepaper}.

\subsubsection{Sinc}

The sinc filter is given by:
\begin{align}
    \label{eq:sinc1}
    \mathbf{w}(\tau,\nu) &= \sqrt{BT}~\text{sinc}(B\tau)~\text{sinc}(T\nu).
\end{align}

\subsubsection{RRC}

The RRC filter is given by:
\begin{align}
    \label{eq:rrc1}
    \mathbf{w}(\tau,\nu) &= \sqrt{BT}~\text{rrc}_{\beta_\tau}(B\tau)~\text{rrc}_{\beta_\nu}(T\nu),
\end{align}
where $0 \leq \beta_\tau,\beta_\nu \leq 1$ and:
\begin{align}
    \label{eq:rrc2}
    \text{rrc}_{\beta}(x) &= \frac{\sin(\pi x (1-\beta)) + 4\beta x\cos(\pi x (1+\beta))}{\pi x (1-(4\beta x)^2)}.
\end{align}

When $\beta_\tau = \beta_\nu = 0$, the RRC filter specializes to the sinc filter. However, when $\beta_\tau, \beta_\nu > 0$, there is bandwidth and time expansion beyond $B$ and $T$ to $(1+\beta_\tau) B$ and $(1+\beta_\nu) T$.

\subsubsection{Gaussian}

The Gaussian filter is given by:
\begin{align}
    \label{eq:gauss1}
    \mathbf{w}(\tau,\nu) &= \sqrt{BT}\bigg(\frac{4\alpha_\tau\alpha_\nu}{\pi^2}\bigg)^{\nicefrac{1}{4}} e^{-\big[\alpha_\tau(B\tau)^2+\alpha_\nu(T\nu)^2\big]}.
\end{align}

When $\alpha_\tau = \alpha_\nu = 1.584$, there is no bandwidth and time expansion beyond $B$ and $T$.

\subsubsection{Gaussian-sinc}

The Gaussian-sinc filter is given by:
\begin{align}
    \label{eq:gs1}
    \mathbf{w}(\tau,\nu) &= \sqrt{BT}~\Omega_\tau\Omega_\nu~\text{sinc}(B\tau)~\text{sinc}(T\nu) \nonumber \\ &~~~\times e^{-\big[\alpha_\tau(B\tau)^2+\alpha_\nu(T\nu)^2\big]}.
\end{align}

When $\alpha_\tau = \alpha_\nu = 0.044$ and $\Omega_\tau = \Omega_\nu = 1.0278$, there is no bandwidth and time expansion beyond $B$ and $T$.

\subsubsection{IOTA} 
The IOTA pulse shaping filter is derived by applying IOTA procedure to an extremely well localized pulse shaping filter (called the prototype filter). The purpose of applying the procedure is to orthogonalize the pulse shape while still being very well localized without time and/or bandwidth expansion. There are two IOTA filters based on the choice of the prototype filter - Gaussian IOTA and prolate spheroidal wave function (PSWF) IOTA.

\subsubsection{Hermite}
The Hermite filter is defined as:
\begin{align}
    \mathbf{w}(\tau, \nu) = \sum_{n_1=0}^{N_c-1}\sum_{n_2=0}^{N_c-1} c_{2n_1} \phi_{2n_1}(\sigma_\tau, \tau)d_{2n_2} \phi_{2n_2}(\sigma_\nu, \nu),
\end{align}
where 
\begin{align}
    \phi_n(\sigma, t) = \sqrt{\sigma}\frac{\pi^{-1/4}}{\sqrt{2^n n!}} H_n(\sigma t) e^{-(\sigma t)^2/2},
\end{align}
$H_n(t)$ is the Hermite polynomial given by:
\begin{align}
    H_n(t) = (-1)^n e^{t^2} \frac{d^n}{dt^n} e^{-t^2},
\end{align}
$N_c$ is the number of even-ordered basis functions, and $\sigma_\nu,\sigma_\tau > 0$ is used to control the pulse width.

\begin{table}
    \centering
    \caption{Comparison of different DD pulse shaping filters.}
    {
    \setlength{\tabcolsep}{2.25pt}
    \renewcommand{\arraystretch}{1.25}
    \begin{tabular}{|>{\centering\arraybackslash}p{2.75cm}|c|c|c|}
         \hline
         Filter & Sidelobe Level & Orthogonal & Time/BW Limited \\ 
         \hline
         Sinc~\cite{bitspaper1,bitspaper2,otfs_book,Calderbank2025_isac,Mohammed2024_pulseshaping} & High & \checkmark & \checkmark \\ 
         RRC~\cite{bitspaper1,bitspaper2,otfs_book,Calderbank2025_isac,Mohammed2024_pulseshaping,Calderbank2025_interleaved,Gopalam2024_tfwindowing, preamblepaper} & Low & \checkmark & $\times$ \\ 
         Gaussian~\cite{Mohammed2024_pulseshaping} & None & $\times$ & \checkmark \\ 
         Gaussian-sinc~\cite{Chockalingam2025_gs} & Low & $\times$ & \checkmark \\
         IOTA~\cite{Mehrotra2025_IOTA} & Low & \checkmark & \checkmark \\
         Hermite~\cite{Chockalingam2025_hermite} & Low & $\times$ & \checkmark \\
         \hline
    \end{tabular}
    }
    \label{tab:filt_comp}
\end{table}

\begin{figure*}
    \centering
    \begin{subfigure}{0.49\linewidth}
    \includegraphics[width=\textwidth]{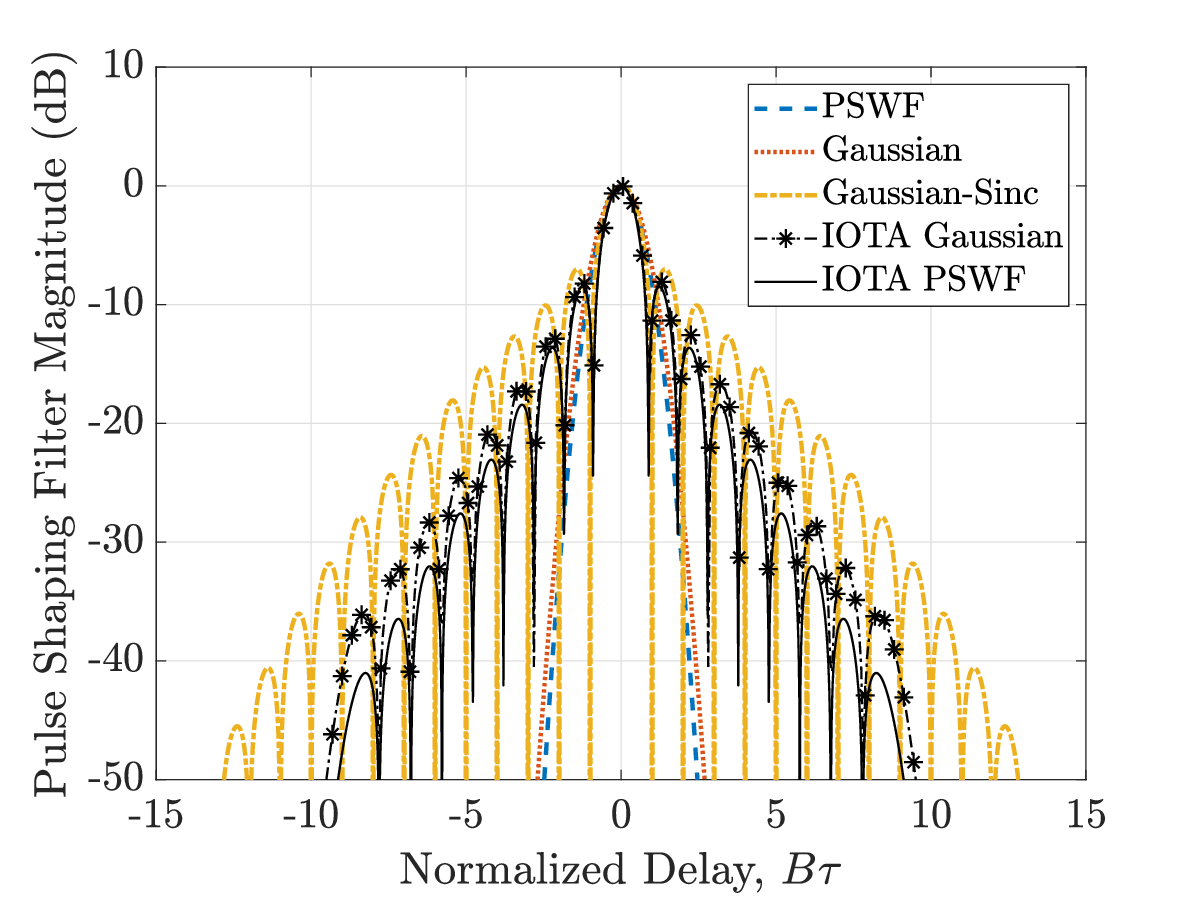}
    \caption{Pulse shaping filter magnitude vs normalized delay.}
        \label{fig:pulse_vs_del}
    \end{subfigure}
    \begin{subfigure}{0.49\linewidth}
        \includegraphics[width=\textwidth]{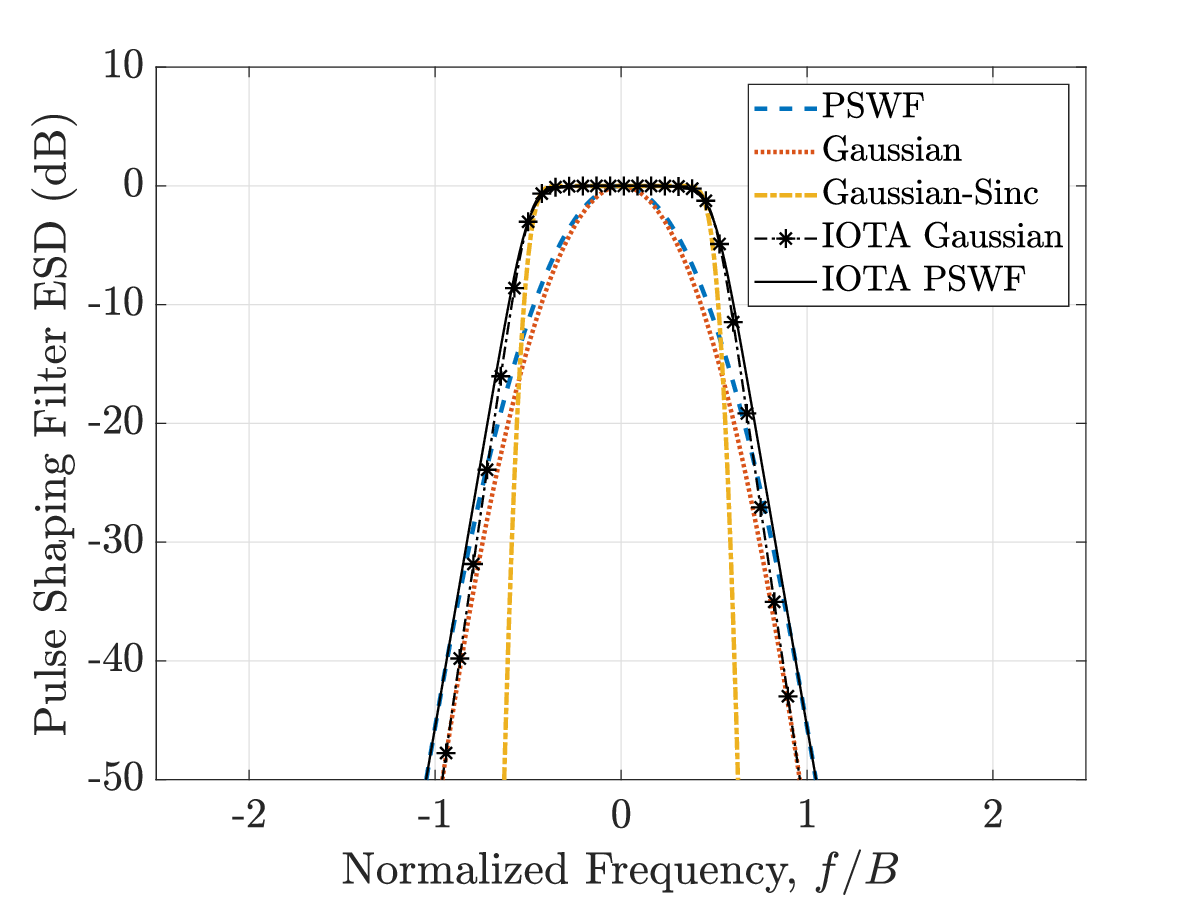}
    \caption{Energy spectral density (ESD) vs normalized frequency.}
        \label{fig:fft_vs_freq}
    \end{subfigure}


    
    \caption{Pulse shaping filter variation and energy concentration. (a): The proposed IOTA pulse shapes have a sharper main lobe and lower sidelobes compared to the Gaussian-sinc filter~\cite{Chockalingam2025_gs}. Moreover, the IOTA PSWF pulse has smaller main lobe and sidelobes compared to the IOTA Gaussian filter. (b): The IOTA pulse shapes have flat energy spectral density within $f \in [\nicefrac{-B}{2},\nicefrac{B}{2}]$ and a roll-off similar to their prototype pulses (Gaussian or PSWF). All pulse shapes have $99.99\%$ energy concentration within $f \in [\nicefrac{-B}{2},\nicefrac{B}{2}]$. See~\cite[Fig. 1]{Mehrotra2025_IOTA} for filter variation with Doppler and energy concentration in time.}
    \label{fig:pulse}
\end{figure*}

Table~\ref{tab:filt_comp} compares various pulse shaping filters in the literature in light of the metrics discussed above. It is desirable for the filter to $(i)$ have low sidelobe level which is synonymous with good localization, $(ii)$ be orthogonal, and $(iii)$ be time and bandwidth limited. While sinc is orthogonal and limited in time and bandwidth, it decays slowly and is not localized. RRC is localized only when there is time and/or bandwidth expansion. Gaussian is very well localized but is not orthogonal. Gaussian-sinc has low sidelobe levels and is time and bandwidth limited but is not orthogonal. Filters obtained through IOTA framework are the most ideal filters (they are localized, orthogonal, and limited in time/bandwidth), while Hermite pulses are not orthogonal. Note that simultaneous optimization of all three metrics is limited by the Balian-Low theorem~\cite{benedetto1994differentiation}.

Figure~\ref{fig:pulse} shows the pulse shaping filter magnitude as a function of the normalized delay and the energy spectral density (ESD) as a function of normalized frequency for various filters. We do not show RRC filter as it requires time/bandwidth expansion and Hermite filter since the filter is very similar in sidelobe levels and ESD to Gaussian-sinc filter~\cite{Chockalingam2025_hermite}. From Fig.~\ref{fig:pulse}(\subref{fig:pulse_vs_del}) PSWF and Gaussian are very well localized and have no sidelobes. Gaussian-sinc has more spread and decays more slowly. The IOTA Gaussian and IOTA PSWF are the most localized filters while also being orthogonal. From Fig~\ref{fig:pulse}(\subref{fig:fft_vs_freq}), the ESD of all the filters is approximately limited to one sided bandwidth $B$. 

\begin{figure}
    \centering
    \includegraphics[width=\linewidth]{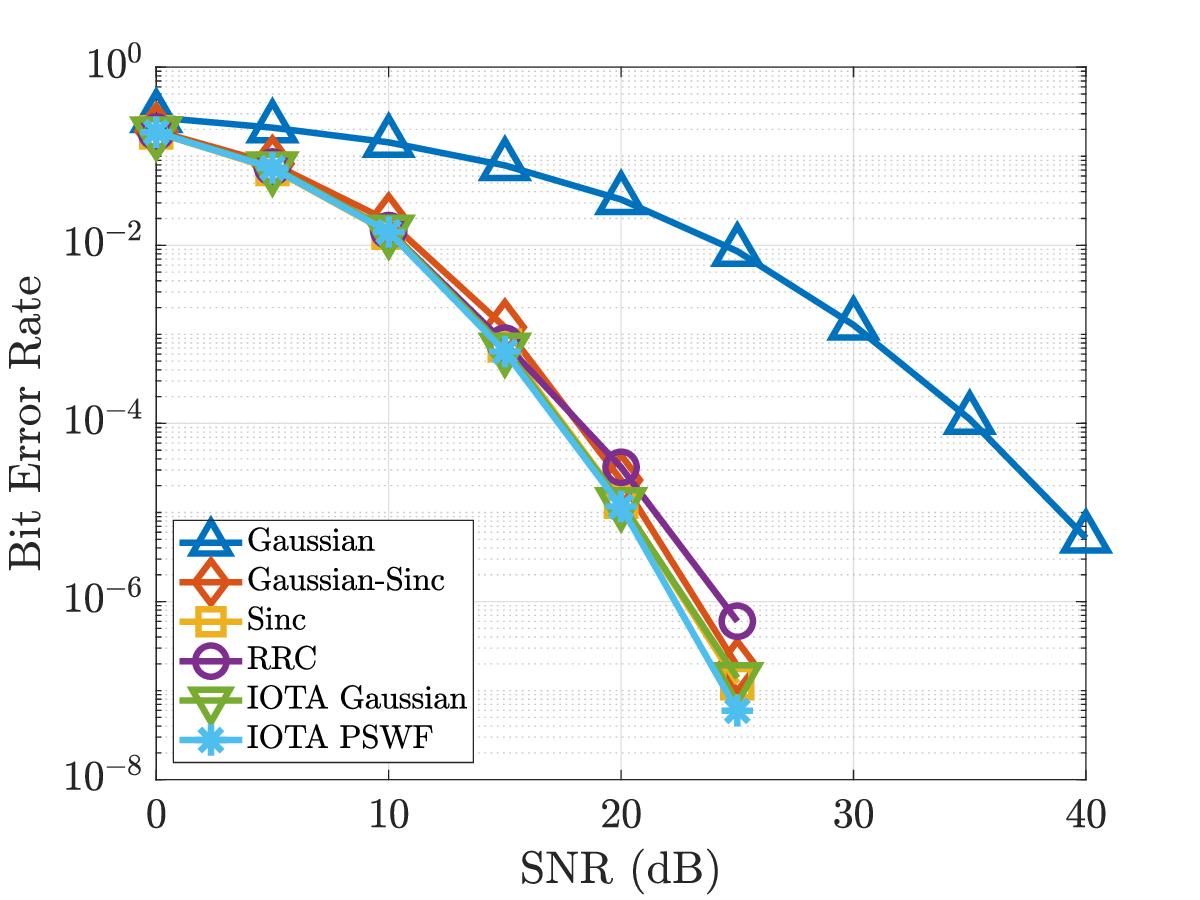}
    \caption{Uncoded $4$-QAM data detection performance with perfect I/O relation knowledge at the receiver.}
    \label{fig:perfectcsi}
\end{figure}

Figure~\ref{fig:perfectcsi} shows the bit-error performance obtained with different pulse shaping filters with perfect I/O knowledge at the receiver. The performance of the Gaussian filter is far inferior compared to all the other filters, owing to its non-orthogonality. The best performance is achieved with IOTA PSWF, IOTA Gaussian, and sinc filters, thanks to their orthogonality. Since Gaussian-sinc is approximately orthogonal, the performance is slightly worse compared to that of sinc. RRC filter follows in performance, however, it is important to note that this performance is obtained with time/bandwidth expansion. 

\subsection{Open Research Problems}
\label{subsec:filter_future}


Our discussion thus far has assumed \emph{separable} filters with the \emph{same functional structure} in delay and Doppler (i.e., $\mathbf{w}(\tau,\nu) = w_{{}_\mathrm{tx}}^{(\tau)}(\tau) w_{{}_\mathrm{tx}}^{(\nu)}(\nu)$ with $w_{{}_\mathrm{tx}}^{(\tau)}(\tau) = w_{{}_\mathrm{tx}}^{(\nu)}(\nu)$). Designing more general filters that are non-separable and have possibly different functional structure in delay and Doppler is an open problem. We have also not addressed the broad question of \emph{optimality} in our discussion -- are there DD filters that meet the Balian-Low theorem and optimally balance the three performance metrics of orthogonality, localization, no time / bandwidth expansion? Furthermore, another question that requires addressing is -- does the choice of the pulse shaping filter affect the predictability of the modulation? Finally, our discussion has been limited to single-antenna links. The question of optimal filter design for multiple-input, multiple-output (MIMO) systems is also open, where it may be possible to choose different pulse shaping filters at different antennas.




\section{Inter-Frame Predictability}
\label{sec:predictability}

\begin{figure}
    \centering
    \includegraphics[width= \linewidth]{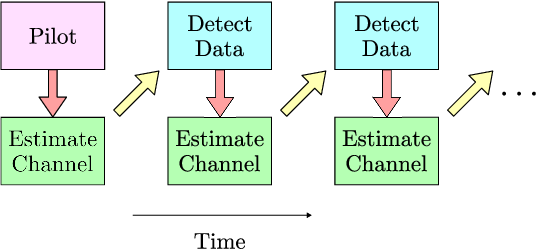}
    \caption{Block diagram of the differential communication scheme.}
    \label{fig:diff_comm_blk_dia}
\end{figure}

Signals from the transmitter undergo multiple reflections from reflectors before reaching the receiver. In the DD domain, these paths are characterized by their delay and Doppler values. The delay of a path is a function of the distance traveled by the path from the transmitter to the receiver and the Doppler of a path is a function of the relative velocity of the reflectors. In an urban environment, these reflectors typically consist of immobile objects like buildings and trees and mobile objects like vehicles. In a short span of time spanning a few frame durations (order of milliseconds), the channel can be assumed to be stationary. This is the basis for inter-frame predictability. Note that, the predictability discussed in the previous section is intra-frame predictability which deals with predictability of the channel within a single frame. 

\subsection{Precoder Design}
\label{subsec:precoder_design}
Inter-frame predictability or the quasi-stationarity of the channel across few frames allows for efficient precoder design. The precoder matrix once generated can be used for multiple frames. One example of precoder~\cite{mattu2025improving} design is presented below.

Consider the Zak-OTFS system model presented in~\eqref{eq:sec3eq25}. The QR-factorization~\cite{horn2012matrix} (alternately, the singular value decomposition of the channel matrix could also be used) of the hermitian of the channel matrix $\mathbf{H}$ is:
\begin{align}
    \mathbf{H}^{\mathsf{H}} = \mathbf{Q}\mathbf{R},
    \label{eq:prec1}
\end{align}
where $\mathbf{Q}\in \mathbb{C}^{MN\times MN}$ is an unitary matrix and $\mathbf{R}\in\mathbb{C}^{MN\times MN}$ is an upper (or lower) triangular matrix.
At the transmitter the transmit vector is precoded as:
\begin{align}
    \mathbf{x}' = \mathbf{Q}\mathbf{x}.
    \label{eq:prec2}
\end{align}
The vector $\mathbf{x}'$ is the transmit vector.
The system model becomes:
\begin{align}
    \mathbf{y} = \mathbf{R}^{\mathsf{H}}\mathbf{Q}^{\mathsf{H}}\mathbf{Q}\mathbf{x} + \mathbf{n} = \mathbf{R}^{\mathsf{H}}\mathbf{x} + \mathbf{n}.
    \label{eq:prec3}
\end{align}
To detect information symbols in $\mathbf{x}$, methods like forward substitution could be used which incurs complexity $\mathcal{O}(M^2N^2)$, compared to the naive equalization with channel matrix $\mathbf{H}$ via MMSE which incurs complexity $\mathcal{O}(M^3N^3)$~\cite{horn2012matrix}.

\subsection{Differential Communication}
\label{subsec:diff_comm}
Inter-frame predictability can also be leveraged to improve the spectral efficiency by means of reducing the pilot symbol transmission. The channel estimate obtained from the previous frame transmission could be used for equalization/detection of the current frame. However, the estimate soon becomes outdated and this necessitates transmission of another pilot symbol. A way to prevent this is to use detected data as pilot symbols and re-estimate the channel after each detection. This estimate can then be used for detection of subsequent frame and the process repeats. This prevents the estimates from becoming outdated very quickly. This is the motivation behind differential communication~\cite{mattu2025differential}. 

Differential communication techniques were introduced in the mid-20th century to address the practical challenges of coherent carrier phase recovery in radio systems. Early forms such as differential phase-shift keying (DPSK) enabled noncoherent detection by encoding information in the phase difference between consecutive symbols rather than in absolute phase, thereby simplifying receiver design and improving robustness to oscillator instability and slow phase variations \cite{feher1960dpsk, simon2005digital}. As bandwidth efficiency requirements increased, higher-order schemes such as differential quadrature phase-shift keying (DQPSK) were developed and later adopted in cellular and satellite systems, including $\pi/4$-DQPSK in second-generation digital standards \cite{proakis2001digital}. In the late 1990s, differential techniques were extended to multiple-antenna systems through differential space-time modulation, which eliminated the need for channel state information in rapidly fading environments \cite{hochwald2000differential}. Although advances in digital signal processing and pilot-assisted coherent detection reduced the performance gap and led to the dominance of coherent modulation in modern broadband systems, differential communication remains relevant in low-complexity, high-mobility, and noncoherent massive MIMO scenarios \cite{marzetta2010noncooperative}.

Block diagram of the differential communication scheme is shown in Fig.~\ref{fig:diff_comm_blk_dia}. To begin the process a pilot frame is transmitted and an estimate of the channel is obtained. Since the channel is almost stationary in one frame duration (order of milliseconds), the same estimate can be used for subsequent data frame detection. Using the detected data, another estimate of the channel is obtained and this process is repeated. This alleviates the need for frequent pilot transmissions. 

Towards achieving differential communication, we first show that the channel can effectively be estimated from the detected data. The time-domain symbols mounted on pulsone bases is:
\begin{align}
    \mathbf{x}[n] = \sum_{k_0, l_0}\mathbf{x}^{(\mathrm{p})}_{k_0, l_0}[n] \mathbf{X}[k_0, l_0],
    \label{eq:dif_det1}
\end{align}
The system model for Zak-OTFS can be represented as:
\begin{align}
    \mathbf{y}[n] = \sum_{k_0, l_0}\sum_{k, l} &\mathbf{h}_{\mathrm{eff}}[k, l]\mathbf{x}^{(\mathrm{p})}_{k_0, l_0}[n-k]e^{\frac{j2\pi}{MN}l(n-k)}\mathbf{X}[k_0, l_0] \nonumber \\
    & + \mathbf{n}[n],
    \label{eq:dif_det2}
\end{align}
where $\mathbf{x}^{(\mathrm{p})}_{k_0, l_0}[n]$ is the $(k_0, l_0)$th pulsone basis in the time-domain, $\mathbf{X}[k_0, l_0]$ is the $(k_0, l_0)$th data symbol mounted on the corresponding pulsone bases. 
\begin{lemma}
    The cross-ambiguity between a pulsone indexed by $(k_0, l_0)$ and $(k_1, l_1)$ is (up to a phase):
    \begin{align}
        \mathbf{A}_{\mathbf{X}_{k_0, l_0}^{(\mathrm{p})}, \mathbf{X}_{k_1, l_1}^{(\mathrm{p})}}[k, l] = \mathbf{A}_{\mathbf{X}^{(\mathrm{p})}, \mathbf{X}^{(\mathrm{p})}}[k - (k_0-k_1), l-(l_0-l_1)],
        \label{eq:diff_det_prelim}
    \end{align}
where $\mathbf{A}_{\mathbf{X}^{(p)}, \mathbf{X}^{(p)}}[k, l]$ is the self-ambiguity of the pulsone.
\end{lemma}
\begin{IEEEproof}
    The cross-ambiguity between a pulsone indexed by $(k_0, l_0)$ and $(k_1, l_1)$ is evaluated as:
    \begin{align}
        \mathbf{A}_{\mathbf{X}_{k_0, l_0}^{(\mathrm{p})}, \mathbf{X}_{k_1, l_1}^{(\mathrm{p})}}[k, l] &= \sum_{k_1, l_1}\sum_{n_1, m_1 \in \mathbb{Z}}\sum_{n_2, m_2 \in \mathbb{Z}}e^{\frac{j2\pi}{N}n_1l'} \times \nonumber \\
        &\hspace{5mm}\delta[k'-k_0-n_1M]\delta[l'-l_0-m_1N] \times \nonumber \\
        &\hspace{5mm}e^{-\frac{j2\pi}{N}n_2(l'-l)}\delta[k'-k-k_1-n_2M]\times \nonumber \\
        &\hspace{5mm}\delta[l'-l-l_1-m_2N]e^{-\frac{j2\pi}{MN}l(k'-k)}.
        \label{eq:diff_det_proof1}
    \end{align}
    From $\delta[k'-k_0-n_1M]\delta[l'-l_0-m_1N]$, we have $n_1 = m_1 = 0$ and $k'=k_0, l' = l_0$, since $k_0, k' = 0, 1, \cdots, M-1$ and $l_0, l' = 0, 1, \cdots, N-1$. Substituting in \eqref{eq:diff_det_proof1}:
    \begin{align}
        \mathbf{A}_{\mathbf{X}_{k_0, l_0}^{(\mathrm{p})}, \mathbf{X}_{k_1, l_1}^{(\mathrm{p})}}[k, l] &= \sum_{k_1, l_1} \sum_{n_2, m_2 \in \mathbb{Z}}e^{-\frac{j2\pi}{N}n_2(l_0-l)} \times \nonumber \\
        &\hspace{5mm}\delta[k_0-k-k_1-n_2M] \times \nonumber \\
        &\hspace{5mm}\delta[l_0-l-l_1-m_2N]e^{-\frac{j2\pi}{MN}l(k_0-k)}.
    \end{align}
    $\delta[k-(k_0-k_1)-n_2M] \implies k = (k_0 - k_1) \bmod M$ and $\delta[l-(l_0-l_1)-m_2N] \implies l = (l_0 - l_1) \bmod N$. This means that the cross-ambiguity is supported (alternately, the sum is non-zero) on a lattice given by the points $((k_0 - k_1)_M, {(l_0 - l_1)_N})$. The self ambiguity $\mathbf{A}_{\mathbf{X}^{(p)}, \mathbf{X}^{(p)}}[k, l]$ of the pulsone is supported on the lattice $(nM, mN)$ \cite{Calderbank2025_isac} and therefore (up to a phase):
    \begin{align}
        \mathbf{A}_{\mathbf{X}_{k_0, l_0}^{(\mathrm{p})}, \mathbf{X}_{k_1, l_1}^{(\mathrm{p})}}[k, l] = \mathbf{A}_{\mathbf{X}^{(\mathrm{p})}, \mathbf{X}^{(\mathrm{p})}}[k - (k_0-k_1), l-(l_0-l_1)].
        \label{eq:diff_det_proof2}
    \end{align}
\end{IEEEproof}

Consider the cross-ambiguity function between the received information symbols and the transmitted information symbols:
\begin{align}
    \mathbf{A}_{\mathbf{y}, \mathbf{x}}[k', l'] &= \frac{1}{MN}\sum_{n=0}^{MN-1}\mathbf{y}[n]\mathbf{x}^*[n-k']e^{-\frac{j2\pi}{MN}l'(n-k')} \nonumber \\
    &= \frac{1}{MN}\Bigg(\sum_{n=0}^{MN-1}\sum_{k_0, l_0}\sum_{k, l}\mathbf{h}_{\mathrm{eff}}[k, l]\mathbf{x}^{(\mathrm{p})}_{k_0, l_0}[n-k] \times \nonumber \\
    &\hspace{5mm}e^{\frac{j2\pi}{MN}l(n-k)}\mathbf{X}[k_0, l_0]\sum_{k_1, l_1}(\mathbf{x}^{(\mathrm{p})}_{k_1, l_1}[n-k'])^*\times \nonumber \\
    &\hspace{5mm}\mathbf{X}^{*}[k_1, l_1]e^{-\frac{j2\pi}{MN}l'(n-k')} + \nonumber \\
    &\hspace{5mm}\sum_{n=0}^{MN-1}\mathbf{n}[n]\mathbf{x}^*[n-k']e^{-\frac{j2\pi}{MN}l'(n-k')}\Bigg) \nonumber \\
    &\overset{(a)}{=} \frac{1}{MN}\sum_{k, l}\mathbf{h}_{\mathrm{eff}}[k, l]\sum_{k_0, l_0}\mathbf{X}[k_0, l_0] \times \nonumber \\
    &\hspace{5mm}\sum_{k_1, l_1}\hspace{-1mm}\mathbf{X}^*[k_1, l_1]\sum_{n=0}^{MN-1}\mathbf{x}^{(\mathrm{p})}_{k_0, l_0}[n-k]e^{\frac{j2\pi}{MN}l(n-k)} \times \nonumber \\
    &\hspace{5mm}(\mathbf{x}^{(\mathrm{p})}_{k_1, l_1}[n-k'])^*e^{-\frac{j2\pi}{MN}l'(n-k')},
    \label{eq:dif_det3}
\end{align}
where step $(a)$ follows because the pulsone samples and noise samples are uncorrelated.
Substituting $\bar{n} = n-k$ in \eqref{eq:dif_det3}:
\begin{align}
    \mathbf{A}_{\mathbf{y}, \mathbf{x}}[k', l'] &= \frac{1}{MN}\sum_{k, l}\mathbf{h}_{\mathrm{eff}}[k, l]\sum_{k_0, l_0}\mathbf{X}[k_0, l_0] \times \nonumber \\
    &\hspace{5mm}\sum_{k_1, l_1}\mathbf{X}^*[k_1, l_1]\sum_{\bar{n}=-k}^{MN-1-k}\mathbf{x}^{(\mathrm{p})}_{k_0, l_0}[\bar{n}]e^{\frac{j2\pi}{MN}l\bar{n}} \times \nonumber \\
    &\hspace{5mm}(\mathbf{x}^{(\mathrm{p})}_{k_1, l_1}[\bar{n}-(k'-k)])^*e^{-\frac{j2\pi}{MN}l'(\bar{n}-(k'-k))} \nonumber \\
    &= \sum_{k, l}\mathbf{h}_{\mathrm{eff}}[k, l]\sum_{k_0, l_0}\mathbf{X}[k_0, l_0]\sum_{k_1, l_1}\mathbf{X}^*[k_1, l_1] \times \nonumber \\
    &\hspace{5mm}\mathbf{A}_{\mathbf{x}^{(\mathrm{p})}_{k_0, l_0}, \mathbf{x}^{(\mathrm{p})}_{k_1, l_1}}[k'-k, l'-l]e^{\frac{j2\pi}{MN}l'(k'-k)} \nonumber \\
    &= \mathbf{h}_{\mathrm{eff}}[k', l'] \ \ *_\sigma \nonumber \\
    &\hspace{-5mm}\underbrace{\left(\sum_{k_0, l_0}\mathbf{X}[k_0, l_0]\sum_{k_1, l_1}\mathbf{X}^*[k_1, l_1] \mathbf{A}_{\mathbf{x}^{(\mathrm{p})}_{k_0, l_0}, \mathbf{x}^{(\mathrm{p})}_{k_1, l_1}}[k', l']\right)}_{\mathbf{B}[k', l']}.
    \label{eq:dif_det4}
\end{align}
Since $\mathbf{A}_{\mathbf{x}^{(\mathrm{p})}_{k_0, l_0}, \mathbf{x}^{(\mathrm{p})}_{k_1, l_1}}[k', l'] = \mathbf{A}_{\mathbf{X}_{k_0, l_0}^{(\mathrm{p})}, \mathbf{X}_{k_1, l_1}^{(\mathrm{p})}}[k, l]$ (i.e., the time-domain ambiguity is same as the DD domain ambiguity), substituting \eqref{eq:diff_det_prelim} in $\mathbf{B}[k', l']$, we have:
\begin{align}
    \mathbf{B}[k, l] &= \sum_{k_0, l_0}\mathbf{X}[k_0, l_0]\sum_{k_1, l_1}\mathbf{X}^*[k_1, l_1] \times \nonumber \\
    &\hspace{5mm}\mathbf{A}_{\mathbf{X}^{(\mathrm{p})},\mathbf{X}^{(\mathrm{p})}}[k-(k_0-k_1), l-(l_0-l_1)].
    \label{eq:dif_det5}
\end{align}
Note $\mathbf{A}_{\mathbf{X}^{(\mathrm{p})},\mathbf{X}^{(\mathrm{p})}}[k-(k_0-k_1), l-(l_0-l_1)] = \mathds{1}_{\{k=(k_0-k_1)\bmod M, l=(l_0-l_1)\bmod N\}}$. Substituting $k_1 = (k-k_0)_M$ and $l_1 = (l-l_0)_N$ in \eqref{eq:dif_det5}:
\begin{align}
    \mathbf{B}[k, l] &= \sum_{k_0, l_0}\mathbf{X}[k_0, l_0]\mathbf{X}^*[(k-k_0)_M, (l-l_0)_N].
    \label{eq:dif_det6}
\end{align}
This implies that the cross-ambiguity between the transmitted data frame and the received frame is the twisted convolution between the effective channel $\mathbf{h}_{\mathrm{eff}}[k, l]$ and the inner product between the transmitted data and its delay and Doppler shifted version. If the information symbols are chosen uniformly at random from a constellation, asymptotically we have\footnote{The asymptotics are only in the frame size and not in the constellation size. The frame size must be chosen sufficiently large so as to enable the law of large numbers (see \cite{diff_det_arxiv} for more details).}:
\begin{align}
    \mathbf{A}_{\mathbf{y}, \mathbf{x}}[k', l'] \approx \mathbf{h}_{\mathrm{eff}}[k', l'] *_\sigma e_{\mathrm{d}}\delta[k]\delta[l] = e_{\mathrm{d}}\mathbf{h}_{\mathrm{eff}}[k', l'],
    \label{eq:dif_det7}
\end{align}
where $e_\mathrm{d}$ is the energy of information symbols. The cross-ambiguity between the received and transmitted data symbols is therefore approximately equal to the estimate of the effective channel, up to a scale.

\subsection{Numerical Results}
\begin{figure}
    \centering
    \includegraphics[width=\linewidth]{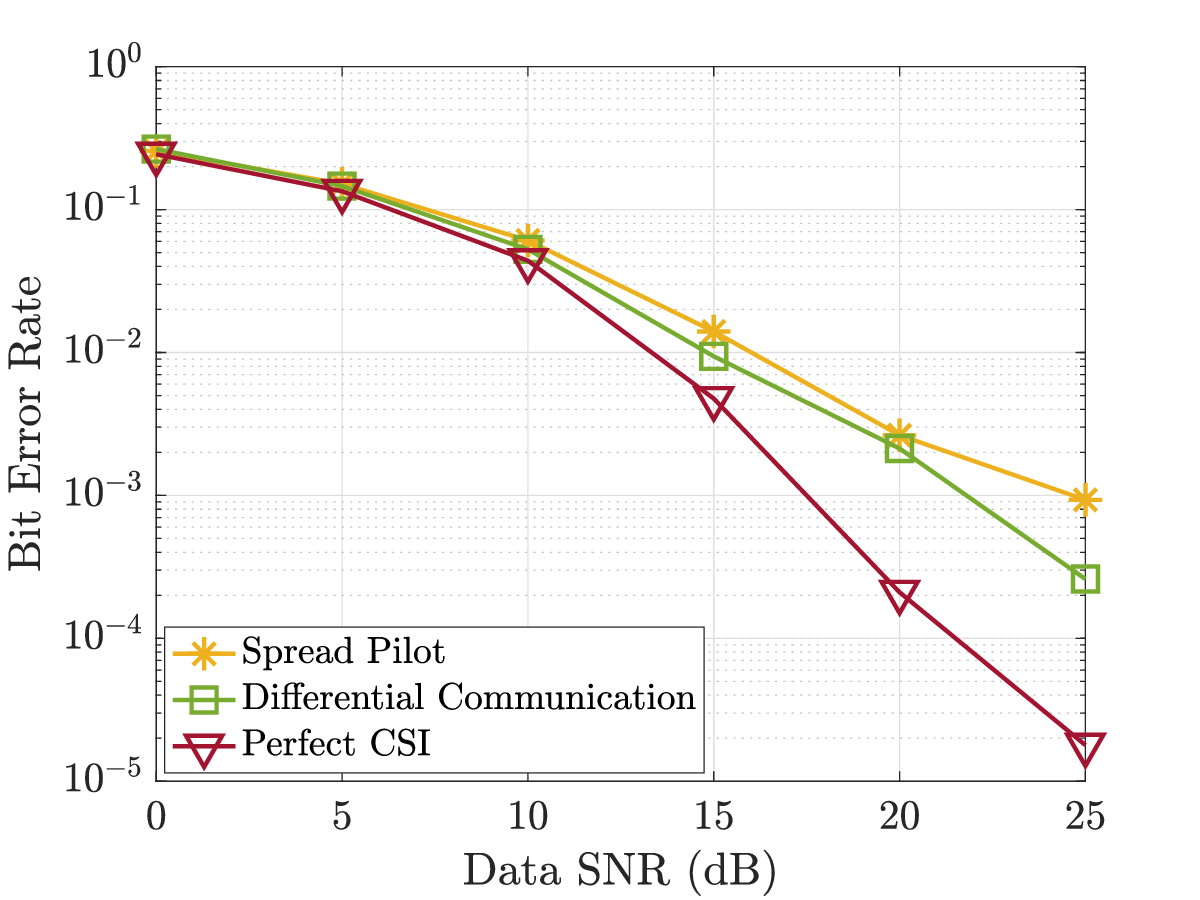}
    \caption{Performance curves with differential communication scheme, spread pilot~\cite{Calderbank2025_isac}, and perfect CSI for 4-QAM modulation. Zak-OTFS with $M=31, N=37, \nu_p=30 $ kHz, RRC pulse shaping with parameter $\beta_\tau=\beta_\nu=0.6$.}
    \label{fig:ber_diff_comm}
\end{figure}
Here we present the bit-error performance of the differential communication scheme. To prevent error propagation, a pilot frame is transmitted every $30$ data frames. In Fig.~\ref{fig:ber_diff_comm}, we compare the bit-error performance of the differential communication scheme with that of spread pilot scheme~\cite{Calderbank2025_isac} and perfect channel knowledge at the receiver. Note that, both the differential communication scheme and spread pilot scheme achieve full spectral efficiency. For fair comparison, we assume same frame energy for both spread pilot and the data frame for differential communication and so the data frame in the differential communication scheme enjoys higher energy (equal to the sum of pilot and data energy in the spread pilot frame). It is seen that the performance of the differential communication scheme is very close to that with perfect channel knowledge and this is followed by the spread pilot scheme. On the receiver complexity, the differential communication scheme just requires cross-ambiguity computations while the spread pilot scheme requires cross-ambiguity followed by pilot removal and data detection. Differential communication scheme achieves better performance at lower complexity compared to the spread pilot scheme.

\subsection{Open Research Problems}
\label{subsec:interframe_future}

In the above discussion, we described a method for reducing the pilot retransmission interval by transmitting a pilot frame at a fixed interval of $30$ data frames. An interesting research direction is to understand the trade-offs between communication performance, sensing / DD channel estimation performance \& energy efficiency, and to design / adapt the pilot retransmission interval to optimally balance these three metrics. An interesting open question is understanding the conditions when differential communication is possible, and understanding relevant trade-offs with various waveform choices, extension to multi-antenna and multi-user systems, interplay with coding, etc.

\section{Uplink Multiple Access}
\label{sec:multipleaccess}
The  3G, 4G, and 5G standards were defined by the challenge of engineering the wireless downlink to support internet browsing and video streaming. Today, cellphones are data sinks, but if new applications like autonomous driving take root, terminals might tomorrow become data sources. An emerging challenge in 6G is supporting a massive number of devices that sporadically and intermittently transmit short data packets to a central base station, under stringent conditions on energy consumption and signaling overhead, and often under mild latency and reliability constraints \cite{uplinkMA_1, uplinkMA_2, uplinkMA_3}. This challenge is referred to as massive machine-type communications (mMTC).  
\begin{figure*}
    \centering
    \includegraphics[width=\linewidth]{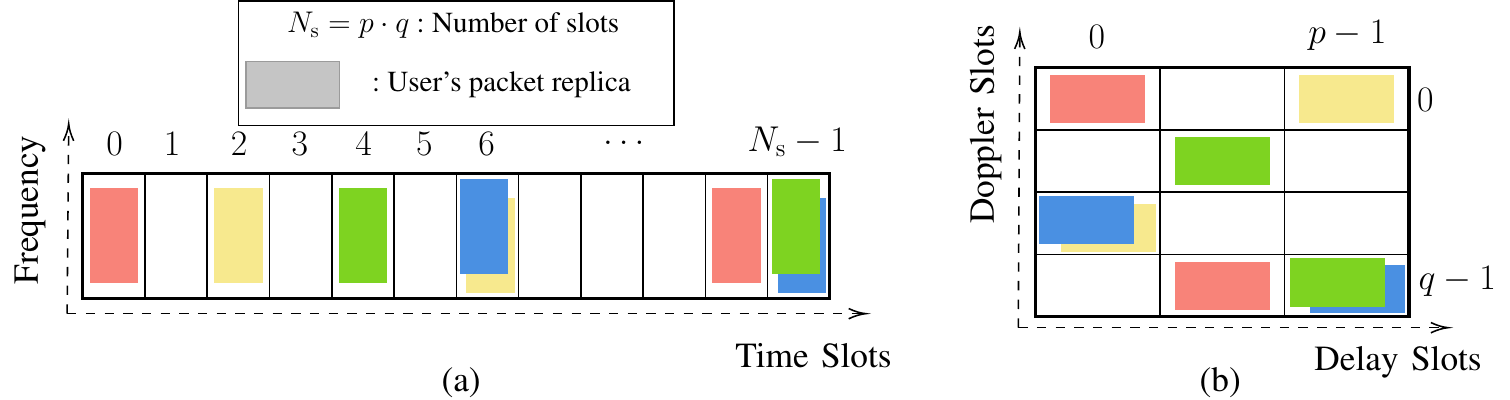}
    \caption{Comparison between (a) conventional coded random access in the time-frequency domain and (b) proposed Zak-OTFS-based coded random access in the delay-Doppler domain. In both cases, colored rectangles represent user-specific packet replicas transmitted in randomly selected slots within the corresponding resource grid.}
    \label{fig:uplinkMA_fig1}
\end{figure*}
\subsection{Coded Random Access}
This section describes how to achieve massive multiple access \cite{uplinkMA_4, uplinkMA_5} while enabling devices to transmit autonomously, whenever new data arrives, without requiring explicit scheduling by the base station. This approach, termed grant-free random access \cite{uplinkMA_6, uplinkMA_7, uplinkMA_8}, reduces latency and protocol complexity but requires careful integration of the medium access control (MAC) and physical (PHY) layers since simultaneous transmission may result in packet collisions. We employ coded random access (CRA) where a user uploads replicas of a data packet into multiple uplink slots \cite{uplinkMA_9}. When a packet is detected, we use successive interference cancellation (SIC) to remove it from the other slots where it is present. Zak-OTFS provides accurate channel prediction in the presence of mobility and delay spread, and predictability of the physical layer is essential if SIC is to be successful. Furthermore, as mentioned in the previous section, the channel in the DD domain varies as allowed by the physics of the reflectors. This implies that DD domain channel estimates remain stationary for a longer duration compared to the time-frequency estimates in the current OFDM standard. For grant-free mMTC, this inherent stability is foundational, it not only drastically reduces pilot overhead but also enables reliable SIC across packet replicas, a critical collision-resolution process that fails under OFDM due to rapid inter-slot channel variations. This predictability motivates the pilot frame shown in Fig.~\ref{fig:uplinkfig2}, wherein the channel does not change significantly between the pilot and subsequent data frames.

Unsourced random access is a theoretical framework proposed by Polyanskiy \cite{uplinkMA_10} to model grant-free random access for the power-constrained AWGN channel. He derived an achievable probability of error for random coding as a function of signal-to-noise ratio that improved dramatically on the performance of naive random-access schemes such as slotted ALOHA (SA) \cite{uplinkMA_11}. Significant improvements to the SA benchmark, particularly for a large number of users, were then obtained by tailoring CRA schemes to the AWGN channel \cite{uplinkMA_12, uplinkMA_13}. Subsequent work has extended CRA schemes to flat Rayleigh fading channels \cite{uplinkMA_14, uplinkMA_15}, and in this Section we illustrate what is possible for CRA in wireless channels that exhibit significant delay and Doppler spread. 

SIC-based architectures rely on accurate estimation of the input/output (I/O) relation over multiple transmission slots. This presents a challenge for OFDM and single carrier (SC) modulation because the effective channel response varies rapidly between slots and prior knowledge of the channel is needed to acquire the I/O relation. The essential difficulty is that of accurately estimating the channel from a very limited number of pilot symbols. The advantage of a predictable waveform such as Zak-OTFS is that accurate channel estimation across slots enables SIC across user packet replicas. We now describe a CRA scheme based on Zak-OTFS and compare performance under high mobility and user density with an OFDM-based CRA baseline (see \cite{uplinkMA_16} for more details). 

\subsection{MAC Layer, Random Access and SIC}
The SA random-access protocol partitions time into slots and active users independently attempt transmission of a single information packet without any coordination. Analysis of SA assumes the collision channel model in which a packet is successfully received only if it is the sole transmission in a given slot, and all colliding packets are lost. This model also assumes that the BS detects all collisions and provides feedback to all users at the end of each slot indicating whether a packet was successfully received or whether retransmission is required in subsequent time slots. SA is a simple and fully decentralized protocol, but the probability of collision grows with the number of active users, significantly reducing throughput and reliability.

More advanced random-access protocols employ slot diversity where users transmit multiple replicas of the same packet in multiple time slots. Diversity slotted ALOHA (DSA) \cite{uplinkMA_17} improved upon the reliability of SA by increasing the probability that at least one replica avoids a collision and is successfully received. 

The introduction of SIC made it possible to resolve packet collisions by iteratively decoding packets and subtracting their reconstructed contributions from received signals. SIC is the foundation of modern CRA protocols such as contention resolution diversity slotted ALOHA (CRDSA) \cite{uplinkMA_6} and irregular repetition slotted ALOHA (IRSA) \cite{uplinkMA_7}. In these schemes, each user transmits a predefined \cite{uplinkMA_6} or randomly selected \cite{uplinkMA_7} number of replicas in independently chosen time slots within a contention window as illustrated in Fig. \ref{fig:uplinkMA_fig1} (taken from \cite{uplinkMA_16}). The access process is modeled as a sparse bipartite graph and the decoding techniques are inspired by belief-propagation methods originally developed for erasure codes.

The single slot structures for cyclic prefix (CP)-OFDM and Zak-OTFS are illustrated in Fig. \ref{fig:uplinkfig2}. The time-bandwidth products of the corresponding CRA schemes are equal, and active users are assigned equivalent physical resources (see \cite{uplinkMA_16} for details). 

Each active Zak-OTFS user independently selects a fixed number of non-overlapping slots with uniform probability. Coherent SIC relies on consistency of transmitted replicas, so that in each slot, the user transmits the same pilot tile and the same modulated data symbols in the corresponding data tile. For comparison, we select cyclic prefix (CP)-OFDM with pilot insertion as a benchmark scheme. To provide a fair comparison, the CP-OFDM system is configured with a subcarrier spacing of $\Delta f = \nu_p$, a CP length matched to the maximum delay spread of the channel, and to enhance robustness against Doppler effects, pilot and data symbols are interleaved across alternating subcarriers.

Decoding in both CP-OFDM and Zak-OTFS schemes begins after all the signals from all slots have been received. The BS first identifies the singleton slots that are free from collisions, then uses the pilot symbols to perform channel estimation. In this system, MMSE equalization is employed for data detection. For OFDM, MMSE-based channel interpolation is additionally used to estimate the channel in slots where replicas occur, enabling successive interference cancellation. This is necessary because the OFDM channel varies across sub-frames and even between slots \cite{uplinkMA_18}. In contrast, for Zak-OTFS, the channel is approximately constant over a frame; hence, the channel estimated from one slot can be directly reused to cancel interference in other slots without requiring interpolation.  Decoding continues until all user data have been decoded or no further singleton slots can be identified. This alignment of decoding processes compares the two schemes on their ability to support channel prediction across data carriers.

\begin{figure*}
    \centering
    \includegraphics[width=0.75\linewidth]{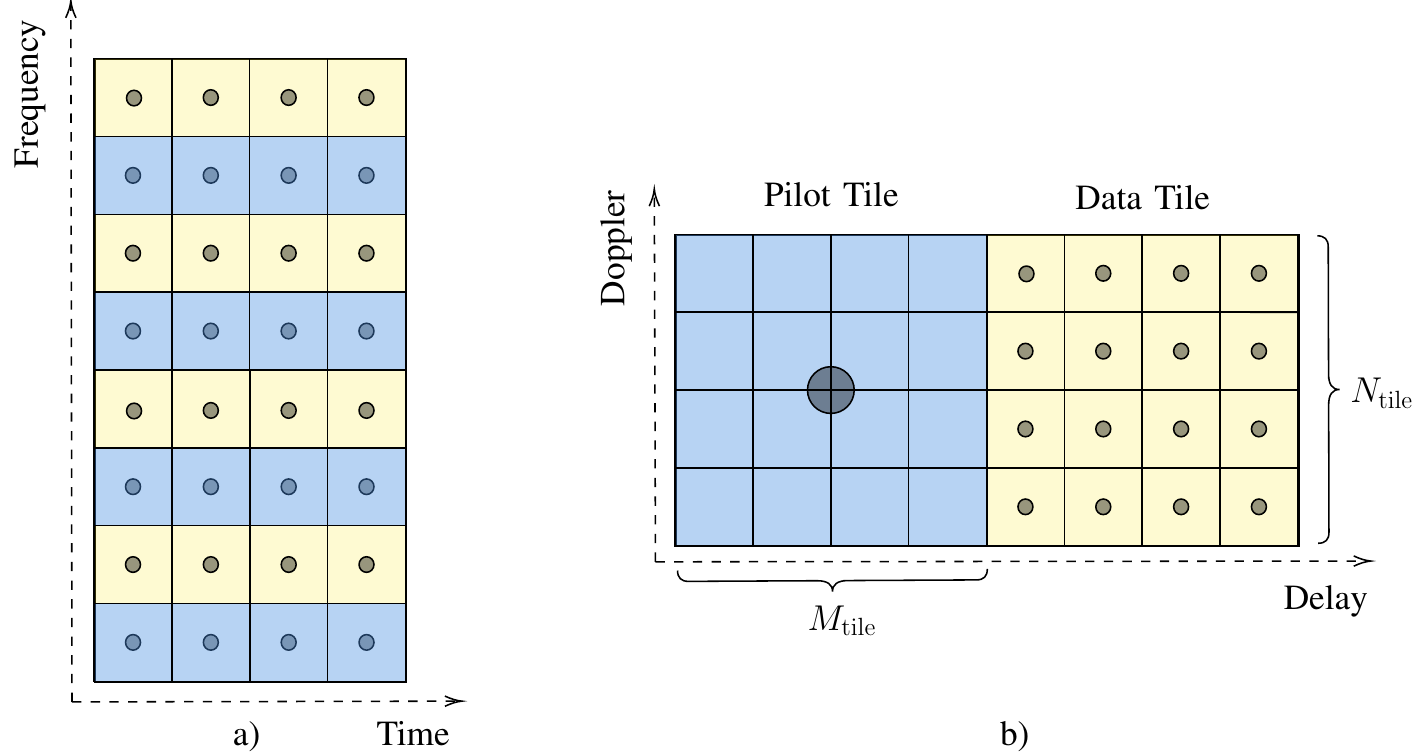}
    \caption{Single slot structures for OFDM and Zak-OTFS in the contention window illustrated in Fig. 1. (a) Time-frequency domain representation of the OFDM slot, where pilot and data symbols are spread across time and frequency resources, following a structure typical of OFDM-based transmissions. (b) Delay-Doppler domain representation of the Zak-OTFS slot, comprising two tiles: a pilot tile (blue) and a data tile (yellow). The pilot tile includes a single point pilot symbol centered in the grid (large circle), while the data tile carries modulated data symbols (small circles).}
    \label{fig:uplinkfig2}
\end{figure*}

\subsection{Performance Evaluation}
We first evaluate performance from the perspective of a single user transmitting in a single slot without interference from other users. We seek to understand the influence of key system parameters such as the number of delay / Doppler bins ($M_{tile} / N_{tile}$), the Doppler period $\nu_p$, and the pulse shaping filter. We consider both small slots ($M_{tile}=4, N_{tile}=4$) and larger slots ($M_{tile}=16, N_{tile}=16$), and for each configuration we employ a combination of BCH code and cyclic redundancy check with code rate approximately 0.5. We match the time and frequency span of both schemes by setting the CP-OFDM carrier spacing $\Delta f = \nu_p$, and the OFDM symbol duration to $\tau_p$. We also match the energy budgets of the pilot and data signals.
The small slot configuration corresponds to $M_{slot}=8, N_{slot}=4$ for Zak-OTFS and to $M_{sub}=8$ subcarriers and $N_{OFDM}=4$ symbols for CP-OFDM. When the Doppler period aligns with the subcarrier spacing ($\nu_p = 30$ kHz) Zak-OTFS exhibits a PLR error floor between $10^{-1}$ and $10^{-2}$ regardless of the choice of pulse shaping filter. CP-OFDM exhibits superior performance because it benefits from the coarse subcarrier spacing. When the Doppler spread is reduced to 5 kHz, CP-OFDM performance degrades significantly, and there is a PLR error floor between $10^{-1}$ and $10^{-2}$ because reduced subcarrier spacing leads to significant ICI.
The larger slot configuration corresponds to slot sizes $M_{slot}=32$, $N_{slot}=16$ for Zak-OTFS. Both Zak-OTFS and CP-OFDM benefit from the improved resolution, with CP-OFDM achieving extremely low PLR for SNRs less than -5 dB. Zak-OTFS with sinc pulse shaping now exhibits a PLR error floor of about $10^{-2}$. The PLR curve for Zak-OTFS with Gaussian pulse shaping has approximately the same shape as the PLR curve for CP-OFDM but is about 10 dB worse (for more details refer Fig. 5 and its associated text in \cite{uplinkMA_16}).
We next evaluate how rapidly time-varying channels impact SIC for both Zak-OTFS and CP-OFDM. We consider a transmission frame with two active users, an uncollided user and a collided user. The uncollided user is assumed to transmit one replica without interference and at least one replica with interference from another user (for example, the yellow user in Fig. \ref{fig:uplinkMA_fig1} (a)), The collided user is exemplified by the blue user in Fig. \ref{fig:uplinkMA_fig1} (a). Every replica is subject to interference and data can only be recovered by SIC after decoding and subtracting the contributions of interfering packets. If the yellow user in Fig. \ref{fig:uplinkMA_fig1}(a) is decoded in slot 2 its channel must be predicted in slot 6 to enable SIC. 
Predictability of Zak-OTFS results in a channel that is almost constant across the entire frame. Once a user is successfully decoded the contributions of replicas can be accurately subtracted from slots where they are present. By contrast, CP-OFDM channel prediction becomes less reliable the further a replica lies from the slot where the channel was estimated, and this time selectivity significantly impacts SIC performance and PLR. In the large slot configuration, a maximum Doppler frequency of 815 Hz is sufficient to make channel prediction unreliable, even in an adjacent slot, completely disabling SIC. Fig. \ref{fig:UplinkMA_fig3} illustrates PLR as a function of the number of active users per frame. Performance of CP-OFDM is reported under the assumption SIC is ineffective and is therefore not applied. By contrast, CRA with Zak-OTFS exhibits strong robustness to high mobility, enabling effective SIC and successful recovery of significantly more data packets. Zak-OTFS with a Gaussian pulse shaping filter is able to support about 60 users with a PLR of about $10^{-3}$. Reduced accuracy of Zak-OTFS channel estimation with a sinc filter results in a PLR error floor of about $10^{-2}$ and a visible performance gap. The methods developed by Polyanskiy \cite{uplinkMA_10} to support grant-free Gaussian random access would use the same resources to support about 200 users (note that the 6-path channel converts a single user into 6 correlated users).
\begin{figure}[htb!]
    \centering
    \includegraphics[width=\linewidth]{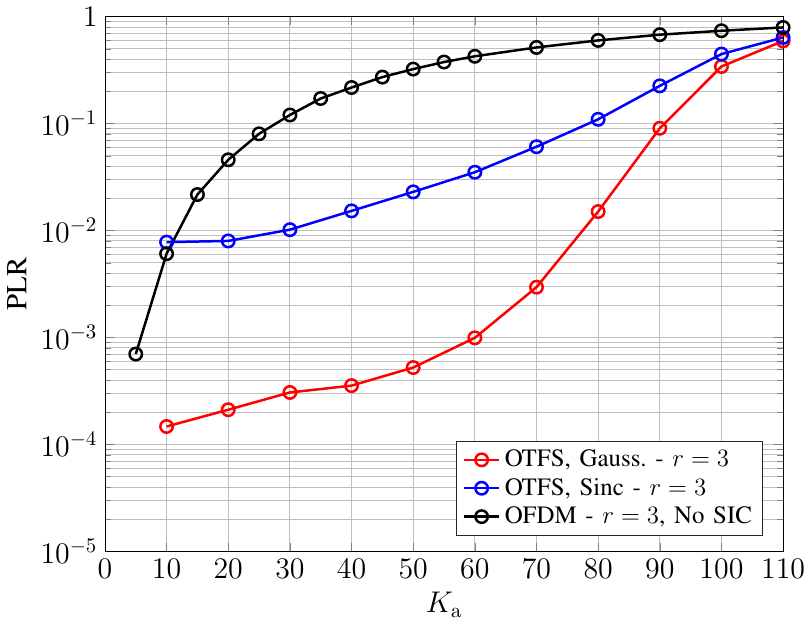}
    \caption{Packet loss rate (PLR) versus $K_a$, the number of active users per frame in a full-frame MMA uplink scenario. There are 128 slots, each slot occupies $M_{slot} = 32$ resources along the delay axis and $N_{slot} = 16$ resources along the Doppler axis, and the Doppler period $\nu_p = 30$ kHz. Each user transmits $r = 3$ replicas of the same packet (as is often assumed in the CRA literature \cite{uplinkMA_19, uplinkMA_20}). The SNR is fixed to 25 dB. Zak-OTFS-based CRA performance is shown for Gaussian and sinc pulse shaping, and compared against the OFDM-based scheme with no SIC.}
    \label{fig:UplinkMA_fig3}
\end{figure}

\subsection{Open Research Problems}
\label{subsec:uplink_future}
Is it possible to design coded random-access schemes for uplink mMTC using predictable modulations other than Zak-OTFS? Is it possible to improve CRA performance through joint detection of preambles and data repeated across multiple slots? Is it possible to estimate channels and separate users even in collided slots through deployment of a massive MIMO base station which provides additional spatial degrees of freedom? How do alternative methods of unsourced random access, such as coded compressed sensing, extend to doubly selective channels? How might they be integrated with Zak-OTFS? How can generative AI based channel estimation be tailored for random-access with doubly selective channels? How can semi-blind channel estimation algorithms be tailored to doubly selective channels?

\section{Improving Spectral Efficiency}
\label{sec:imp_spec_eff}
\begin{figure}
    \centering
    \begin{subfigure}{\linewidth}
        \includegraphics[width=\linewidth]{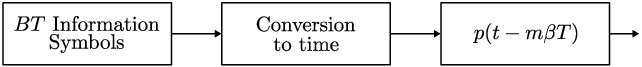}
        \caption{Conventional (Faster-than-Nyquist signaling)}
        \label{fig:conv_approach}
    \end{subfigure}\\
    \vspace{2mm}
    \begin{subfigure}{\linewidth}
        \includegraphics[width=\linewidth]{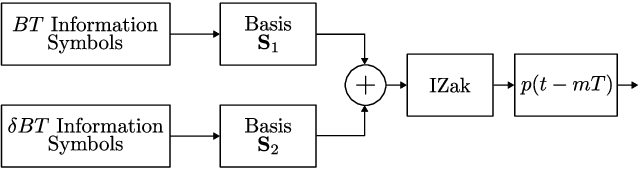}
        \caption{MUB}
        \label{fig:mub_approach}
    \end{subfigure}
    \caption{Two ways of improving spectral efficiency. Conventional scheme transmits $BT$ information symbols in the TF domain and improves spectral efficiency by increasing the sampling frequency (captured by ${0 < \beta < 1}$) in the time domain, while the MUB approach uses mutually unbiased bases in the DD domain to mount $BT + \delta BT = BT(1+\delta), 0<\delta<1$ information symbols, which requires no change in the sampling rate ($\beta = 1$) in the time domain. Setting $\nicefrac{1}{\beta} = 1+\delta$ equates spectral efficiency of the two approaches.}
    \label{fig:ftn_comp}
\end{figure}

Data transmission in bandwidth-limited channels, which are characteristic of wireless and wireline communication systems, is typically performed at the Nyquist rate. Formally, for a channel with bandwidth $B$, the number of independent data symbols that can be transmitted at the Nyquist rate in a time interval $T$ is $BT$~\cite{Tse2005}. This is achieved by mounting $BT$ independent information symbols on time-frequency pulses with the property that shifts in time of $\nicefrac{1}{B}$ (or its multiples) and shifts in frequency of $\nicefrac{1}{T}$ (or its multiples) are orthogonal to one another. In other words, the information symbols are mounted on a $BT$-dimensional basis which spans the $BT$ space. This concept of Nyquist signaling is also easily extended to communication in signal domains other than time-frequency, e.g., to the DD domain, as exemplified by Zak transform based orthogonal time frequency space, or Zak-OTFS~\cite{bitspaper1,bitspaper2}, by defining pulses that are orthogonal to shifts in delay of $\nicefrac{1}{B}$ and shifts in Doppler of $\nicefrac{1}{T}$ (or their multiples). The idea remains to mount $BT$ information symbols on an orthogonal basis defined in the appropriate domain.

Although the maximum number of orthogonal information signals that are possible in a $BT$-dimensional space is $BT$, more than $BT$ information symbols can be transmitted if the interference from non-orthogonal signaling can be handled at the transmitter/receiver. This forms the basis for faster-than-Nyquist signaling~\cite{mazo_ftn, liveris_ftn}. Most of the schemes achieving higher spectral efficiency~\cite{paul2025cholesky,rusek_ftn, calavolpe_ftn_tf_2lattice_poweralloc_turbo, zhang_ftn_tf_hexalattice_isioptmzn} use the setup shown in Fig.~\ref{fig:ftn_comp}(\subref{fig:conv_approach}). In this approach, the information symbols in the transform domain are converted to time domain and mounted on a band-limited pulse which is sampled at $\beta T, 0<\beta<1$, where $T$ is the Nyquist sampling interval. This requires changing the sampling frequency and possibly the sampling clock. Changing the sampling frequency introduces ISI.

Here we take the mutually unbiased bases (MUB)~\cite{Schwinger1960unitary} approach shown in Fig~\ref{fig:ftn_comp}(\subref{fig:mub_approach})~\cite{mattu2026improving}. Two bases are said to be MUB if the interference from one basis looks like Gaussian noise to the other\footnote{For an in-depth discussion on MUBs and their construction, interested readers are referred to \cite{Aug2024paper, preamblepaper}.\label{foot:mub}}. We construct two bases $\mathbf{S}_1$ and $\mathbf{S}_2$ so that they are mutually unbiased (see~\eqref{eq:mut_unb_basis}). This construction is carried out in the DD domain~\cite{preamblepaper} using Zak-transform based orthogonal time frequency space (Zak-OTFS) modulation~\cite{bitspaper1, bitspaper2}. $BT$ information symbols are mounted on the first basis and $\delta BT, 0 < \delta < 1$ symbols are mounted on the second basis. The modulated bases are then superimposed. The inverse Zak transform block converts the DD symbols to time which is mounted on a pulse sampled at the \textit{Nyquist rate}. Effectively, in MUB, the interference is moved from faster sampling of the transmit pulse to the interference between the two bases.

In the following, we describe our approach in detail.

Recall from~\eqref{eq:prec3} the system model after precoding is: $$ \mathbf{y} = \mathbf{R}^{\mathsf{H}}\mathbf{x} + \mathbf{n}.$$ At the receiver, an MMSE matrix is constructed as:
\begin{align}
    \label{eq:precoder4}
    \mathbf{W} = (\mathbf{R}\mathbf{R}^{\mathsf{H}}+\sigma^2\mathbf{I})^{-1}\mathbf{R},
\end{align}
where $\sigma^2$ is the variance of the additive noise. Finally, the receiver computes:
\begin{align}
    \label{eq:precoder5}
    \mathbf{y}' = \mathbf{W}\mathbf{y}.
\end{align}

\subsection{Mounting Info. Symbols on Mutually Unbiased Bases}
\label{subsec:mounting_symbols}

Let $\mathbf{x}_1', \mathbf{x_2}' \in \mathcal{A}^{MN \times 1}$ be two information vectors. Let $\mathbf{S}_1, \mathbf{S}_2 \in \mathbb{C}^{MN \times MN}$ denote two sets of bases. The bases each span the $MN$-dimensional space with the property that they are unbiased with respect to each other. Mathematically, the bases $\mathbf{S_1}, \mathbf{S_2}$ satisfy:
\begin{align}
    \label{eq:mut_unb_basis}
    \vert \mathbf{S}_i^{\mathsf{H}}\mathbf{S}_j\vert = \begin{cases}
        \mathbf{I}_{MN\times MN}, \ \ \ \ \ \quad \quad \text{if } i=j \\
        \frac{1}{\sqrt{MN}}\mathbf{1}_{MN \times MN}, \quad \text{if } i\neq j
    \end{cases}.
\end{align}

We mount $2MN$ information symbols on the Nyquist-grid\footnote{Here, by Nyquist grid, we mean the original OTFS grid, which has $MN$ degrees of freedom, with bandwidth $B=M\nu_p$ and time $T=N\tau_p$.}:
\begin{align}
    \label{eq:ftn_sys_model}
    \mathbf{y}' = \mathbf{W}\mathbf{H}\big(\sqrt{\alpha}\mathbf{Q}\mathbf{S}_1\mathbf{x}_1' + \sqrt{1-\alpha}\mathbf{Q}\mathbf{S}_2\mathbf{x}_2'\big) + \mathbf{W}\mathbf{n},
\end{align}
where 
the parameter $\alpha$ allocates power between the two frames\footnote{Energy in the first frame is $\alpha$ and in the second frame is $1-\alpha$, since $\mathbf{x}$ is drawn from a unit energy constellation. The total transmit energy is still unity.}. Notice that using the system model as defined above \textit{does not} incur energy, time, or bandwidth expansion.
Using the fact that $\mathbf{W}\mathbf{H}\mathbf{Q} \approx \mathbf{I}_{MN \times MN}$\footnote{$\mathbf{W}\mathbf{H}\mathbf{Q} = \mathbf{W}\mathbf{R}^{\mathsf{H}}\mathbf{Q}^{\mathsf{H}}\mathbf{Q} = (\mathbf{RR}^\mathsf{H} + \sigma^2\mathbf{I})^{-1}\mathbf{R}\mathbf{R}^{\mathsf{H}} \approx \mathbf{I}$, for small $\sigma^2$.\label{foot:mmse}}, the input-output relation can be approximated as:
\begin{align}
    \label{eq:ftn_pre_sys_model}
    \mathbf{y}' \approx \sqrt{\alpha}\mathbf{S}_1\mathbf{x}_1' + \sqrt{1-\alpha}\mathbf{S}_2\mathbf{x}_2' + \mathbf{n}'.
\end{align}
Note that we mount the information symbols on the mutually unbiased bases and information symbols used in $\mathbf{x}_1$ and $\mathbf{x}_2$ need not be from a uniquely decodable constellation set~\cite{harshan2011two}.

\subsection{Detection of Info. Symbols}
\label{subsec:detection_symbols}
At the receiver, the vector $\mathbf{y} \in \mathbb{C}^{MN \times 1}$ is received. The received $\mathbf{y}$ is combined to get $\mathbf{y}'$ (see~\eqref{eq:precoder5}). To detect $\mathbf{x}_1'$ and $\mathbf{x}_2'$ from $\mathbf{y}'$, we proceed as follows.
For detecting the $i$th frame ($i \in \{1, 2\}$), we perform pre-multiplication by the complex conjugate transpose of the basis matrix $\mathbf{S}_{i}$:
\begin{align}
    \label{eq:first_frame_eq}
    \mathbf{S}_{i}^{\mathsf{H}}\mathbf{y}' = \beta_{i}\mathbf{x}_{i}' + \tilde{\mathbf{n}}_{i},~i \in \{1,2\},
\end{align}
where $\tilde{\mathbf{n}}_{i} = \mathbf{S}_{i}^{\mathsf{H}}(\beta_{j}\mathbf{S}_{j}\mathbf{x}_{j}' + \mathbf{n}'), j \in \{2, 1\}$ is the resulting noise for $\beta_{1} = \sqrt{\alpha},~\beta_{2} = \sqrt{1-\alpha}$. Since matrices $\mathbf{S}_i, i=1, 2$ are chosen to be orthonormal, there is no noise enhancement. To recover information symbols $\mathbf{x}_i' \in \mathbb{C}^{MN \times 1}$, we perform and element-wise detection in Gaussian channel as:
\begin{align}
    \label{eq:first_frame_det}
    \hat{\mathbf{x}}'_i[q] = \underset{s \in \mathcal{A}}{\arg\min} \Vert \mathbf{S}_i^{\mathsf{H}}\mathbf{y}'[q] - \beta_is\Vert^2,
\end{align}
where $q = 0, 1, \cdots, MN-1$. Next, for detection of second frame, we cancel the contribution of the detected frame from the received frame as:
\begin{align}
    \label{eq:first_frame_rem}
    \bar{\mathbf{y}} = \mathbf{y}' - \beta_i\mathbf{S}_i\hat{\mathbf{x}}_i'.
\end{align}
To detect the second frame, the same Gaussian detection is performed as:
\begin{align}
    \label{eq:second_frame_det}
    \hat{\mathbf{x}}'_j[q] = \underset{s \in \mathcal{A}}{\arg\min} \Vert \mathbf{S}_j^{\mathsf{H}}\bar{\mathbf{y}}[q] - \beta_js\Vert^2,
\end{align}
where $q = 0, 1, \cdots, MN-1$.

\subsection{Trellis Coded Modulation}
\label{subsec:tcm}
From \eqref{eq:first_frame_eq} it is clear that the noise floor increases, since the effective noise is the sum of the additive Gaussian noise and the term $\mathbf{S}_i^{\mathsf{H}}\beta_j\mathbf{S}_j\mathbf{x}_j'$. This decreases the effective signal to interference plus noise ratio. The detector described in \eqref{eq:first_frame_det} or \eqref{eq:second_frame_det} therefore would have poor bit-error performance. To lift the information symbols well above noise level we use coding, specifically TCM, which ensures that the coded symbols do not create additional interference to the other frame, i.e., the number of uncoded symbols and coded symbols is the same albeit the coded symbols come from a higher constellation alphabet. The detection of information symbols follows that described in Sec. \ref{subsec:detection_symbols} with the only change being that instead of using detection in Gaussian channel as shown in \eqref{eq:first_frame_det}, we use a Viterbi decoder. 

\subsection{Turbo Iterations}
\label{subsec:turbo}
The first frame is decoded in the presence of the second frame. Subsequently, the decoded frame is re-encoded and subtracted from the received frame before the second frame is decoded. The performance can be improved if more rounds of interference cancellation and decoding are performed. This is termed as turbo iterations. Specifically, the decoded second frame is re-encoded and removed from the received frame and the first frame is decoded again. The decoded first frame is re-encoded and subtracted from the received frame before decoding the second frame. This forms one turbo iteration. Many turbo iterations could be performed until the performance improvement between two successive turbo iterations becomes insignificant.

\subsection{Maximizing Effective Rate}
\label{subsec:eff_rate}
In this Subsection, we derive the effective rate (for qualitative insight) as a function of the power distribution between the frames and sparsity of the second frame. Let $P_1 = \alpha P$ denote the energy in the first frame and $P_2 = (1-\alpha)P$ denote the energy in the second frame. Let $\nicefrac{P}{\sigma^2}$ denote the effective signal-to-noise ratio. The signal to interference plus noise ratio for the first frame is:
\begin{align}
    \mathsf{SINR}_1 = \frac{P_1}{\sigma^2 + \frac{\delta P_2}{MN}},
\end{align}
where $\delta \in [0, 1]$ denotes the sparsity of the second frame. The denominator follows from the fact that the second frame acts as interference to the first frame and that the interference is flat (see~\eqref{eq:mut_unb_basis}). Before we detect the second frame, we detect the first frame and subtract it from the received frame. Let $P_{s_1}$ denote the symbol error probability for the first frame. Then the signal to interference plus noise ratio for the second frame is:
\begin{align}
    \mathsf{SINR}_2 = \frac{P_2}{\sigma^2 + \frac{P_1P_{s_1}}{MN}}.
\end{align}
Given the $\mathsf{SINR}_1$, $P_{s_1}$ is given by~\cite{ungerboeck1982channel}:
\begin{align}
    P_{s_1} = Q\Big(\sqrt{2d_{\mathsf{free}}\mathsf{SINR}_1}\Big),
\end{align}
where $Q(\cdot)$ denotes the $Q$-function, $d_{\mathsf{free}}$ is the free Euclidean distance of the TCM trellis. From the Shannon formula, the rate can be written as:
\begin{align} 
    R_1 &= \log_2(1+\mathsf{SINR}_1) \\
    R_2 &= \log_2(1+\mathsf{SINR}_2).
\end{align}
The goal is to now choose $P_1$ and $P_2$, or equivalently $\alpha$, to maximize the effective rate $R_1+\delta R_2$. Fig.~\ref{fig:rate_vs_alpha_vs_delta} shows the plot of effective rate as a function of $\alpha$ and $\delta$. The optimal choice for $\alpha$ is a value close to, but not equal to, $1$\footnote{Here we have chosen the empirically optimal $\alpha$ value. Deriving the theoretically optimal value of $\alpha$ is an important direction of future research.\label{foot:opt_alpha}}.

\begin{figure}
    \centering
    \includegraphics[width=0.8\linewidth]{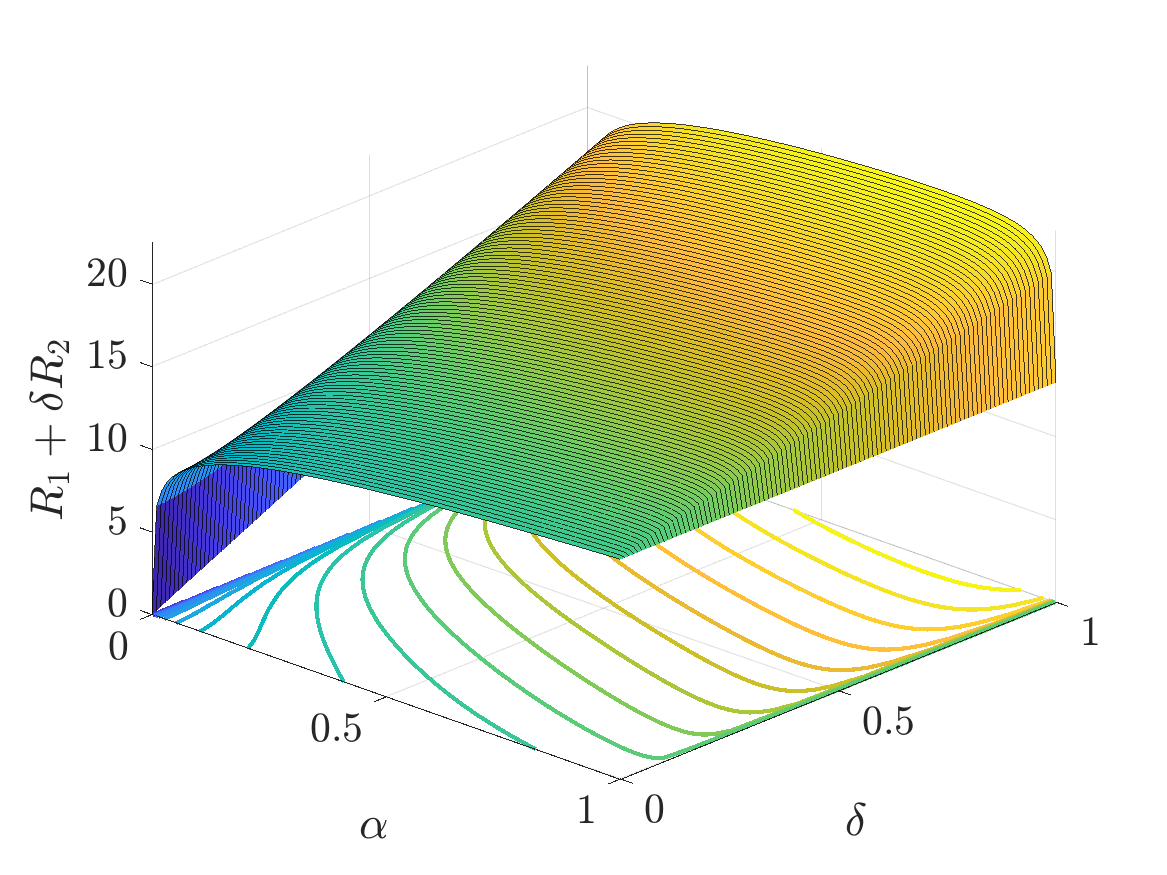}
    \caption{Effective rate as a function of $\alpha$ and $\delta$ for $40$ dB SNR. $M=31, N=37, d_{\mathsf{free}}=\sqrt{20}$. For a given $\delta$, increasing $\alpha$ increases the effective rate, but dips when $\alpha=1$, which corresponds to single frame transmission. For a given $\alpha$, the effective rate increases with $\delta$. The plot is agnostic to domain.}
    \label{fig:rate_vs_alpha_vs_delta}
\end{figure}

\subsection{Spectral Efficiency \& Performance Tradeoff}
\label{subsec:choosing_delta}
The input-output relation is given by~\eqref{eq:ftn_pre_sys_model}. Suppose the detection for frame $\mathbf{x}_1'$ is carried out first. Then, the signal energy is:
\begin{align}
    \label{eq:choose_delta1}
    E_s = \Vert \sqrt{\alpha}\mathbf{S}_x\mathbf{x}_1'\Vert_2^2 = \vert\alpha\vert = \alpha,
\end{align}
since the information symbols are chosen from a unit energy constellation and the basis is orthonormal. The interference plus noise power is:
\begin{align}
    \label{eq:choose_delta2}
    E_{n} = \Vert\sqrt{1-\alpha}\mathbf{S}_2\mathbf{x}_2'+\mathbf{n}'\Vert_2^2 = (1-\alpha)\delta + \sigma^2,
\end{align}
where the last equality follows from the fact that $\mathbf{x}_2'$ is $\delta$-sparse and $\sigma^2$ is the variance of noise. For detection of $\mathbf{x}_1'$, we would require $E_s$ to be greater than $E_n$. Let $\gamma>1$ be a real number, then:
\begin{align}
    \label{eq:choose_delta3}
    E_s \geq \gamma E_n \implies \alpha \geq \gamma((1-\alpha)\delta + \sigma^2) \implies  \delta \leq \frac{\alpha-\gamma\sigma^2}{\gamma(1-\alpha)}.
\end{align}
Hence, $\delta$ cannot be arbitrarily large. The inverse relation between $\delta$ and $\gamma$ implies that increasing $\delta$ leads to poor performance for given $\sigma^2$ and $\alpha$. This puts a limit on the maximum achievable spectral efficiency much like the Mazo limit for the conventional faster-than-Nyquist signaling~\cite{mazo1975faster}.

\subsection{Complexity}
\label{subsec:compl}
\begin{table}[t]
    \caption{Complexity of the MUB approach.}
    \label{tab:compl}
    \renewcommand{\arraystretch}{1.25}
    \centering
    \begin{tabular}{|c|c|}
        \hline
        Operation & Complexity \\
        \hline
        Mount on basis &  $\mathcal{O}(M^2N^2)$ \\
        \hline
        QR-precoding & $\mathcal{O}(M^2N^2)$\\
        \hline
        MMSE equalization & $\mathcal{O}(M^3N^3)$\\
        \hline
        TCM encoding \& decoding & $\mathcal{O}(MN)$ \& $\mathcal{O}(2^mMN)$\\
        \hline
        Interference cancellation & $\mathcal{O}(MN) + \mathcal{O}(M^2N^2)$\\
        \hline
    \end{tabular}
\end{table}
The complexity involved in each step of the MUB approach is presented in Table~\ref{tab:compl}. At the transmitter, mounting the information symbols on the bases incurs complexity square in the frame dimensions. QR-precoding also incurs similar complexity. Equalization via MMSE incurs cubic complexity owing to matrix inverse operation. At the transmitter (receiver), TCM encoding (decoding) incurs complexity $\mathcal{O}(MN)$ ($\mathcal{O}(2^mMN)$, $m$ is the memory size). Each turbo iteration has an interference cancellation step which has complexity $\mathcal{O}(MN)+\mathcal{O}(M^2N^2)$. The overall complexity is dominated by the MMSE equalization which is common to both MUB and Nyquist-sampling approach.

\subsection{Numerical Results}
\label{subsec:num_res_mub}
\begin{figure}
    \centering
    \includegraphics[width=0.9\linewidth]{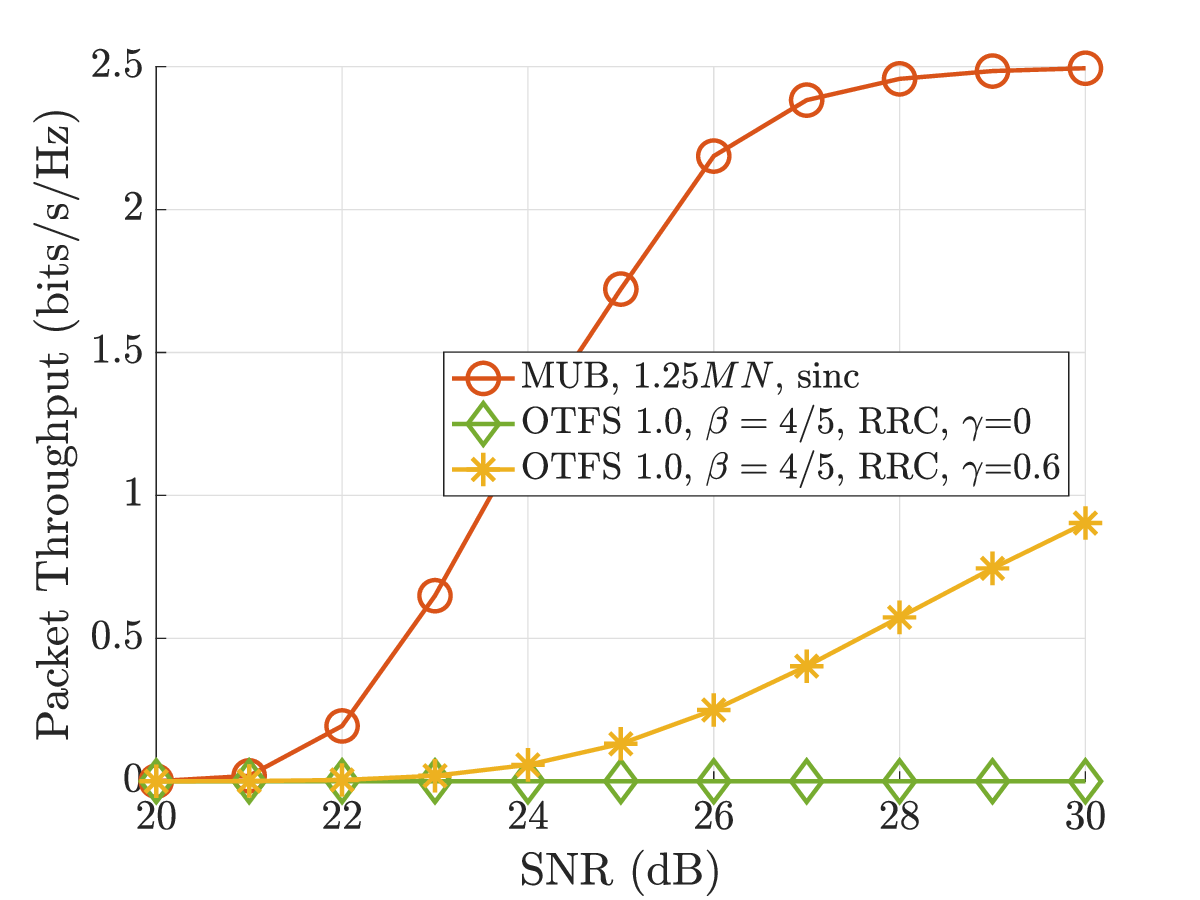}
    \caption{Throughput of the proposed MUB scheme as a function of SNR for $1.25MN$ symbols accounting for packet retransmission. A 6-path Vehicular-A channel~\cite{veh_a}, and parameters $M=31, N=37, \nu_p=30$ kHz. The proposed MUB scheme achieves the highest throughput.}
    \label{fig:throughput}
    \vspace{-4mm}
\end{figure}

We simulate Zak-OTFS MUB approach, OFDM, and OTFS 1.0 system models with FTN approach.
For the MUB scheme, we consider the first frame ($\mathbf{x}_1$) to be a full frame containing $MN$ symbols and the superimposed second frame ($\mathbf{x}_2$) to be sparse ($=0.25MN$ symbols) and detect the full frame first. We TCM encode both the frames. For TCM encoding, we use the following rule:
\begin{align}
    \label{eq:4_qam_16_qam_encoding}
    s = (-1)^{(\mathbf{g}_{11}^\top\mathbf{b})_2} - 3(-1)^{(\mathbf{g}_{12}^\top\mathbf{b})_2}
\end{align}
where $\mathbf{g}_{11} = [0 \ 1 \ 0]^\top, \mathbf{g}_{12} = [1 \ 1 \ 1]^\top$ are the generator polynomials, $\mathbf{b} \in \{0, 1\}^{n \times 1}$ is a binary vector to the input of the encoder, and $n = 3$ is the sliding window length. The output of TCM encoder takes values in the set $\{-2, -1, 1, 2\}$. Complex information symbols are constructed by treating TCM encoded output at even indices as real part and at the odd indices as corresponding imaginary part. The TCM encoder converts 4-QAM symbols (i.e., 2 bits) to amplitude modulated 16-QAM symbols (i.e., 4 bits)\footnote{In this paper, as a proof of concept, we use this design for the TCM encoder. However, optimizing the code design is an important direction for future research.}. For Zak-OTFS, OTFS 1.0, and OFDM, we use the same precoding technique described in Sec.~\ref{subsec:precoder_design}. Further, for fair comparison we use the same TCM encoding for all modulation schemes. The power distribution between the frames is fixed by choosing $\alpha=0.9$.

For the MUB scheme, we transmit $1.25MN$ symbols in bandwidth $B$ and time $T$. This implies a compression ratio of $\nicefrac{1}{1.25} = \nicefrac{4}{5}$. Therefore for comparison, in OFDM and OTFS 1.0 we choose the compression factor $\beta=\nicefrac{4}{5}$.

Let $p_b$ be the BER and $L$ the packet length in bits. Packet success probability $P_s=(1-p_b)^L$, packet error probability $P_e=1-P_s$. Let $N$ be the number of transmissions per packet; with automatic repeat request (ARQ), $\mathbb{E}[N]=1/P_s$. For rate $R$, throughput is $T=\frac{R}{\mathbb{E}[N]}=R P_s=R(1-p_b)^L$. For MUB (our approach), $L=2.5MN, R=2.5$ and for OTFS 1.0 FTN, $L=2 MN, R=\tfrac{2}{\beta(1+\gamma)}$.

Fig.~\ref{fig:throughput} shows the throughput as a function of SNR using the above calculations. The proposed MUB approach achieves the highest throughput and at high SNRs reaches the maximum achievable limit, while OTFS 1.0 with $\gamma=0.6$ achieves only about 1 bits/s/Hz. OTFS 1.0 with $\gamma=0$ does not achieve usable throughput since the BER floors and would require multiple retransmissions.

\subsection{Open Research Problems}
\label{open_res_prob_mub}
In this Section, we used TCM for encoding both frames. Would using LDPC codes give better performance? The order in which the frames are detected/decoded plays an important role, which has not been considered here. The value of $\alpha$ is influenced by this decision and therefore choosing to detect/decode the full frame first pushes the optimal value of $\alpha$ to be close to $1$. Detecting the sparse frame first may result in different optimal $\alpha$ and may also improve performance. Lastly, the effect of distribution of $1.25MN$ symbols between the two frames remains to be studied.

\section{Equalization}
\label{sec:equalization}

The choice between embracing and preventing ICI/ISI (recall Fig.~\ref{fig:prevent_vs_embrace}) has implications on the overall performance of the wireless system. For example, while the philosophy of embracing ICI/ISI allows for tighter packing of information symbols, such systems are harder to equalize. On the other hand, the philosophy of preventing ICI/ISI results in more separated information symbols but such systems are easier to equalize. In this Section, we discuss the equalization complexity of modulation schemes from both philosophies and show how to reduce this complexity when we embrace ICI/ISI\footnote{In this Section, we reduce the complexity of Zak-OTFS by using a transform that maps the DD channel to the frequency domain. However, an alternate architecture where a spread carrier is overlaid on a CP-OFDM signal is also possible and is studied in~\cite{nisar2026improving}. Acquiring channel is efficient in Zak-OTFS and equalizing is efficient in CP-OFDM and this approach combines the advantages of both.}.


\subsection{OFDM}
\label{subsec:ofdm_eq}

In OFDM, the philosophy is to avoid/prevent ISI/ICI. In other words, the subcarrier spacing $\Delta f$ is adjusted such that the maximum Doppler spread $\nu_{\max}$ is insignificant in comparison, i.e., $\Delta f \gg \nu_{\max}$. The information grid for OFDM shown in Fig.~\ref{fig:prevent_vs_embrace} where the information symbols are placed on a coarse grid. This choice of preventing ISI/ICI ensures that the information symbols do not interact with one another when passing through the channel which translates to easier equalization at the receiver.

Equalization in OFDM is carried out using a one-tap equalizer~\cite{tse2005fundamentals}. The system model for OFDM can be written as:
\begin{align}
    \label{eq:ofdm_eq1}
    \mathbf{y}_{\mathsf{F}} = \mathbf{H}_{\mathsf{F}}\mathbf{x}_{\mathsf{F}} + \mathbf{n}_{\mathsf{F}},
\end{align}
where $\mathbf{y}_{\mathsf{F}}, \mathbf{x}_{\mathsf{F}}, \mathbf{n}_{\mathsf{F}}\in \mathbb{C}^{MN\times 1}$ are respectively the received, transmitted, and noise vector in the frequency domain. The channel matrix $\mathbf{H}_{\mathsf{F}}$ (under the assumption $\Delta f \ll \nu_{\max}$) is diagonal. The one-tap equalizer is:
\begin{align}
    \label{eq:ofdm_eq2}
    \mathbf{y}_{\mathsf{F}}' = \mathbf{W}_{\mathsf{F}}\mathbf{y}_{\mathsf{F}},
\end{align}
where $\mathbf{W}_{\mathsf{F}}\in \mathbb{C}^{MN \times MN}$ is either $\mathbf{H}_{\mathsf{F}}^{-1}$ (the least-squares equalizer) or $(\mathbf{H}_{\mathsf{F}}^\mathsf{H}\mathbf{H}_{\mathsf{F}} + \sigma^2\mathbf{I})^{-1}\mathbf{H}_{\mathsf{F}}^{\mathsf{H}}$ (the MMSE equalizer). Since $\mathbf{H}_{\mathsf{H}}$ is a diagonal matrix, computation of $\mathbf{W}_{\mathsf{F}}$ (also a diagonal matrix) is very efficient and incurs complexity linear in the frame dimension.


\subsection{Zak-OTFS}
\label{subsec:zak_otfs_eq}
In Zak-OTFS, on the other hand, the philosophy is to embrace ISI/ICI. This implies that the channel ``spreads'' the information symbols, and the spread is greater than the symbol separation along both delay and Doppler axes. A channel with $P$ resolvable paths results in $P$ copies of each transmitted symbols. Therefore the channel matrix has (at least) $P$ non-zero elements along each column. Further, each received symbol is also a superposition of $P$ transmitted symbols. Therefore each row has (at least) $P$ non-zero elements. Practical pulse shapes decay slowly, so these add a few non-zero entries in each row and column around each of the $P$ entries. The result is that the channel matrix in the DD domain is not sparse. Computing the least squares or MMSE equalizer incurs complexity cubic in the frame dimension. 

The complexity of equalizing Zak-OTFS can be reduced if the equalization is carried out in the frequency domain instead of the DD domain~\cite{Mattu2025_npj}. Symbols in the DD domain can be moved to the frequency domain using the inverse discrete frequency Zak transform (IDFZT)~\cite[Chapter~8]{otfs_book}\footnote{Note that, this expression is different from DZT~\cite{lampel2022onotfs, yogesh2023onthe}. DZT converts the information symbols in the DD domain to time domain whereas IDFZT converts DD domain to frequency domain.}:
\begin{align}
    \label{eq:zakofdm4}
    \mathbf{s}[i] = \frac{1}{\sqrt{M}} \sum_{k_0=0}^{M-1} \mathbf{X}[k_0,i \bmod N] e^{-j\frac{2\pi}{MN} i k_0},
\end{align}
where $\mathbf{X} \in \mathbb{C}^{M\times N}$ is the matrix of DD symbols, and $\mathbf{s} \in \mathbb{C}^{MN \times 1}$ is a vector of symbols in the frequency domain. Note that, the IDFZT is a unitary transform. Let $\mathbf{R} \in \mathbb{C}^{MN \times MN}$ denote the IDFZT matrix, then the system model for Zak-OTFS can be expressed as (continuing on from~\eqref{eq:sec3eq25}):
\begin{align}
    &\mathbf{y} = \mathbf{H}\mathbf{x} + \mathbf{n} = \mathbf{R}^{\mathsf{H}}\mathbf{R}\mathbf{H}\mathbf{R}^{\mathsf{H}}\mathbf{R}\mathbf{x} + \mathbf{n} \nonumber \\
    &\implies \mathbf{R}\mathbf{y} = \mathbf{R}\mathbf{H}\mathbf{R}^{\mathsf{H}}(\mathbf{R}\mathbf{x}) + \mathbf{R}\mathbf{n} \implies \mathbf{r} = \mathbf{H}_{\mathsf{FD}}\mathbf{s} + \mathbf{w},
\end{align}
where $\mathbf{H}_{\mathsf{FD}}\in\mathbb{C}^{MN \times MN}$ is the frequency domain channel matrix, $\mathbf{r}, \mathbf{s}, \mathbf{w}\in\mathbb{C}^{MN \times 1}$ are respectively the transmitted symbols, received symbols, and noise vector in the frequency domain. Since the IDFZT is unitary, both system models are equivalent. The FD channel matrix $\mathbf{H}_{\mathsf{FD}}$ is modulo banded as shown in Fig.~\ref{fig:mod_banded}. This structure in the channel matrix can be exploited to devise low-complexity equalization algorithms. One such algorithm is the conjugate gradient method (CGM)~\cite{liu2020energy}. The idea is to convert the received DD symbols to frequency domain, equalize in the frequency domain and move the equalized symbols back to the DD domain for detection. 

\begin{figure}
    \centering
    \includegraphics[width=0.9\linewidth]{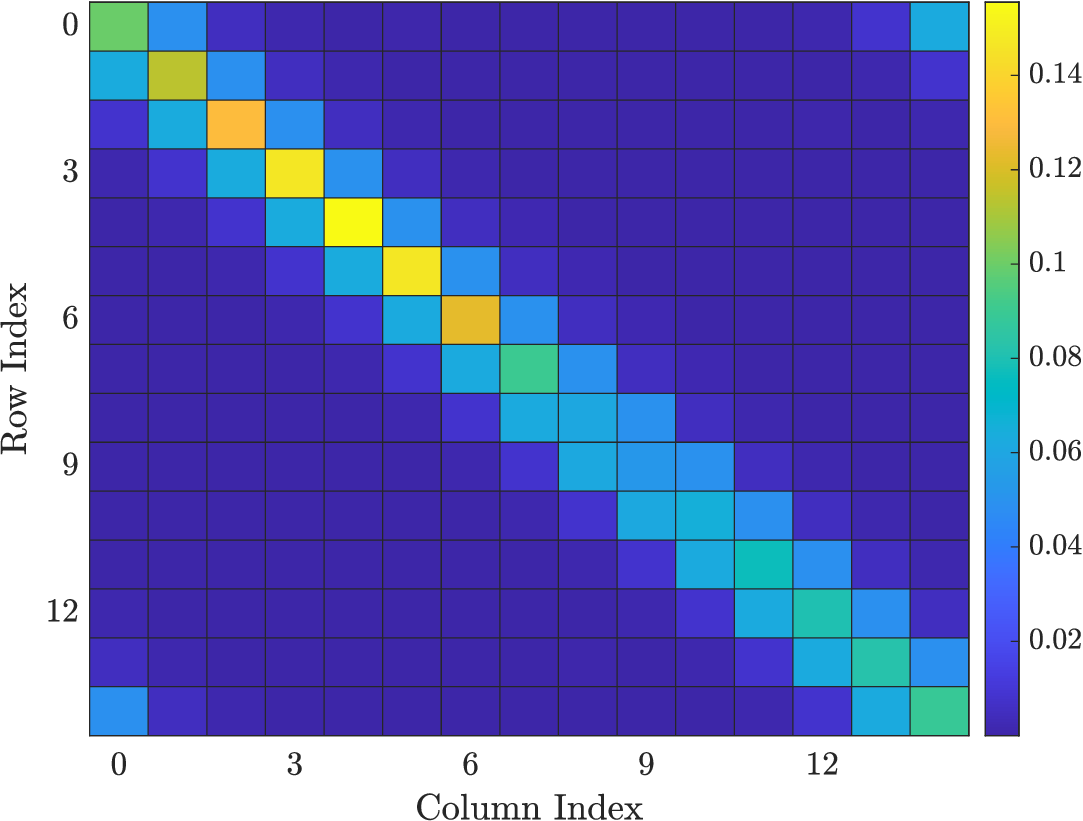}
    \caption{Modulo banded structure in the FD channel matrix $\mathbf{H}$ for $M=3, N=5, \nu_p = 30$ kHz, and the Veh-A channel~\cite{veh_a}.}
    \label{fig:mod_banded}
\end{figure}

\begin{algorithm}
    \caption{Iterative conjugate gradient algorithm}
    \label{alg:conj_grad}
    \begin{algorithmic}[1]
        \STATE \textbf{Inputs:} Channel matrix $\mathbf{H}$, received vector $\mathbf{r}$, covariance matrix of noise $\mathbf{R}_n$, number of delay bins $M$, number of Doppler bins $N$, tolerance $\epsilon$, maximum number of iterations $k$
        \STATE \textbf{Initialize:} $\breve{\mathbf{H}} = \mathbf{H}^{\mathsf{H}}\mathbf{H}$, $\mathbf{b} = \mathbf{H}^{\mathsf{H}}\mathbf{r}$, $\tilde{\mathbf{s}}^{(0)} = \mathbf{0}_{MN\times 1}$, $\mathbf{c}^{(0)} = \mathbf{b} - \breve{\mathbf{H}}\tilde{\mathbf{s}}^{(0)}-\mathbf{R}_n\tilde{\mathbf{s}}^{(0)}$, $\mathbf{p}^{(0)} = \mathbf{c}^{(0)}$, $c^{(0)}_{\mathsf{norm}} = \Vert\mathbf{c}\Vert_2^2$
        \FOR {$i=1:k$}
            \STATE $\mathbf{a}_{\mathbf{p}} = \breve{\mathbf{H}}{\mathbf{p}^{(i-1)}}+\mathbf{R}_n\mathbf{p}^{(i-1)}$
            \STATE $\alpha = \frac{c^{(i-1)}_{\mathsf{norm}}}{\mathbf{p}^{\mathsf{H}}\mathbf{a}_{\mathbf{p}}}$
            \STATE $\tilde{\mathbf{s}}^{(i)} = \tilde{\mathbf{s}}^{(i-1)} + \alpha \mathbf{p}^{(i-1)}$
            \STATE $\mathbf{c}^{(i)} = \mathbf{c}^{(i-1)} - \alpha\mathbf{a}_{\mathbf{p}}$
            \STATE $c^{(i)}_{\mathsf{norm}} = \Vert\mathbf{c}^{(i)}\Vert_2^2$
            \IF{$c^{(i)}_{\mathsf{norm}} < \epsilon^2$}
                \STATE break
            \ENDIF
            \STATE $\mathbf{p}^{(i)} = \mathbf{c}^{(i)} + \frac{c^{(i)}_{\mathsf{norm}}}{c^{(i-1)}_{\mathsf{norm}}}\mathbf{p}^{(i-1)}$
            \STATE $c^{(i-1)}_{\mathsf{norm}} = c^{(i)}_{\mathsf{norm}}$
        \ENDFOR
        \RETURN $\tilde{\mathbf{s}}^{(i)}$
    \end{algorithmic}
\end{algorithm}

The algorithm listing of CGM is provided in Algorithm~\ref{alg:conj_grad}. Instead of computing the matrix inverse directly which incurs cubic complexity, the CGM algorithm iteratively computes the inverse while solving for the equalized symbol vector. The complexity of the CGM algorithm is dependent on the number of non-zero entries in the channel matrix. This is characterized by a parameter called the ``spread width'' (denoted by $b$), which denotes the number of non-zero entries on either side of the diagonal.

In the Algorithm, $\mathbf{p}^{(i)}$ denotes the search direction vector used to update the estimated symbol vector $\tilde{\mathbf{s}}^{(i)}$. At initialization, the estimate is set to $\tilde{\mathbf{s}}^{(0)}=\mathbf{0}$, and hence the initial residual is:
\begin{align}
\mathbf{c}^{(0)}
&= \mathbf{b} - \left(\breve{\mathbf{H}}+\mathbf{R}_n\right)
\tilde{\mathbf{s}}^{(0)} = \mathbf{b} = \mathbf{H}^{\mathsf{H}}\mathbf{r}.
\end{align}
Since no previous search direction is available at the first iteration, the initial search direction is chosen as the initial residual, i.e.:
\begin{equation}
\mathbf{p}^{(0)}
= \mathbf{c}^{(0)}
= \mathbf{H}^{\mathsf{H}}\mathbf{r}.
\end{equation}
In subsequent iterations, $\mathbf{p}^{(i)}$ is recursively constructed from the current residual and the previous search direction so that the search directions are conjugate with respect to the system matrix $\breve{\mathbf{H}}+\mathbf{R}_n$. Specifically, the search direction is updated as:
\begin{equation}
\mathbf{p}^{(i)}
=
\mathbf{c}^{(i)}
+
\frac{\left\|\mathbf{c}^{(i)}\right\|_2^2}
     {\left\|\mathbf{c}^{(i-1)}\right\|_2^2}
\mathbf{p}^{(i-1)}.
\end{equation}
The vector $\mathbf{p}^{(i-1)}$ therefore determines the direction along which the current symbol estimate is updated according to:
\begin{equation}
\tilde{\mathbf{s}}^{(i)}
=
\tilde{\mathbf{s}}^{(i-1)}
+
\alpha\mathbf{p}^{(i-1)}.
\end{equation}

\subsection{Numerical Results}
\begin{figure}
    \centering
    \includegraphics[width=0.99\linewidth]{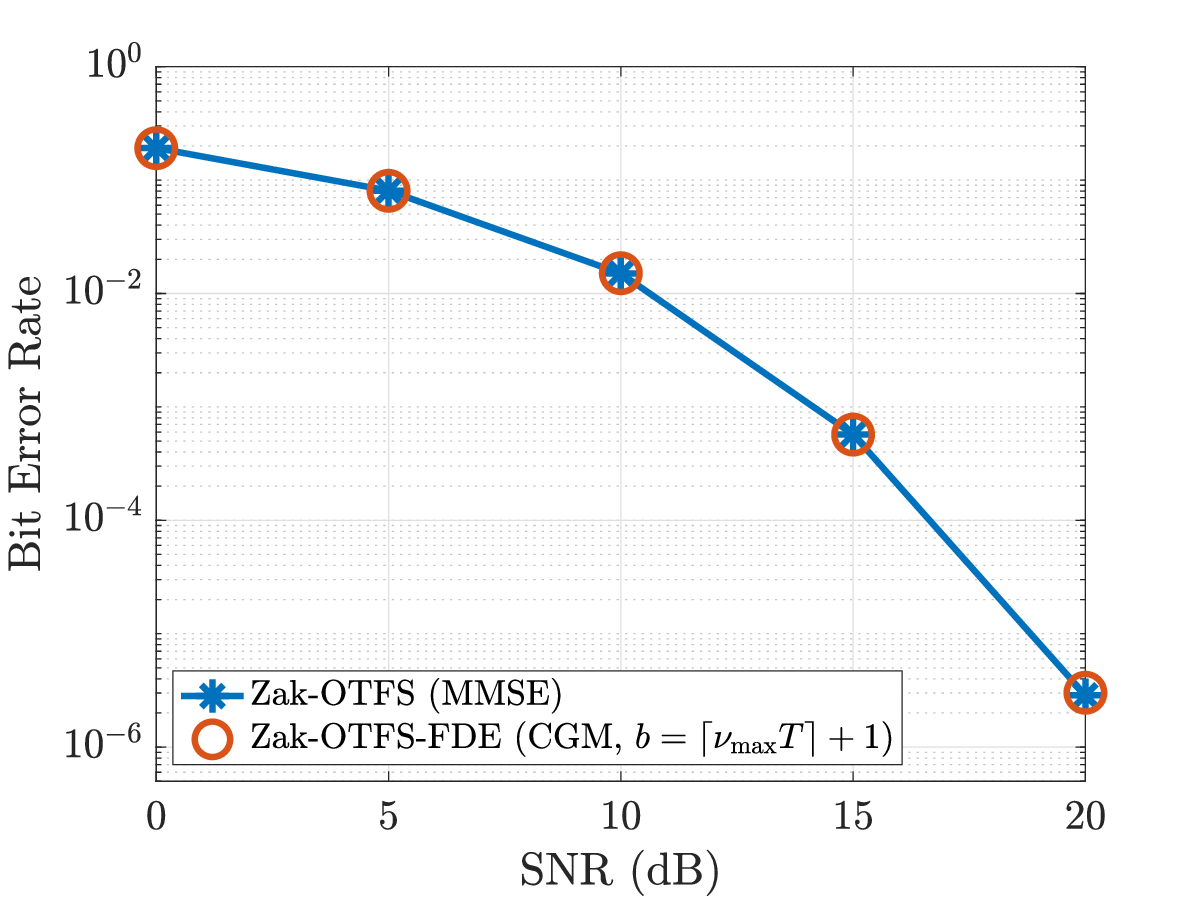}
    \caption{BER performance comparing the DD and FD equalization (FDE) of Zak-OTFS. FDE of Zak-OTFS via CGM algorithm with $k=250, \epsilon=10^{-6}, b=\lceil\nu_{\max}T\rceil+1$. Simulation parameters: $M = 31$, $N = 37$, $\nu_{p} = 30$ kHz, Veh-A channel with $\nu_{\max} = 815$ Hz, RRC pulse shaping filter with $\beta_\tau = \beta_\nu = 0.6$.}
    \label{fig:perf_csi_vs_algo_M31_N37_highdop}
\end{figure}
Figure~\ref{fig:perf_csi_vs_algo_M31_N37_highdop} shows the bit-error performance of Zak-OTFS using MMSE equalizer in the DD domain and CGM in the frequency domain. There is no performance difference between the two approaches. The DD equalization incurs complexity cubic in the order of frame dimension while frequency domain equalization (FDE) incurs complexity linear in the frame dimension.

\subsection{Open Research Problems}
\label{subsec:equalization_future}


Our method above is an example of a low-complexity equalization scheme, and in general the problem of overcoming the $O(M^3N^3)$ equalization complexity barrier of Zak-OTFS is open. Is it possible to design a one-tap equalizer for Zak-OTFS, akin to the equalizer in OFDM? Furthermore, how should equalization be performed in MIMO systems? The broad question of the optimality of separating equalization from channel estimation and symbol detection remains open. Another interesting open problem is the design of filters that reduce the equalization complexity (by reducing the effective number of taps) by introducing controlled interference in the transfer domain.


\section{Neural Receiver}
\label{sec:neural_rx}

\subsection{Principle of Neural Reception as Deconvolution}
Wireless signal propagation is fundamentally governed by convolution. The discrete-time baseband model for linear time invariant (LTI) channels is:
\begin{equation}
    y[n] = (h * x)[n] + w[n],
\end{equation}
where $h$ is the channel impulse response. Therefore, receiver design reduces to approximating the inverse operator:
\begin{equation}
    x[n] \approx (g * y)[n], \quad g \approx h^{-1}.
\end{equation}

Recent work proposes a \emph{universal neural receiver} that explicitly adopts this perspective, where the neural network is designed to invert convolution rather than implicitly learn it \cite{calderbank_universal_rx}. The key idea is to separate \emph{what to invert} (which is channel-dependent) from \emph{how to invert it} (which depends on the specific implementation details). This two-stage approach results in lightweight architectures capable of real-time adaptation at the \emph{speed of wireless}.

\subsection{Reservoir Computing as a Neural Equalizer}
Reservoir computing (RC), particularly when implemented using echo-state networks (ESNs), provides an efficient realization of neural receivers. An ESN consists of a fixed recurrent reservoir with $\mathsf{N}_{\mathrm{n}}$ randomly interconnected neurons, an input weight layer and a trainable linear readout:
\begin{align}
    \mathbf{s}[n] &= \sigma\big(\mathbf{W}_{\mathrm{res}} \mathbf{s}[n-1] + \mathbf{W}_{\mathrm{in}} \mathbf{y}[n]\big), \\
    \hat{\mathbf{x}}[n] &= \mathbf{W}_{\mathrm{out}} \mathbf{s}[n],
\end{align}
where $\sigma(\cdot)$ denotes a non-linear activation function, e.g., the hyperbolic tangent function, $\mathbf{s}[n] \in \mathbb{C}^{\mathsf{N}_{\mathrm{n}} \times 1}$ denotes the internal reservoir state at time index $n$, $\mathbf{y}[n] \in \mathbb{C}^{\mathsf{d}_{\mathrm{in}} \times 1}$ and $\hat{\mathbf{x}}[n] \in \mathbb{C}^{\mathsf{d}_{\mathrm{out}} \times 1}$ respectively denote the input to and output from the ESN at time index $n$, and $\mathbf{W}_{\mathrm{in}} \in \mathbb{C}^{\mathsf{N}_{\mathrm{n}} \times \mathsf{d}_{\mathrm{in}}}$, $\mathbf{W}_{\mathrm{res}} \in \mathbb{C}^{\mathsf{N}_{\mathrm{n}} \times \mathsf{N}_{\mathrm{n}}}$ and $\mathbf{W}_{\mathrm{out}} \in \mathbb{C}^{\mathsf{d}_{\mathrm{out}} \times \mathsf{N}_{\mathrm{n}}}$ respectively denote the input, reservoir and output weight matrices.

Only $\mathbf{W}_{\mathrm{out}}$ is trained via least squares to enable \emph{online, low-complexity learning} using limited pilot symbols \cite{jere_xai_rc_mimo}:
\begin{equation}
    \mathbf{W}_{\mathrm{out}} = \arg\min_{\mathbf{W}} \|\mathbf{W}\mathbf{S} - \mathbf{X}\|_F^2,
\end{equation}
where $\mathbf{S} = \begin{bmatrix}
    \mathbf{s}[0] & \cdots & \mathbf{s}[T-1]
\end{bmatrix} \in \mathbb{C}^{\mathsf{N}_{\mathrm{n}} \times T}$ and $\mathbf{X} = \begin{bmatrix}
    \hat{\mathbf{x}}[0] & \cdots & \hat{\mathbf{x}}[T-1]
\end{bmatrix} \in \mathbb{C}^{\mathsf{d}_{\mathrm{out}} \times T}$ respectively denote the matrix of reservoir states and ESN outputs / training labels across $T$ time steps.

\subsection{Signal Processing Interpretation of ESNs}
A key insight is that a \emph{linearized ESN} (with identity $\sigma(\cdot)$) is equivalent to a linear time-invariant (LTI) filter. A single neuron corresponds to a first-order IIR filter:
\begin{equation}
    H(z) = \frac{1}{1 - a z^{-1}},
\end{equation}
while a multi-neuron reservoir implements a higher-order IIR system \cite{jere_equalization_rc}. Therefore, an ESN realizes:
\begin{equation}
    \hat{x}[n] = \sum_{k=0}^{\infty} g[k]\, y[n-k],
\end{equation}
where the effective impulse response $g[k]$ is parameterized by the reservoir. From a geometric perspective, the ESN constructs a function approximation:
\begin{equation}
    g \in \mathrm{span}\{g_k\}_{k=1}^{K},
\end{equation}
where each neuron contributes a basis filter. The training corresponds to projecting the inverse channel response onto this subspace \cite{jere_xai_rc_mimo}.

\subsection{Structured Inverse Filtering and Toeplitz Representation}
Over a block length $T$, the LTI channel can be written as:
\begin{equation}
    \mathbf{y} = \mathbf{H}\mathbf{x},
\end{equation}
where $\mathbf{H}$ is a lower triangular Toeplitz (LTT) matrix. The neural receiver approximates:
\begin{equation}
    \hat{\mathbf{x}} = \mathbf{G}\mathbf{y}, \quad \mathbf{G} \approx \mathbf{H}^{-1}.
\end{equation}

The ESN implicitly parameterizes $\mathbf{G}$ as a finite-length inverse impulse response. Due to truncation, $\mathbf{G}$ is approximately Toeplitz but not strictly banded. Furthermore, the impulse response can be expressed as a sum of exponentials:
\begin{equation}
    g[k] = \sum_{i=1}^{N} c_i \lambda_i^k,
\end{equation}
corresponding to a \emph{Vandermonde (pole) expansion}, which explains the expressive power of ESNs \cite{rc_universal_approx}.

\subsection{Incorporating Domain Knowledge}
Unlike conventional deep learning approaches, RC can directly incorporate domain knowledge by configuring the reservoir weights. For LTI channels, the framework presented thus far enables utilizing channel statistics to initialize ESNs appropriately via pole placement~\cite{jere_equalization_rc,jere_xai_rc_mimo,calderbank_universal_rx}. Similar insights extend to MIMO systems, where domain knowledge, e.g., separability of the channel across antennas and time / frequency, can be used to design and configure ESNs appropriately~\cite{jere_xai_rc_mimo}.

\subsection{$2$D Neural Receiver for Linear Time Varying Channels}
The RC architecture described thus far was limited to one dimension in order to approximate LTI channels. In high-mobility scenarios, the channel is better represented by linear time varying (LTV) channels. In this context, a \emph{two-dimensional} RC ($2$D-RC) architecture has been proposed in~\cite{xu_2drc_icc,xu_2drc_twc} to equalize multicarrier OTFS (MC-OTFS) transmissions. The discretized MC-OTFS input-output relation is given by 2D circular convolution:
\begin{equation}
    Y[k,l] = \sum_{m,n} H[m,n] X[k-m,l-n].
\end{equation}

The $2$D-RC embeds the $2$D circular convolution structure into its architecture to equalize channels in the delay-Doppler domain. Unlike earlier approaches requiring multiple $1$D reservoirs, the $2$D-RC uses a single structured network and achieves improved performance and generalization \cite{xu_2drc_twc}.

\subsection{Open Research Problems}
\label{subsec:neural_rx_future}

In this Section, we described a general neural network architecture capable of de-convolving general LTI channels and LTV channels in the specific case of MC-OTFS transmissions. Key questions remain unaddressed, such as how to extend the $2$D-RC architecture for MC-OTFS to other waveforms, e.g., Zak-OTFS, where the input-output relation is given by twisted convolution as in~\eqref{eq:sec3eq21}? Moreover, we have not addressed the general question of \emph{interpretability} -- is it possible to physically interpret the learned parameters of the neural receiver, e.g., in terms of the delay-Doppler representation of the channel? Furthermore, is it possible to utilize similar architectures at the \emph{transmitter}, e.g., for precoding? Finally, extensions to multi-antenna, multi-user systems bring forth the possibility of generalizing the $1$D and $2$D neural receiver architectures to \emph{three dimensions}, which has not been explored in prior work.

\section{Radar}
\label{sec:radar}

\begin{figure*}[!ht]
\centering
\begin{subfigure}{\linewidth}
    \includegraphics[width=\textwidth]{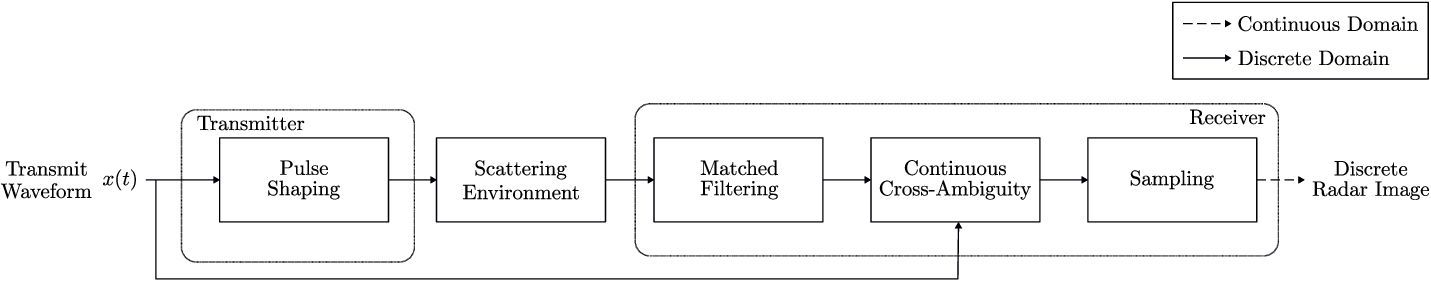}
\caption{Continuous radar architecture.}
    \label{fig:cont_radar}
\end{subfigure} \\
\begin{subfigure}{\linewidth}
    \includegraphics[width=\textwidth]{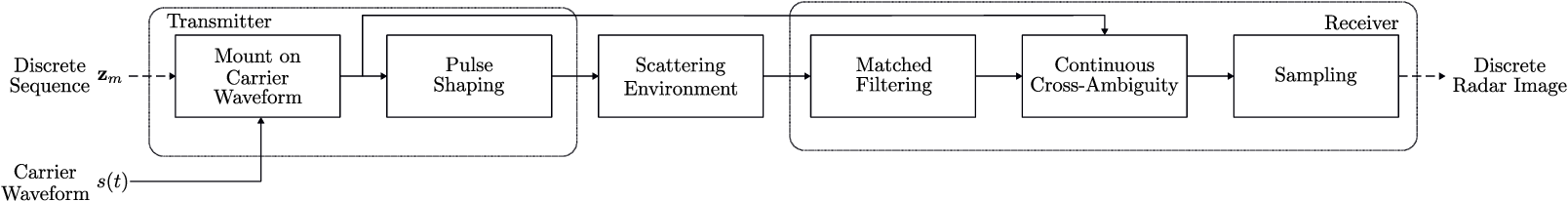}
\caption{Discrete radar architecture with separate optimization of sequences and carriers.}
    \label{fig:disc_phase_coded_radar}
\end{subfigure} \\
\begin{subfigure}{\linewidth}
    \includegraphics[width=\textwidth]{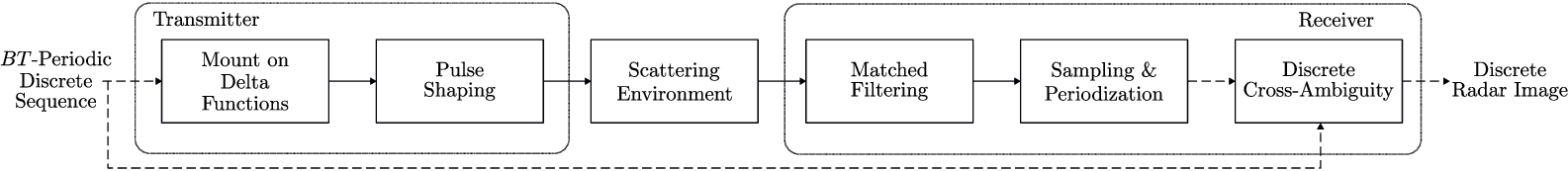}
\caption{Proposed discrete radar architecture in~\cite{Mehrotra2025_EURASIP}.}
    \label{fig:disc_prop_radar}
\end{subfigure}
\caption{Block diagrams showing different radar architectures. Figure adapted from~\cite{Mehrotra2025_EURASIP}.}
\vspace{-5mm}
    \label{fig:block_diag}
\end{figure*}




At a high-level, the operation of a radar consists of transmitting a waveform, which is reflected by the scattering environment and is processed at the receiver~\cite{Woodward1953,Skolnik1980,Levanon2004}. The system model in~\eqref{eq:prelim1} also applies to radar systems, with the channel spreading function $h(\tau,\nu)$ denoting the reflectivity of physical objects in the scattering environment as a function of delay $\tau$ and Doppler $\nu$. Upon sampling, the discrete system model in~\eqref{eq:prelim2} applies. Radar systems estimate the sampled function $\widehat{\mathbf{h}}_{\mathsf{eff}}[k,l]$ in~\eqref{eq:prelim2} (``discrete radar image'') as a function of delay index $k$ and Doppler index $l$ by computing the cross-ambiguity function between the transmitted and received waveforms as in~\eqref{eq:prelim6}. The goal of radar waveform design is to choose the transmitted waveform such that the estimated image of the scattering environment, $\widehat{\mathbf{h}}_{\mathsf{eff}}[k,l]$, accurately matches the ground-truth channel spreading function.

There are two common architectures for radar systems -- continuous radar and discrete radar with separate optimization of sequences -- which are illustrated in Figures~\ref{fig:block_diag}(\subref{fig:cont_radar})-(\subref{fig:disc_phase_coded_radar}).

Continuous radars~\cite{Auslander1985,Miller1991,Moran2001_mathofradar,Moran2004_grouptheory_radar,Calderbank2015_ltv,Jankiraman2018_fmcw} (Fig.~\ref{fig:block_diag}(\subref{fig:cont_radar})) transmit a continuous TD waveform $x(t)$ after pulse shaping to limit time and bandwidth to $T$ and $B$ respectively. The receiver performs matched filtering to remove the impact of pulse shaping, following which a continuous domain cross-ambiguity function is computed between the output of the matched filter and the transmit waveform. The radar forms a discrete radar image by sampling the continuous cross-ambiguity function at multiples of the delay and Doppler resolution.

Conventional discrete radars~\cite{benedetto_phasecoded,Dang2020,Tang2022,Costas1984,Bell2003,Vehmas2021,Pezeshki2009,Fei2024_phasecoded_compl} (see Fig.~\ref{fig:block_diag}(\subref{fig:disc_phase_coded_radar})) mount a discrete sequence $\mathbf{z}_m$ onto a continuous carrier waveform (e.g., a rectangular pulse train), where the sequence is optimized separately from the carrier waveform. Possible strategies for sequence optimization include coding either the phase~\cite{benedetto_phasecoded,Dang2020,Tang2022,Fei2024_phasecoded_compl}, frequency~\cite{Costas1984,Bell2003,Vehmas2021} or amplitude~\cite{Pezeshki2009} of the carrier waveform. After pulse shaping at the transmitter and matched filtering at the receiver, a continuous cross-ambiguity function is computed, which is sampled to obtain the discrete radar image. 

\begin{figure*}[!ht]
\centering
\begin{subfigure}{0.32\linewidth}
    \includegraphics[width=\textwidth]{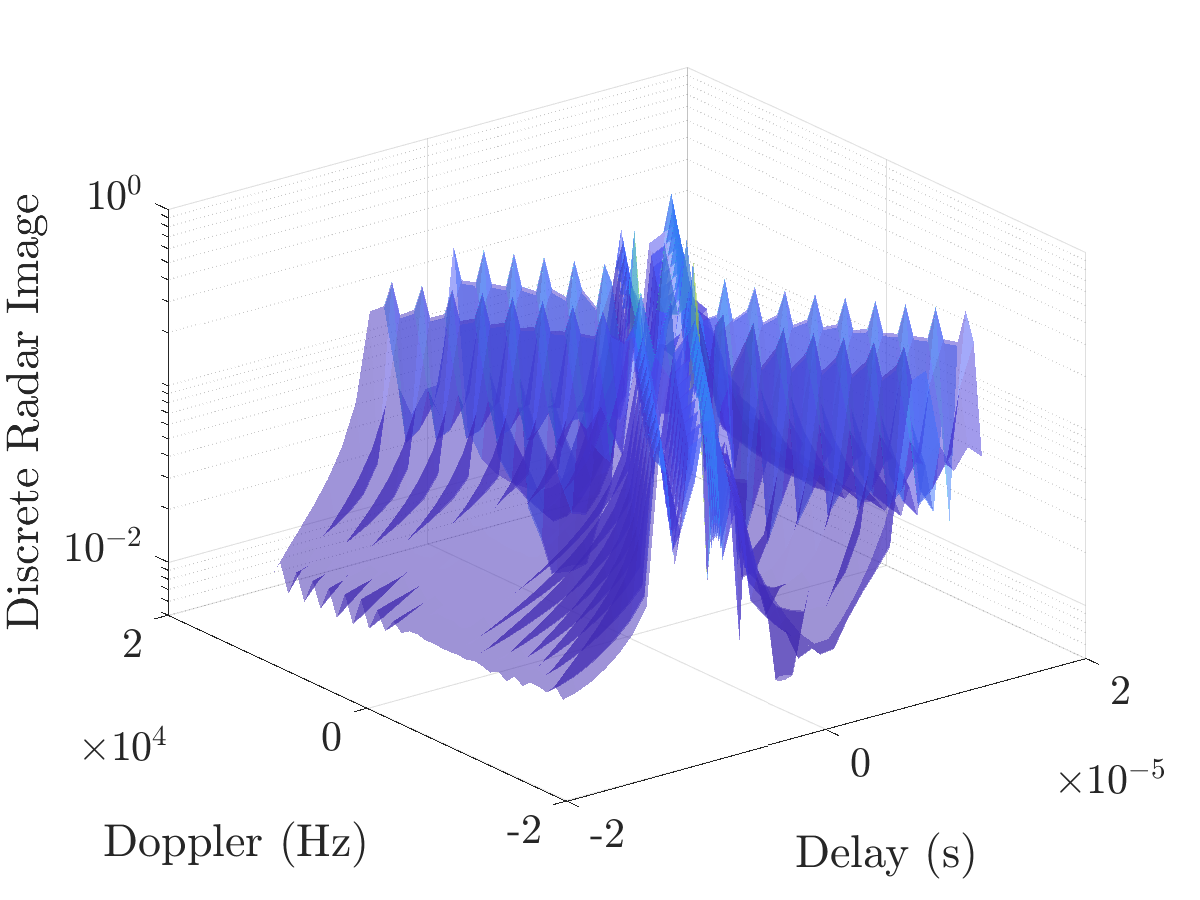}
\caption{Continuous (Fig.~\ref{fig:block_diag}(\subref{fig:cont_radar})).}
    \label{fig:heatmaps_fmcw}
\end{subfigure}
\begin{subfigure}{0.32\linewidth}
    \includegraphics[width=\textwidth]{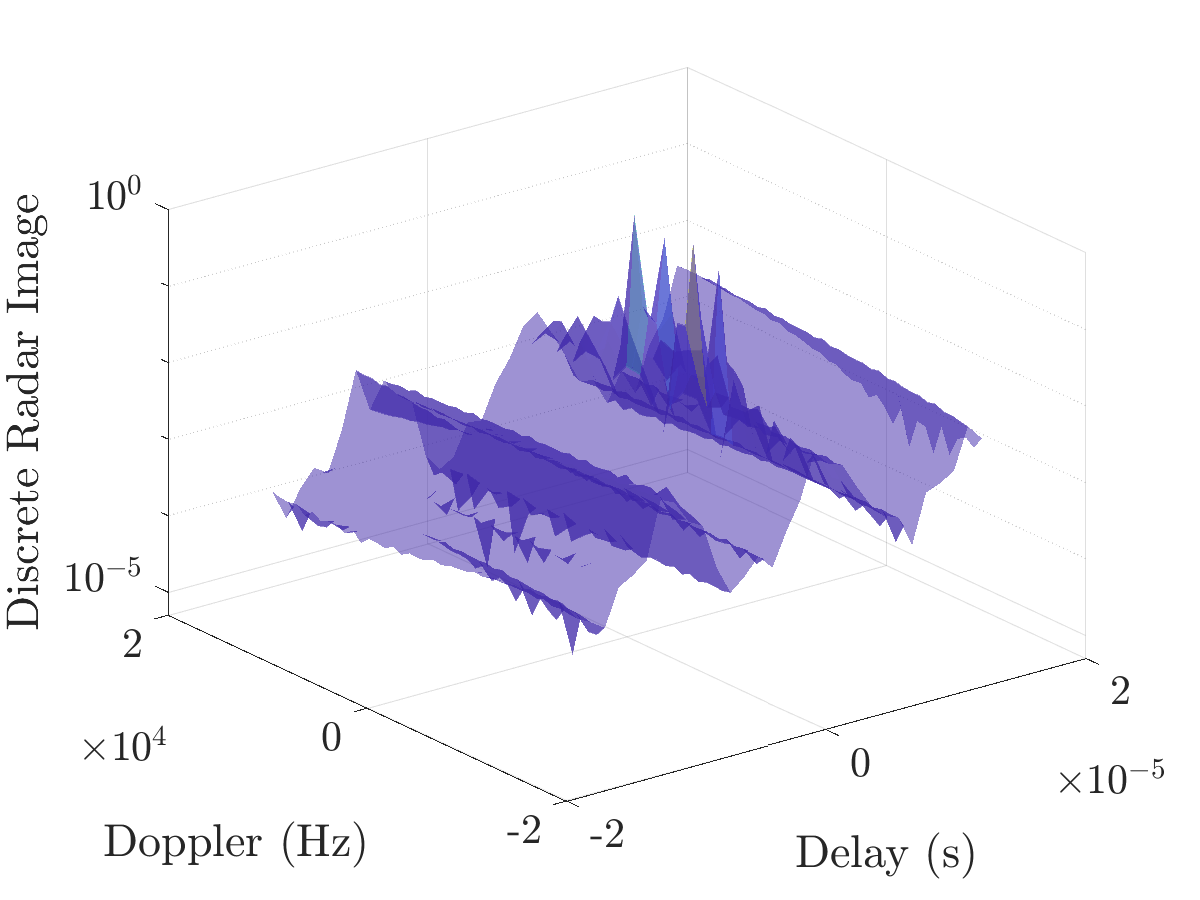}
\caption{Discrete (Fig.~\ref{fig:block_diag}(\subref{fig:disc_phase_coded_radar})).}
    \label{fig:heatmaps_phasecoded}
\end{subfigure}
\begin{subfigure}{0.32\linewidth}
    \includegraphics[width=\textwidth]{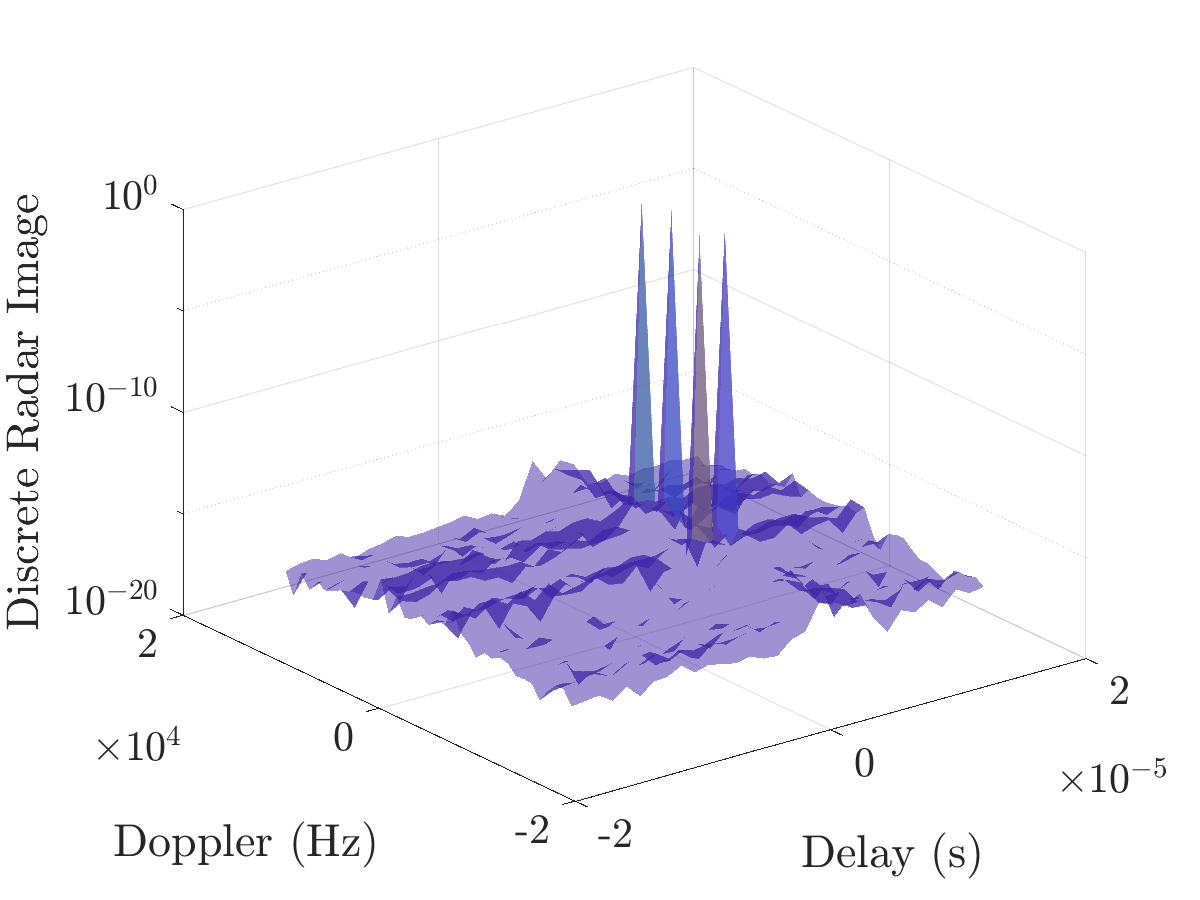}
\caption{Proposed (Fig.~\ref{fig:block_diag}(\subref{fig:disc_prop_radar})).}
    \label{fig:heatmaps_zak}
\end{subfigure}
\caption{Continuous \& discrete radar architectures generate radar images with significant sidelobes around ground-truth target locations due to the choice of carrier waveforms. The proposed approach in~\cite{Mehrotra2025_EURASIP} generates a perfectly localized radar image.}
\vspace{-5mm}
    \label{fig:heatmaps}
\end{figure*}

In~\cite{Mehrotra2025_EURASIP}, we propose a different architecture for discrete radar, illustrated in Fig.~\ref{fig:block_diag}(\subref{fig:disc_prop_radar}). We mount $BT$-periodic discrete sequences onto delta functions, which are pulse shaped and matched filtered at the received. The output of the matched filter is first sampled and periodized to period $BT$. Then, a discrete cross-ambiguity function is computed between the discrete transmitted sequence and the sampled \& periodized matched filter output to obtain the discrete radar image. 

\begin{figure*}[!ht]
\centering
\begin{subfigure}{0.45\linewidth}
    \includegraphics[width=\textwidth]{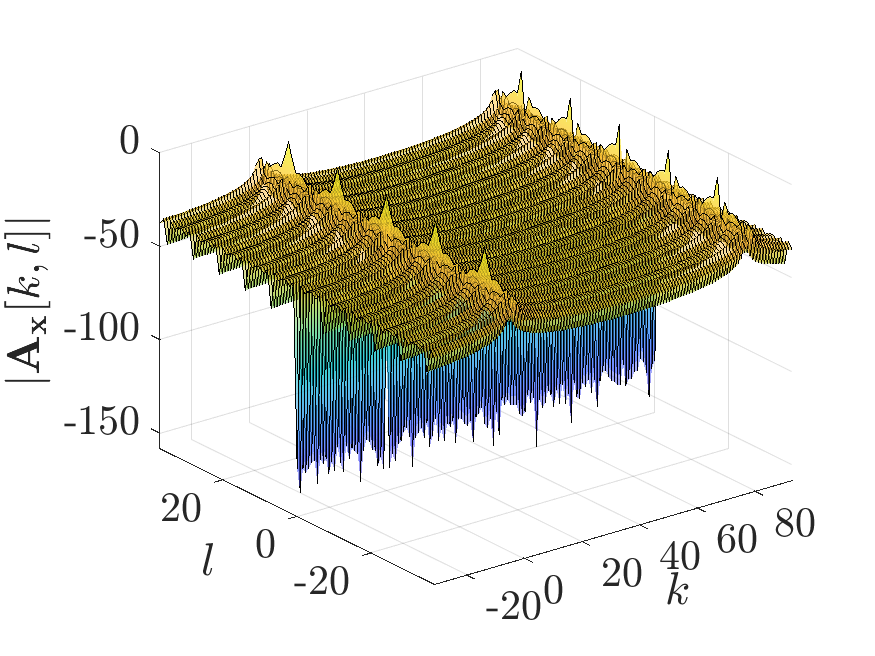}
\caption{ZC phase-coded waveform~\cite{benedetto_phasecoded}.}
    \label{fig:sa_benedetto}
\end{subfigure}
\begin{subfigure}{0.45\linewidth}
    \includegraphics[width=\textwidth]{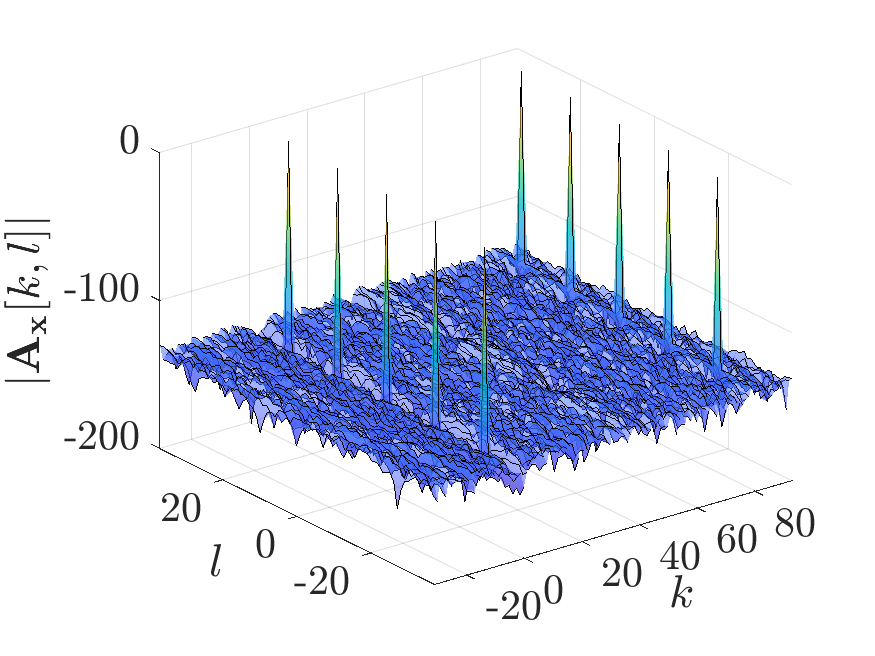}
\caption{Chirp eigenvector from Example~\ref{ex:comm_subgrp_ex2}.}
    \label{fig:sa_proposed}
\end{subfigure}
\caption{The proposed discrete radar architecture in Fig.~\ref{fig:block_diag}(\subref{fig:disc_prop_radar}) achieves perfectly localized ``bed-of-nails'' ambiguity functions. (a) The standard approach of phase coding a rectangular waveform~\cite{benedetto_phasecoded} in discrete radars (see Fig.~\ref{fig:block_diag}(\subref{fig:disc_phase_coded_radar})) has high sidelobes due to poor carrier waveform ambiguity characteristics. (b) Choosing the transmitted waveform as an eigenvector of a maximal commutative subgroup results in ``bed-of-nails'' ambiguity functions. Figure adapted from~\cite{Mehrotra2025_EURASIP}.}
\vspace{-5mm}
    \label{fig:sa}
\end{figure*}

The key difference with conventional discrete radars (Fig.~\ref{fig:block_diag}(\subref{fig:disc_phase_coded_radar})) is that the approach from~\cite{Mehrotra2025_EURASIP} jointly optimizes the sequence and carrier waveform, which significantly reduces sidelobes of the ambiguity function and improves target detection, as we show in Section~\ref{subsec:impl_radar_wvf}. The Heisenberg-Weyl group theory described in Section~\ref{sec:foundations} also enables developing radar waveform libraries with low peak-to-average power ratio (PAPR), which we describe in Section~\ref{subsec:papr_radar}. Sampling prior to the cross-ambiguity computation results in significant complexity gains over both continuous and discrete radars (Figs.~\ref{fig:block_diag}(\subref{fig:cont_radar})-(\subref{fig:disc_phase_coded_radar})), as we show in Section~\ref{subsec:low_compl_radar}. Finally, the proposed architecture enables \emph{instantaneous polarimetry} at low computational cost, as described in Section~\ref{subsec:pol_radar}.

\begin{figure*}[!ht]
\centering
\begin{subfigure}{0.49\linewidth}
    \includegraphics[width=\textwidth]{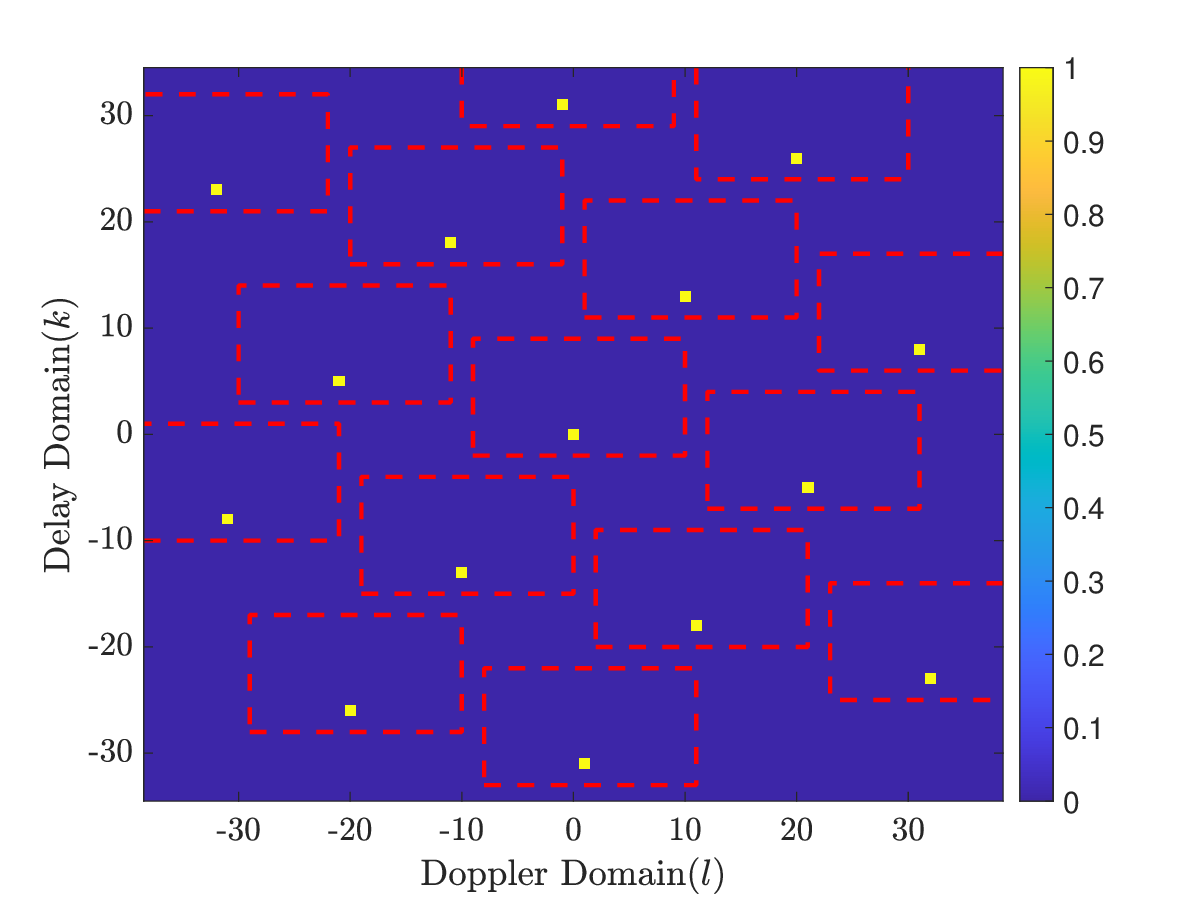}
\caption{Choice of $\mathcal{S}$ satisfying~\eqref{eq:no_aliasing}.}
    \label{fig:no_aliasing_ex1}
\end{subfigure}
\begin{subfigure}{0.49\linewidth}
    \includegraphics[width=\textwidth]{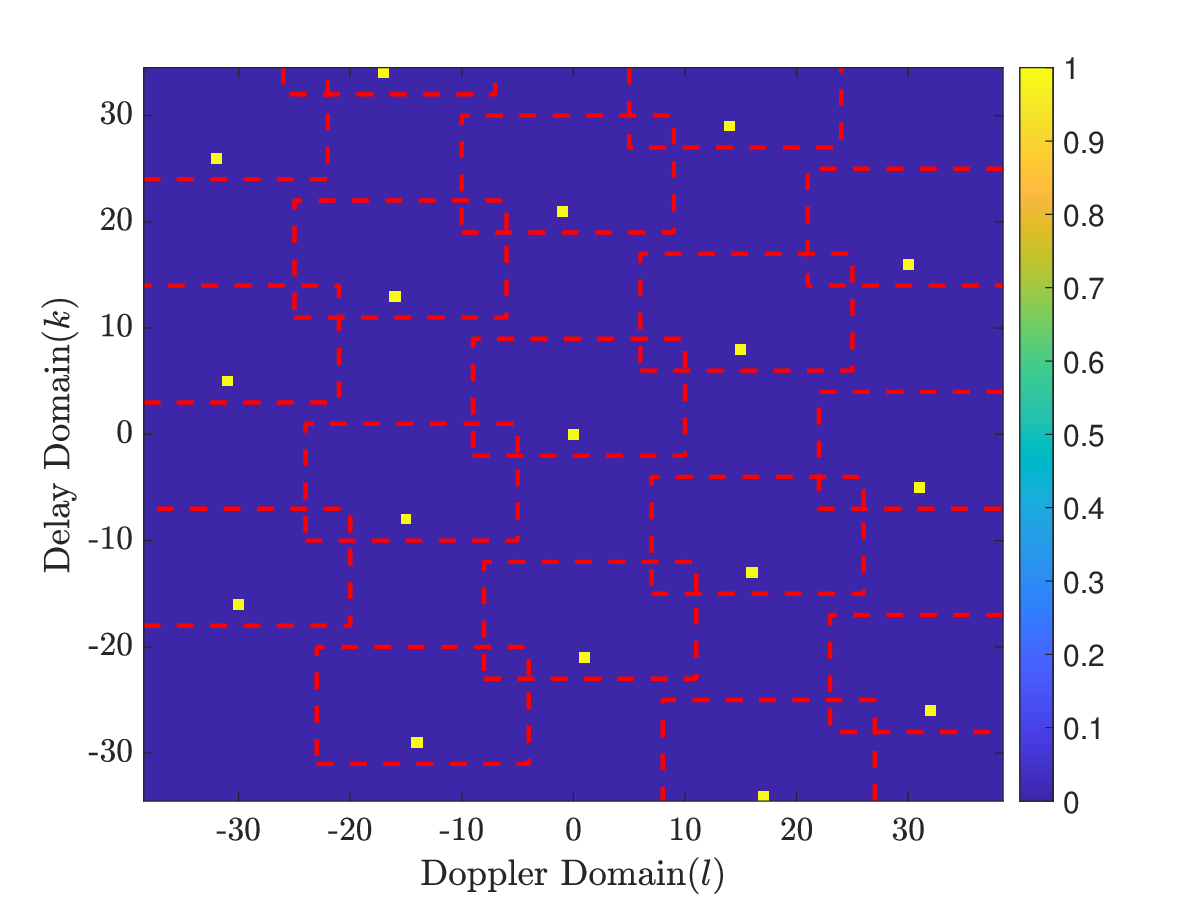}
\caption{Choice of $\mathcal{S}$ not satisfying~\eqref{eq:no_aliasing}.}
    \label{fig:no_aliasing_ex2}
\end{subfigure}
\caption{Example illustrating the choice of an appropriate maximal commutative subgroup $\mathcal{T}_{MN}$ whose delay-Doppler index set $\mathcal{S}$ satisfies the condition in~\eqref{eq:no_aliasing} for a given maximum scattering environment support $\mathcal{C}$. In this example, $\mathcal{C} = [k_{\min},k_{\max}] \times [l_{\min},l_{\max}]$ with $k_{\min} = -2$, $k_{\max} = 8$, $l_{\min} = -9$, $l_{\max} = 9$. Figure adapted from~\cite{Mehrotra2025_WCLSpread,Mehrotra2025_EURASIP}.}
\vspace{-5mm}
    \label{fig:no_aliasing_ex}
\end{figure*}

Fig.~\ref{fig:heatmaps} compares the discrete radar images obtained using the three radar architectures from Fig.~\ref{fig:block_diag} for a four target scattering environment. In Fig.~\ref{fig:heatmaps}(\subref{fig:heatmaps_fmcw}) we consider a continuous radar with an up-slope and down-slope linear frequency modulated (LFM) transmit waveform~\cite{Calderbank2015_ltv}. The resulting radar image suffers from extremely high sidelobes due to the ambiguity function characteristics of the LFM waveform (see~\cite{zakotfs_ltv} for more details). In Fig.~\ref{fig:heatmaps}(\subref{fig:heatmaps_phasecoded}) we consider a discrete radar with a Zadoff-Chu (ZC) phase-coded rectangular waveform. Despite improved target discrimination capabilities, there are significant sidelobes due to the non-zero sidelobes of the waveform (cf. Fig.~\ref{fig:sa}(\subref{fig:sa_benedetto})). The proposed approach generates a radar image with perfectly localized targets in Fig.~\ref{fig:heatmaps}(\subref{fig:heatmaps_zak}).

\begin{remark}
    \label{rmk:radar1}
    Radars ask questions of scattering environments in order to control higher-layer functions such as target tracking. Images of scatterers are a means to an end rather than an end in itself. There may be small differences in the images of a scattering environment produced by continuous and discrete radars (cf. Fig.~\ref{fig:heatmaps}), but both images have the information necessary to control higher layer functions. We refer the reader to the work of Bell~\cite{Bell2002} for information theoretic measures of effectiveness, to Kershaw and Evans~\cite{Evans2002} for information theoretic criteria for waveform scheduling to support tracking. For more details, see~\cite{Moran2004_wvf_lib,Moran2009_wvf_lib_survey,Evans2002,Moran2006}.
\end{remark}

\subsection{Benefit 1: ``Bed-of-Nails'' Ambiguity Functions}
\label{subsec:impl_radar_wvf}

One benefit of the proposed discrete radar architecture is that it enables perfectly localized \emph{``bed-of-nails''} ambiguity functions. Recall that an accurate estimate of the DD scattering environment is achieved when the self-ambiguity function of the transmit waveform is a ``thumbtack'', $\mathbf{A}_{\mathbf{x}}[k,l] \approx \delta[k] \delta[l]$. However, Moyal's Identity (Identity~\ref{idty:moyal}) limits what can be achieved.

The standard approach taken by discrete radars (see Fig.~\ref{fig:block_diag}(\subref{fig:disc_phase_coded_radar})) is to regard the waveform as a signal modulated onto a carrier and to separate the carrier modulation and demodulation processes from the analysis of ambiguity. However, their performance is limited by the ambiguity function characteristics of the chosen carrier waveform. As an illustration, Fig.~\ref{fig:sa}(\subref{fig:sa_benedetto}) plots the magnitude of the self-ambiguity function for a Zadoff-Chu phase-coded rectangular waveform as defined in~\cite{benedetto_phasecoded}, with the waveform designed to have self-ambiguity function magnitude close to $1$ on the line $2\alpha k - l \equiv 0 \bmod{MN}$. The ambiguity characteristics of the rectangular carrier waveform results in significant sidelobes outside the locations given by the line.

In contrast, the proposed discrete radar architecture in Fig.~\ref{fig:block_diag}(\subref{fig:disc_prop_radar}) achieves \emph{perfectly localized} ``bed-of-nails'' ambiguity functions with minimal sidelobes. Specifically, we choose the radar waveform as an eigenvector of a maximal commutative subgroup. Recall from Corollary~\ref{corr:max_comm_ambg} that for a maximal commutative subgroup $\mathcal{T}_{MN}$ with eigenbasis $\big\{\mathbf{v}_{i}\big\}_{i=1}^{MN}$, the self-ambiguity function of any eigenvector $\mathbf{v}_{i}$ vanishes for all $e^{\frac{j2\pi}{MN}m} \mathcal{D}_{(k,l)} \not\in \mathcal{T}_{MN}$ and is unimodular only at $e^{\frac{j2\pi}{MN}m} \mathcal{D}_{(k,l)} \in \mathcal{T}_{MN}$. Fig.~\ref{fig:sa}(\subref{fig:sa_proposed}) illustrates the advantage of our approach. We consider the discrete chirp eigenvector corresponding to the Heisenberg-Weyl maximal commutative subgroup in Example~\ref{ex:comm_subgrp_ex2}, and observe essentially no sidelobes in the self-ambiguity function magnitude outside the self-ambiguity function support $2\alpha k - l \equiv 0 \bmod{MN}$.

\subsubsection{How to Select a Radar Waveform}
\label{subsubsec:radar_wvf_cryst}

How should one choose an appropriate maximal commutative subgroup $\mathcal{T}_{MN}$ in our approach? Moyal's Identity (Identity~\ref{idty:moyal}) states that a perfect ``thumbtack'' ambiguity function is unachievable. However, it is feasible to design a waveform whose self-ambiguity function is zero at all locations excluding the origin in a connected region $\mathcal{C}$ corresponding to the maximum possible support of the scattering environment\footnote{e.g., $\mathcal{C} = [k_{\min},k_{\max}] \times [l_{\min},l_{\max}]$ based on prior knowledge of the minimum/maximum delay and Doppler spreads of the scattering environment.}, $(k,l) \in \mathcal{C}$, $(k,l) \neq (0,0)$. Recall from Corollary~\ref{corr:max_comm_ambg} that our approach of transmitting an eigenvector of a maximal commutative subgroup $\mathcal{T}_{MN}$ yields self-ambiguity functions that are unimodular only at delay-Doppler indices in the set $\mathcal{S} = \big\{(k,l) \big| e^{\frac{j2\pi}{MN}m} \mathcal{D}_{(k,l)} \in \mathcal{T}_{MN} \big\}$, akin to a ``bed-of-nails''.

\begin{figure*}[!ht]
\centering
\begin{subfigure}{0.32\linewidth}
    \includegraphics[width=\textwidth]{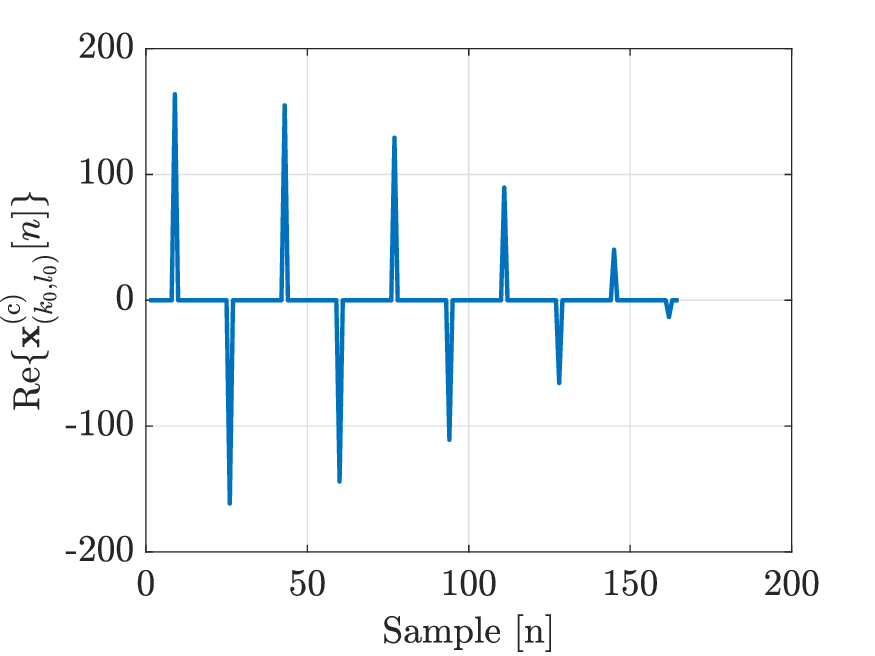}
\caption{Pulsone basis element.}
    \label{fig:tdwvf_pulsone}
\end{subfigure}
\begin{subfigure}{0.32\linewidth}
    \includegraphics[width=\textwidth]{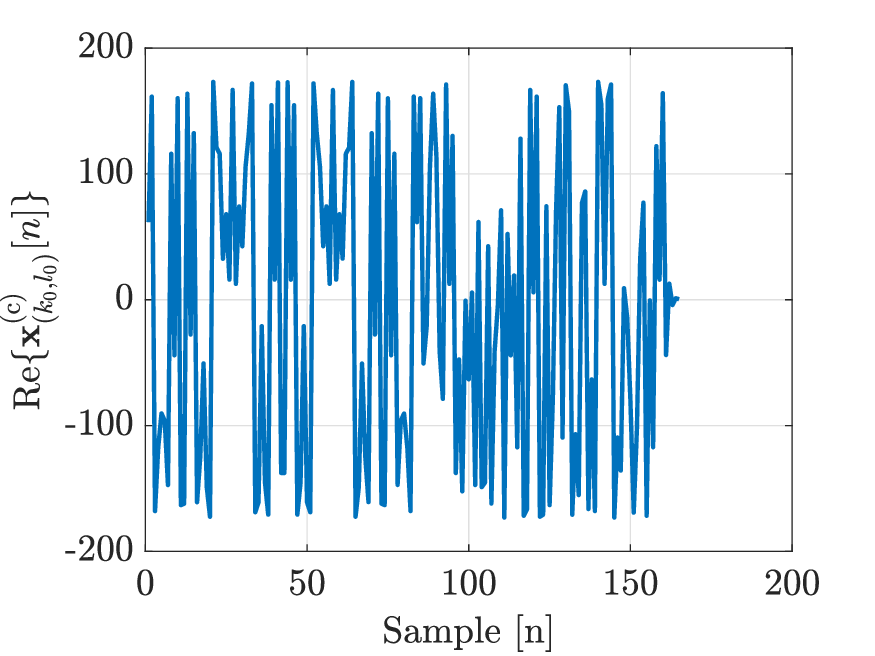}
\caption{GDAFT of pulsone.}
    \label{fig:tdwvf_cazac}
\end{subfigure}
\begin{subfigure}{0.32\linewidth}
    \includegraphics[width=\textwidth]{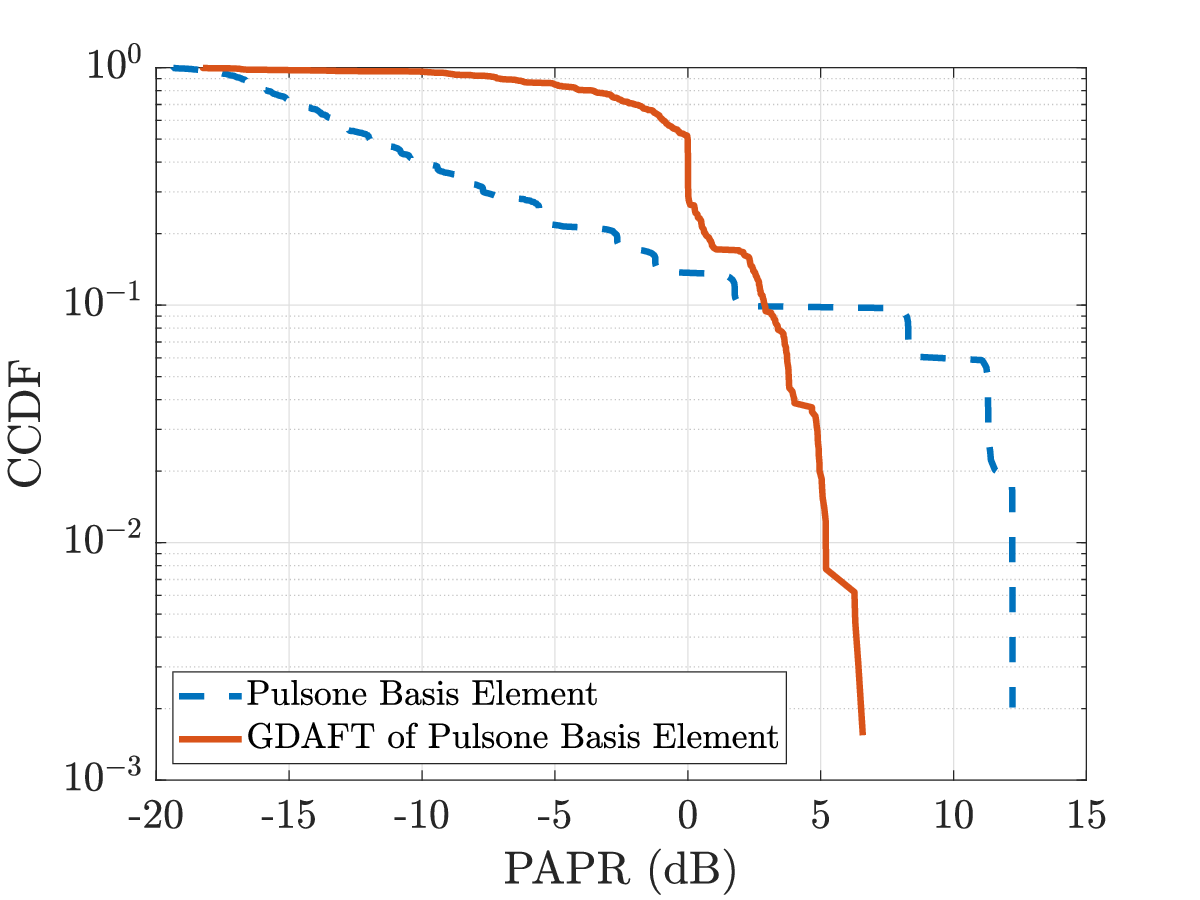}
\caption{PAPR comparison.}
    \label{fig:tdwvf_papr}
\end{subfigure}
\caption{The GDAFT (Definition~\ref{def:gdaft}) reduces the PAPR of the pulsone basis element (Example~\ref{ex:comm_subgrp_ex1}) by about $5.6$ dB. Figure adapted from~\cite{Mehrotra2025_WCLSpread,Mehrotra2025_EURASIP}.}
\vspace{-5mm}
    \label{fig:papr}
\end{figure*}

Therefore, we propose to choose $\mathcal{T}_{MN}$ such that translates of the maximum scattering environment support $\mathcal{C}$ by the elements of $\mathcal{S}$ do not overlap:
\begin{equation}
    \label{eq:no_aliasing}
    \bigg(\!\bigcup_{(k,l) \in \mathcal{S}}\!\big(\mathcal{C} + (k,l)\big)\!\bigg)\!\cap\!\bigg(\!\bigcup_{(k',l') \in \mathcal{S}}\!\big(\mathcal{C} + (k',l')\big)\!\bigg)\!=\!\emptyset.
\end{equation}
for all $(k,l) \neq (k',l')$.

This is the \emph{crystallization condition} from Section~\ref{subsec:crystallization}, stating that an image of the scattering environment can be read off from the response to the radar waveform at delay-Doppler locations around the locations in $\mathcal{S}$. Fig.~\ref{fig:no_aliasing_ex} illustrates how to choose $\mathcal{T}_{MN}$ for a given support $\mathcal{C}$.

\subsection{Benefit 2: Waveform Libraries with Small PAPR}
\label{subsec:papr_radar}

Another benefit of the proposed discrete radar architecture is that it enables defining waveform libraries with low peak-to-average power ratio (PAPR). Adaptive radars define libraries of waveforms from which an appropriate waveform is chosen in real-time based on the operational requirement~\cite{Moran2004_wvf_lib,Moran2009_wvf_lib_survey,Evans2002,Moran2006}. For instance, it was shown in~\cite{Moran2004_wvf_lib,Moran2009_wvf_lib_survey} that waveform libraries that rotate and chirp a template waveform maximize the mutual information for tracking applications. The ability to rotate a template waveform has also shown to be useful in sensing scattering environments where the product of the maximum delay and Doppler spreads exceeds the time-bandwidth product $BT$~\cite{Calderbank2025_isac}. At the same time, it is essential to ensure that each waveform in the library can be transmitted at low hardware cost. A standard metric to quantify the hardware cost of a waveform transmission is the PAPR~\cite{Jiang2008_papr}, which measures the ratio of the instantaneous power to the average power of a waveform. Waveforms with a large PAPR are undesirable from a hardware perspective since they require high dynamic range components and degrade the efficiency of power amplifiers.

The theory of symplectic transformations presented in Section~\ref{subsec:foundation_hwgroup} enables meeting both objectives, i.e., defines waveform libraries with low PAPR. Recall from Lemma~\ref{lmm:weil_prop} that symplectic transformations rotate the ambiguity functions of their inputs, e.g., the GDAFT in Definition~\ref{def:gdaft}. The GDAFT and similar transformations can be used to define a discrete radar waveform library following the principles in~\cite{Moran2004_wvf_lib,Moran2009_wvf_lib_survey}. 

A useful property of the GDAFT is that it maps pulsones to low PAPR waveforms, which have been used for spread carrier communication~\cite{Mehrotra2025_WCLSpread} as described in Section~\ref{subsec:zak_spread}. Fig.~\ref{fig:papr} illustrates how the GDAFT maps the ``peaky'' pulsone waveform (Fig.~\ref{fig:papr}(\subref{fig:tdwvf_pulsone})) to a constant amplitude waveform (Fig.~\ref{fig:papr}(\subref{fig:tdwvf_cazac})). Fig.~\ref{fig:papr}(\subref{fig:tdwvf_papr}) plots the complementary cumulative distribution function (CCDF) of the PAPR (in dB) of both waveforms, showing a $5.6$ dB PAPR reduction on applying the GDAFT.

\begin{figure*}[!t]
    \centering
    \begin{subfigure}{0.43\linewidth}
        \centering
        \includegraphics[width=\linewidth]{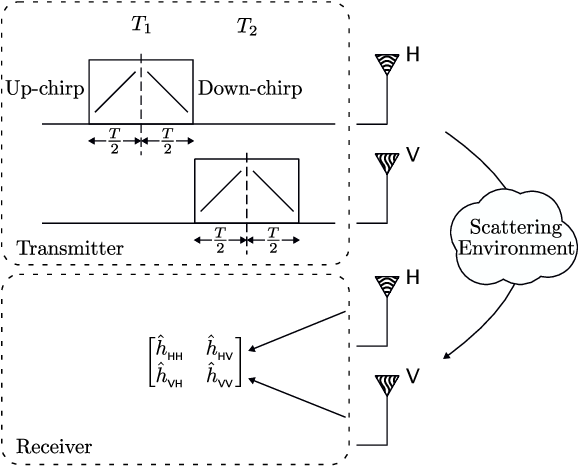}
        \caption{Sequential polarimetry via FMCW.}
        \label{fig:pol_fmcw}
    \end{subfigure}
    \hspace{15mm}
    \begin{subfigure}{0.39\linewidth}
        \centering
        \includegraphics[width=\linewidth]{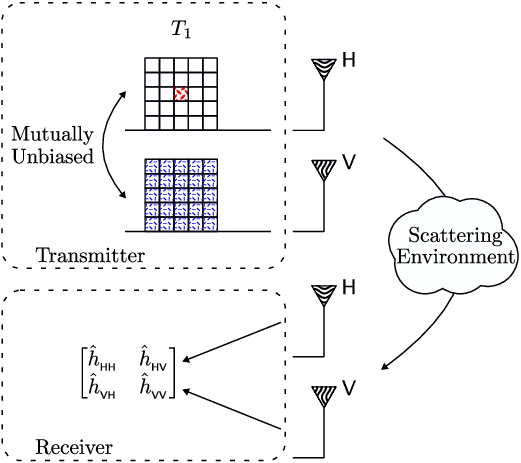}
        \caption{Proposed instantaneous polarimetry via Zak-OTFS.}
        \label{fig:pol_zak}
    \end{subfigure}
    \caption{Comparison of different approaches for polarimetry. (a) Sequential polarimetry with FMCW transmits polarized FMCW waveforms over two frames, with each frame subdivided into two halves with an up-chirp and a down-chirp respectively. The associated Doppler resolution is $\nicefrac{2}{T}$ and the computational complexity is $\mathcal{O}(B^2T^2)$. (b) Instantaneous polarimetry with Zak-OTFS transmits a Zak-OTFS pulsone and a mutually unbiased spread waveform obtained via a unitary transformation of the pulsone in a single frame. Compared to the sequential approach in (a), the proposed approach has $2 \times$ smaller latency, $2 \times$ improved Doppler resolution of $\nicefrac{1}{T}$, and a computational complexity of only $\mathcal{O}(BT \log T)$. Figure adapted from~\cite{Mehrotra2025_polTRS}.}
    \label{fig:pol_block_dia}
\end{figure*}

\subsection{Benefit 3: Low-Complexity Cross-Ambiguity Computation}
\label{subsec:low_compl_radar}

The third benefit of the proposed discrete radar architecture is the complexity gain in the cross-ambiguity computation over the conventional continuous and discrete radar architectures. As described previously, radars form an image of the scattering environment by computing the cross-ambiguity function from Definition~\ref{def:amb_fun} between the received and transmitted signals. In the following Lemma, we show that the complexity of calculating the cross-ambiguity with the proposed discrete radar architecture in Fig.~\ref{fig:block_diag}(\subref{fig:disc_prop_radar}) is only near-linear compared to quadratic complexity with conventional  radar architectures in Figs.~\ref{fig:block_diag}(\subref{fig:cont_radar})-(\subref{fig:disc_phase_coded_radar}).

\begin{lemma}[\cite{Mehrotra2025_EURASIP}]
    \label{lmm:low_compl_radar}
    The cross-ambiguity between two unit-norm $MN$-periodic sequences $\mathbf{x}$ and $\mathbf{y}$, where $\mathbf{y}$ is the Zak-OTFS pulsone from Example~\ref{ex:comm_subgrp_ex1}, has $\mathcal{O}(MN\log N)$ complexity. The complexity remains unchanged on applying a symplectic transformation to $\mathbf{x}$ and $\mathbf{y}$. However, when $\mathbf{y}$ is an amplitude-, phase- or frequency-coded waveform, the computational complexity is $\mathcal{O}(M^2N^2)$.
\end{lemma}

\begin{IEEEproof}
    Substituting the pulsone in Definition~\ref{def:amb_fun}:
    \begin{align}
        \label{eq:low_compl1}
        \mathbf{A}_{\mathbf{x},\mathbf{y}}[k, l] &= \sum_{n \in \mathbb{Z}_{MN}}\mathbf{x}[n]\mathbf{y}^*[(n-k)_{{}_{MN}}]e^{-\frac{j2\pi}{MN}l(n-k)} \nonumber \\
        &= \frac{1}{\sqrt{N}}\sum_{n \in \mathbb{Z}_{MN}}\mathbf{x}[n] e^{-\frac{j2\pi}{MN}l(n-k)} \sum_{d \in \mathbb{Z}} e^{-\frac{j2\pi}{N} d l_0} \nonumber \\ &~~~~~~~~~~~~~~~\times \delta[(n-k-k_0-dM)_{{}_{MN}}] \nonumber \\
        &= \frac{1}{\sqrt{N}}\sum_{\bar{k} \in \mathbb{Z}_{M}}\sum_{\bar{l} \in \mathbb{Z}_{N}}\mathbf{x}[\bar{k}+\bar{l}M] e^{-\frac{j2\pi}{MN}l(\bar{k}+\bar{l}M-k)} \nonumber \\ &\times \sum_{d \in \mathbb{Z}} e^{-\frac{j2\pi}{N} d l_0} \delta[(\bar{k}+\bar{l}M-k-k_0-dM)_{{}_{MN}}] \nonumber \\
        &= \frac{1}{\sqrt{N}} e^{-\frac{j2\pi}{MN} \big[l((k+k_0)_{{}_M}-k) - \lfloor \frac{k+k_0}{M} \rfloor l_0 M \big]} \nonumber \\ &\times \bigg(\sum_{\bar{l} \in \mathbb{Z}_{N}} \mathbf{x}[(k+k_0)_{{}_M}+\bar{l}M] e^{-\frac{j2\pi}{N}(l + l_0)\bar{l}}\bigg).
    \end{align}
    
    The summation within the parenthesis simply corresponds to an FFT operation, hence for a given $k \in \mathbb{Z}_{{}_{M}}$ its value for all $l \in \mathbb{Z}_{{}_{N}}$ can be computed in $\mathcal{O}(N\log N)$ complexity. Across all $(k,l) \in \mathbb{Z}_{{}_{M}} \times \mathbb{Z}_{{}_{N}}$, the complexity is $\mathcal{O}(MN\log N)$.

    A direct consequence of the above result in conjunction with Lemma~\ref{lmm:weil_prop} is that the complexity remains $\mathcal{O}(MN\log N)$ on applying a symplectic transformation to $\mathbf{x}$ and $\mathbf{y}$.

    The general expression of an amplitude-, phase- or frequency-coded waveform is~\cite{Pezeshki2009}:
    \begin{equation}
        \label{eq:radarcoding1}
        \mathbf{y}(t) = \sum_{m \in \mathbb{Z}_{MN}} \mathbf{z}_{m} s(t-m T_{\mathsf{c}}),
    \end{equation}
    where $\mathbf{z}$ denotes an $MN$-periodic coding sequence and $s(t)$ denotes the carrier waveform (``chip'') with chip duration $T_{\mathsf{c}}$. On sampling~\eqref{eq:radarcoding1} at sampling rate $\nicefrac{1}{T_{\mathsf{c}}}$, we obtain the $MN$-length sampled waveform:
    \begin{equation}
        \label{eq:radarcoding2}
        \mathbf{y}[n] = \sum_{m \in \mathbb{Z}_{MN}} \mathbf{z}_{m} s[(n-k)_{{}_{MN}}].
    \end{equation}

    Substituting~\eqref{eq:radarcoding2} in Definition~\ref{def:amb_fun}:
    \begin{align*}
        \mathbf{A}_{\mathbf{x},\mathbf{y}}[k, l]\!&=\!\sum_{n \in \mathbb{Z}_{MN}}\!\mathbf{x}[n]\mathbf{y}^*[(n-k)_{{}_{MN}}]e^{-\frac{j2\pi}{MN}l(n-k)} \nonumber \\
        &=\!\sum_{n,m}\!\mathbf{x}[n] \mathbf{z}^{*}_{m} s^{*}[(n-k-m)_{{}_{MN}}]e^{-\frac{j2\pi}{MN}l(n-k)}.
    \end{align*}

    In the absence of any additional structure on the coding sequence $\mathbf{z}$ and the sampled carrier waveform, computing the above expression requires $\mathcal{O}(M^2N^2)$ complexity. Moreover, it has been shown in~\cite{Calderbank2015_ltv} that the complexity of calculating the cross-ambiguity with a continuous radar architecture based on LFM waveforms is also $\mathcal{O}(M^2N^2)$.
\end{IEEEproof}

\subsection{Benefit 4: Instantaneous Polarimetry at Low Complexity}
\label{subsec:pol_radar}

The proposed architecture also enables polarimetry at near-linear computational complexity. Polarimetry is the ability to measure the scattering response of the environment across orthogonal polarizations, and is an important tool for enhancing the performance of both wireless communication and radar systems. In wireless communication, polarimetry provides a diversity gain~\cite{Vaughan1990_pol_div,Paulraj2002_pol_div,Valuenzela2002_pol_div,Mark2006_pol_div}, thereby improving the reliability of communication, as well as a spatial multiplexing gain~\cite{Andrews2001_pol_dof,Marzetta2002_pol_dof,Hughes2008_pol_dof,Poon2011_pol_dof}, which increases the capacity of the wireless link. Similarly, polarimetry increases the waveform degrees-of-freedom in radar systems~\cite{Boerner1990_polsurvey,Nehorai2009_polsurvey,Calderbank2009_wvfsurvey,Antar2002_polcomparison,Giuli1990_polsimult1,howard2007simple,Pezeshki2008,Calderbank2006_pol_phasecoded,Hochwald1995_polmodel,Cloude2005_poleig}, providing more information about the target and enabling improved detection of targets with small radar cross section (RCS), such as drones.

Polarimetry is enabled in radar and communication systems by transmitting and receiving on two orthogonal polarizations, e.g., on vertical and horizontal polarizations. The receiver estimates the $2 \times 2$ \emph{polarimetric scattering response} of the wireless/radar channel across all four combinations of transmit and receive polarizations. A standard approach is to transmit polarized waveforms \emph{sequentially} across two frames~\cite{Calderbank2009_wvfsurvey,Antar2002_polcomparison,Giuli1990_polsimult1,howard2007simple,Pezeshki2008}; see Fig.~\ref{fig:pol_block_dia}(\subref{fig:pol_fmcw}) for an example with frequency modulated continuous wave (FMCW) transmissions. From its measurements in each frame, the receiver estimates $2 \times 1$ slices of the full $2 \times 2$ polarimetric scattering response. Such an approach does not provide instantaneous estimates of the scattering response within a single frame. Changes in the scattering environment between the two frames (due to mobility) may partially decorrelate the obtained estimates~\cite{Antar2002_polcomparison,Giuli1990_polsimult1,howard2007simple,Pezeshki2008,Calderbank2006_pol_phasecoded}. FMCW systems also result in a Doppler resolution of $\nicefrac{2}{T}$ in a frame interval $T$ due to the need to transmit two chirps with positive slope (``up-chirp'') and negative slope (``down-chirp'') within the same frame, as shown in Fig.~\ref{fig:pol_block_dia}(\subref{fig:pol_fmcw}). Sequential polarimetry also prevents frame-by-frame processing \& increases the system latency, which is a critical factor for radar and communication performance in highly dynamic environments. When utilizing continuous waveforms, such as FMCW and pulsed waveforms~\cite{Jankiraman2018_fmcw,Uysal2020_phasecoded_fmcw,Skolnik1980,Levanon2004}, the computational complexity of sequential polarimetry is \emph{quadratic} in the time-bandwidth product~\cite{Calderbank2015_ltv,zakotfs_ltv,Mehrotra2025_EURASIP}.

To unlock the full benefits of polarimetry, it would be ideal to estimate the $2 \times 2$ polarimetric scattering response \emph{instantaneously} within a single transmission frame. Previous work~\cite{Giuli1990_polsimult1,howard2007simple,Pezeshki2008} transmits \emph{mutually unbiased} waveforms\footnote{\textcolor{black}{The term ``mutually unbiased'' is from quantum information theory~\cite{Schwinger1960unitary}. Formally, two $d$-length waveforms are mutually unbiased if their inner product has magnitude $\nicefrac{1}{\sqrt{d}}$. Measurements from one waveform are ``statistically independent'' to those from the other waveform with uniform probability $\nicefrac{1}{d}$.}} with small inner products simultaneously across orthogonal polarizations. Mutual unbiasedness ensures that the contribution of the other waveform looks like noise to the receiver when projected onto the basis of one of the transmit waveforms -- thus enabling the receiver to estimate the full $2 \times 2$ polarimetric scattering response estimation from a single received frame. Mutually unbiased waveforms have been designed in prior work~\cite{Giuli1990_polsimult1,howard2007simple,Pezeshki2008} by \emph{phase-coding} a common carrier waveform, e.g., a rectangular waveform, with mutually unbiased sequences, e.g., Zadoff-Chu sequences with distinct roots~\cite{Giuli1990_polsimult1} or complementary Golay pairs~\cite{howard2007simple,Pezeshki2008}. However, as shown in Section~\ref{subsec:low_compl_radar}, the computational complexity of phase-coding is \emph{quadratic} in the time-bandwidth product.

In~\cite{Mehrotra2025_polTRS}, we utilize the Zak-OTFS framework to design mutually unbiased waveforms as opposed to phase-coding. Specifically, we transmit a pulsone and a GDAFT-transformed pulsone -- which is a chirp as described in Section~\ref{subsec:foundation_hwgroup} and is mutually unbiased to the pulsone -- across orthogonal polarizations. The proposed approach enables instantaneous polarimetry at \emph{near-linear} computational complexity, with greater clutter resilience compared to phase-coding. 

Table~\ref{tab:prior_work} summarizes the contribution of~\cite{Mehrotra2025_polTRS}.

\begin{table}
    \centering
    \caption{Comparison of different approaches for polarimetry; $B$ denotes frame bandwidth, $T$ denotes frame interval.}
    {
    \setlength{\tabcolsep}{2.25pt}
    \renewcommand{\arraystretch}{1.25}
    \begin{tabular}{|c|c|c|c|c|}
         \hline
         Approach & Frame(s) & Target(s) & Doppler Res. & Complexity \\
         \hline
         \textbf{Zak-OTFS~\cite{Mehrotra2025_polTRS}} & \textbf{1} & $\mathbf{> 1}$ & $\mathbf{\nicefrac{1}{T}}$ & $\mathbf{\mathcal{O}(BT\log T)}$ \\
         Phase-coded~\cite{Giuli1990_polsimult1,howard2007simple,Pezeshki2008} & $1$ & $> 1$ & $\nicefrac{1}{T}$ & $\mathcal{O}(B^2T^2)$ \\
         FMCW~\cite{Antar2002_polcomparison,Jankiraman2018_fmcw} & $2$ & $1$ & $\nicefrac{2}{T}$ & $\mathcal{O}(B^2T^2)$ \\
         \hline
    \end{tabular}
    }
    \vspace*{-0.1in}
    \label{tab:prior_work}
\end{table}

\subsubsection{Polarimetric System Model}
\label{subsubsec:pol_sys_model}

We begin by modeling the polarimetric scattering response of a $P$-path channel. In uni-polarized systems, e.g. in~\eqref{eq:hphy}, the channel gain of each path $p \in \{1,\cdots,P\}$ is modeled by a complex scalar $h_{p}$. With dual-polarized transmit and receive antennas (along horizontal $\mathsf{H}$ and vertical $\mathsf{V}$ polarizations) the channel gain is modeled by a $2 \times 2$ \emph{polarimetric scattering response}~\cite{Nehorai2009_polsurvey,Calderbank2009_wvfsurvey,Antar2002_polcomparison,Giuli1990_polsimult1,howard2007simple,Pezeshki2008}:
\begin{align}
    \label{eq:pol_prelim1}
    \mathbf{H}_{p} &= \begin{bmatrix}
        h_{p}^{\mathsf{HH}} & h_{p}^{\mathsf{HV}} \\ h_{p}^{\mathsf{VH}} & h_{p}^{\mathsf{VV}}
    \end{bmatrix} = \mathbf{C}_{\mathsf{RX}} \mathbf{\Sigma}_{p} \mathbf{C}_{\mathsf{TX}},
\end{align}
where $\mathbf{C}_{\mathsf{TX}}$ (resp. $\mathbf{C}_{\mathsf{RX}}$) is a $2 \times 2$ matrix characterizing the polarization coupling at the transmitter (resp. receiver), and $\mathbf{\Sigma}_{p}$ is a $2 \times 2$ matrix of polarimetric scattering coefficients of the $p$th path/target\footnote{$\mathbf{\Sigma}_{p} = \mathbf{I}_{2}$ for a line-of-sight path with no reflection.}. Thus, the dual-polarized generalization of the system model in~\eqref{eq:prelim2} is:
\begin{align}
    \label{eq:pol1}
    \mathbf{y}^{(j)}[n]\!&=\!\sum_{i \in \{\mathsf{V},\mathsf{H}\}}\!\sum_{k,l \in \mathbb{Z}_{MN}}\!\mathbf{h}_{\mathsf{eff}}^{(j,i)}[k,l] \mathbf{x}^{(i)}[(n-k)_{{}_{MN}}] e^{\frac{j2\pi}{MN}l(n-k)} \nonumber \\ &~~~~~~~~~~~~~~~~~~~~+ \mathbf{w}^{(j)}[n],~i,j \in \{\mathsf{V},\mathsf{H}\},
\end{align}
where $\mathbf{x}^{(i)}$ denotes the signal transmitted by the $i$-polarized transmit antenna, $\mathbf{h}_{\mathsf{eff}}^{(j, i)}[k, l]$ denotes the effective channel between the $i$-polarized transmit antenna and the $j$-polarized receive antenna, and $\mathbf{w}^{(j)}[n]$ denotes the additive noise at the $j$-polarized receive antenna. In~\eqref{eq:pol1}, the polarimetric effective channel $\mathbf{h}_{\mathsf{eff}}^{(j, i)}[k, l]$ is defined similarly as in Section~\ref{subsec:zak_intro_to_zak} using the polarimetric continuous DD channel representation, $\mathbf{h}_{\mathrm{phy}}^{(j, i)}(\tau,\nu) = \sum_{p=1}^{P} h_{p}^{ji} \delta(\tau-\tau_p) \delta(\nu-\nu_p)$, where $h_{p}^{ji}$ is $(j,i)$th entry of the matrix $\mathbf{H}_{p}$ in~\eqref{eq:pol_prelim1}.

\begin{figure*}[!t]
    \centering
    \begin{subfigure}{0.45\linewidth}
        \centering
        \includegraphics[width=\linewidth]{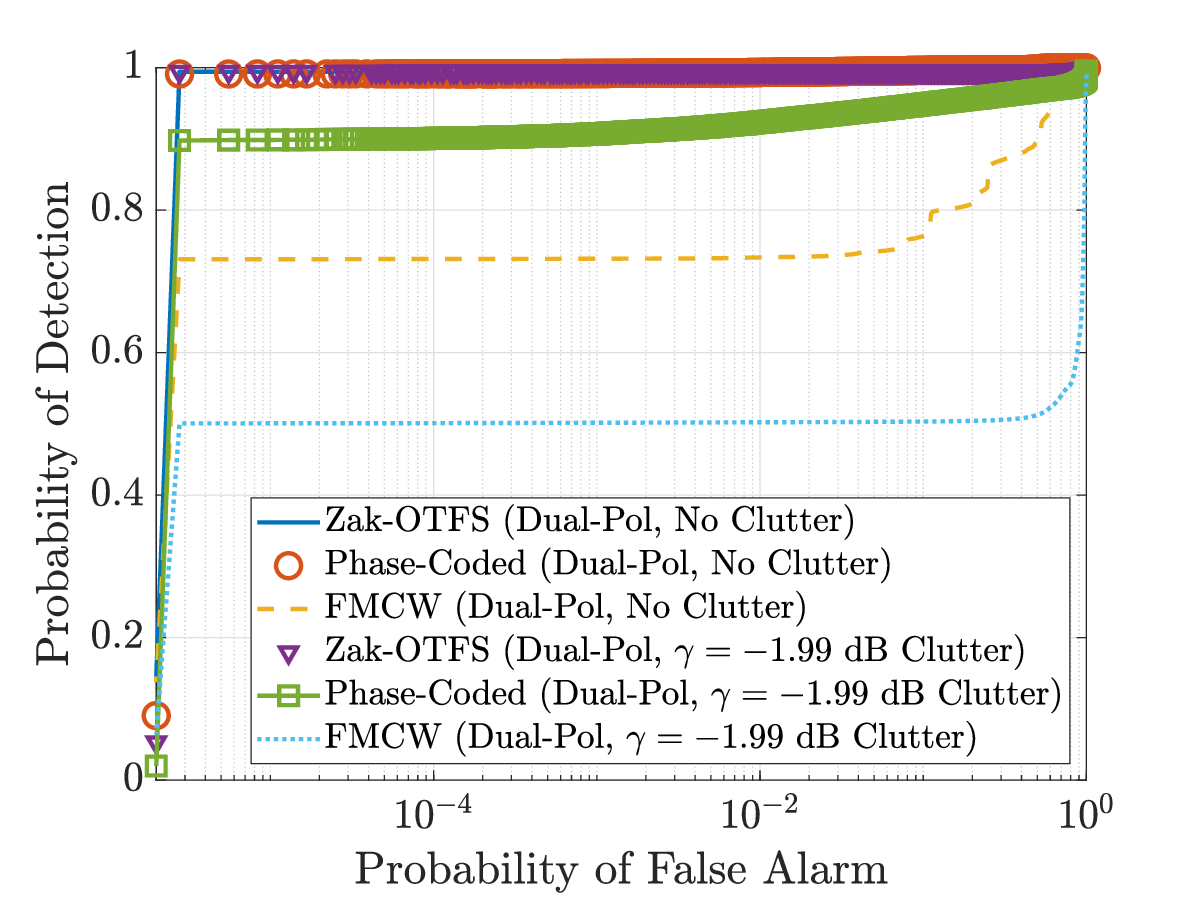}
        \caption{Detection ROC in constant-$\gamma$ clutter.}
        \label{fig:pol_roc}
    \end{subfigure}
    \hspace{15mm}
    \begin{subfigure}{0.45\linewidth}
        \centering
        \includegraphics[width=\linewidth]{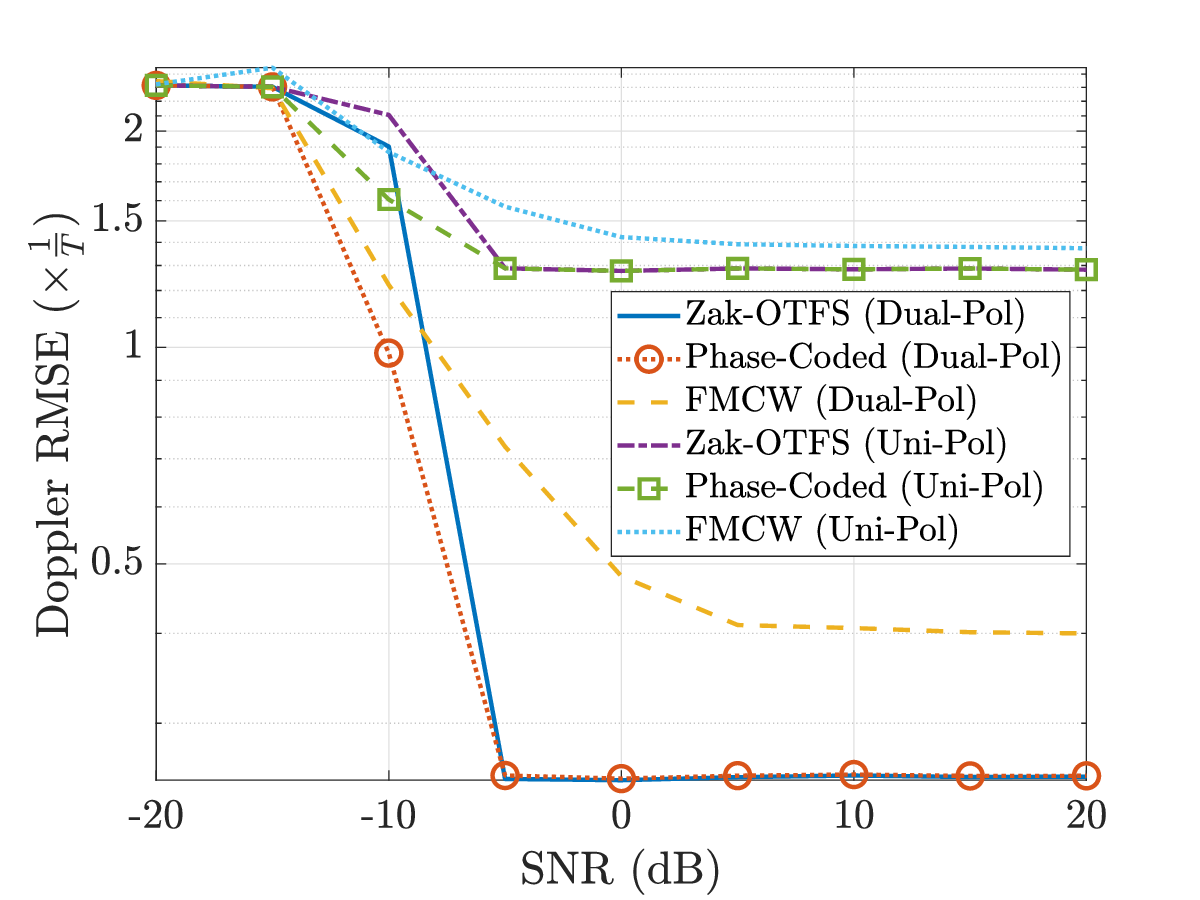}
        \caption{Doppler estimation RMSE (no clutter).}
        \label{fig:pol_mse}
    \end{subfigure}
    \caption{Single target detection and estimation performance. (a) Receiver operating characteristic (ROC) curve showing ideal target detection and greater clutter resilience with the proposed approach over phase-coded and FMCW waveforms. (b) Root mean squared error (RMSE) for Doppler estimation, normalized by the Doppler resolution $\nicefrac{1}{T}$. Doppler RMSE is similar for dual-polarized Zak-OTFS and phase-coded waveforms at high SNR, with $\sim 1.5 \times$ improvement over FMCW due to no loss in Doppler resolution. Significant improvements with dual-polarized vs uni-polarized waveforms. Figure adapted from~\cite{Mehrotra2025_polTRS}.}
    \label{fig:pol_results}
\end{figure*}

\subsubsection{Proposed Approach}
\label{subsubsec:pol_prop_appr}

To enable instantaneous polarimetry, we exploit the following \emph{mutual unbiasedness} property of the Zak-OTFS pulsone from Example~\ref{ex:comm_subgrp_ex1} and its GDAFT-transformed chirp waveform from Example~\ref{ex:comm_subgrp_ex2}:
\begin{align}
    \label{eq:unbiased}
    \mathbf{A}_{\mathbf{c},\mathbf{p}_{(k_0,l_0)}}[k,l] &= \frac{C_{(k_0,l_0)}[k,l]}{\sqrt{MN}},
\end{align}
where $\mathbf{c}$ denotes the output of the GDAFT applied to the pulsone waveform $\mathbf{p}_{(k_0,l_0)}$ localized at delay-Doppler indices $(k_0,l_0)$, $C_{(k_0,l_0)}[k,l]$ is a complex phase, $\big|C_{(k_0,l_0)}[k,l] \big| = 1$, and $\mathbf{A}_{\mathbf{y},\mathbf{x}}[k,l]$ denotes the cross-ambiguity function from Definition~\ref{def:amb_fun}. For instantaneous polarimetry, we transmit the two waveforms along orthogonal polarizations, e.g.,
\begin{align}
    \label{eq:prop1}
    \mathbf{x}^{(\mathsf{H})}[n] = \mathbf{p}_{(k_0,l_0)}[n],~\mathbf{x}^{(\mathsf{V})}[n] = \mathbf{c}[n].
\end{align}

For accurate channel estimation, the GDAFT parameters are chosen such that $\mathbf{c}[n]$ satisfies the crystallization condition in~\eqref{eq:no_aliasing} for all four components of the polarimetric effective channel $\mathbf{h}_{\mathsf{eff}}^{(j, i)}[k, l]$, for all $i,j \in \{\mathsf{V},\mathsf{H}\}$.

Mutual unbiasedness per~\eqref{eq:unbiased} enables accurate estimation of all four polarimetric effective channels $\mathbf{h}_{\mathsf{eff}}^{(j, i)}[k, l]$. Without loss of generality, we prove the result for the example considered in~\eqref{eq:prop1}. The estimate of $\mathbf{h}_{\mathsf{eff}}^{(j, i)}[k, l]$ from~\eqref{eq:prelim6} and~\eqref{eq:pol1} is:
\begin{align}
    \label{eq:prop3}
    \widehat{\mathbf{h}}_{\mathsf{eff}}^{(j, i)}[k,l] &= \sum_{i' \in \{\mathsf{V},\mathsf{H}\}} \mathbf{h}_{\mathsf{eff}}^{(j,i')}[k,l] \ast_{\sigma_{d}} \mathbf{A}_{\mathbf{x}^{(i')},\mathbf{x}^{(i)}}[k,l],
\end{align}
where $\ast_{\sigma_{d}}$ denotes discrete twisted convolution per~\eqref{eq:sec3eq22}.

The expression in~\eqref{eq:prop3} is the sum of two terms:
\begin{align}
    \label{eq:prop4}
    \widehat{\mathbf{h}}_{\mathsf{eff}}^{(j, i)}[k,l] &= \mathbf{h}_{\mathsf{eff}}^{(j,i)}[k,l] \ast_{\sigma_{d}} \mathbf{A}_{\mathbf{x}^{(i)},\mathbf{x}^{(i)}}[k,l] \nonumber \\ &+ \mathbf{h}_{\mathsf{eff}}^{(j,\bar{i})}[k,l] \ast_{\sigma_{d}} \mathbf{A}_{\mathbf{x}^{(\bar{i})},\mathbf{x}^{(i)}}[k,l],
\end{align}
where $\bar{i}$ denotes a polarization different from $i$ in the set $\{\mathsf{V},\mathsf{H}\}$. Since each sequence $\mathbf{x}^{(i)}$ satisfies the crystallization condition in~\eqref{eq:no_aliasing}, the first term is simply $\mathbf{h}_{\mathsf{eff}}^{(j,i)}[k,l]$. To simplify the second term, we substitute~\eqref{eq:unbiased} to obtain:
\begin{align}
    \label{eq:prop5}
    \widehat{\mathbf{h}}_{\mathsf{eff}}^{(j, i)}[k,l] &= \mathbf{h}_{\mathsf{eff}}^{(j,i)}[k,l] + \mathbf{h}_{\mathsf{eff}}^{(j,\bar{i})}[k,l] \ast_{\sigma_{d}} \frac{C[k,l]}{\sqrt{MN}} \nonumber \\
    &\approx \mathbf{h}_{\mathsf{eff}}^{(j,i)}[k,l],
\end{align}
where $C[k,l]$ is a phase term similar to that in~\eqref{eq:unbiased}. Since the second term is the twisted convolution of the effective channel $\mathbf{h}_{\mathsf{eff}}^{(j,\bar{i})}[k,l]$ with a constant amplitude term, it simply raises the noise floor of the channel estimate. Computing each cross-ambiguity term only incurs $\mathcal{O}(BT \log T)$ complexity per Lemma~\ref{lmm:low_compl_radar}, i.e., the overall complexity is \emph{near-linear} in $BT$.

\emph{Remark:} The estimate shown in \eqref{eq:prop5} is the maximum-likelihood estimate of the channel. Compared to the FMCW, Zak-OTFS based radar achieves two times improvement in the Doppler resolution. This is because FMCW splits the available duration into two halves where in each half the FMCW transmits an up and a down chirp. Furthermore, in the presence of multiple targets, elimination of ghost targets in FMCW is a challenge (see~\cite{nisar2026zak} for more details). Also, Zak-OTFS allows constructing waveform libraries with desired self-ambiguity properties via symplectic transformations such as GDAFT (see Definition~\ref{def:gdaft}). However, no such theory has been developed for MC-OTFS. Lastly, the response to a pilot in MC-OTFS is dependent on its location whereas Zak-OTFS provides uniform response regardless of the pilot location -- a property that might be desirable for radar.

\emph{Remark:} In a discrete-time model for $MN$-periodic sequences, the received vector $\mathbf{y}[n]$ is obtained from the transmitted vector $\mathbf{x}[n]$ by (twisted) convolution. Consider the problem of selecting an unbiased estimator theta for the scattering environment $\mathbf{h}[k, l]$, where $k=0,1, \cdots , M-1$ and $l=0,1, \cdots , N-1$. The Cramer-Rao Lower Bound (CRLB) is the minimum covariance of such an estimator and is given by $\mathbf{F}^{-1}(\theta)$, where $\mathbf{F}(\theta)$ is the Fisher Information Matrix (FIM). The FIM can be written in terms of the self-ambiguity function of the input $\mathbf{x}[n]$. The structure of the FIM $\mathbf{F}(\theta)$ is then determined by the symmetries of the waveform $\mathbf{x}[n]$, which is a reason for choosing highly symmetric waveforms. When $\mathbf{x}[n]$ is a Zak-OTFS carrier waveform, and when $\mathbf{h}[k, l]$ is supported in $\mathbb{Z}_M \times \mathbb{Z}_N$ (the crystallization condition), the FIM $\mathbf{F}(\theta)$ turns out to be the identity matrix. The interpretation is that the Zak-OTFS carrier waveform is an optimal unbiased estimator. The same is true for an AFDM carrier waveform (given a corresponding crystallization condition), as is expected since it is unitarily equivalent to the Zak-OTFS waveform. When $\mathbf{x}[n]$ is a discrete tone or a discrete pulse, the FIM $\mathbf{F}(\theta)$ has a block structure, is highly singular, and the non-zero eigenvalues are greater than one. The interpretation is that the carrier waveform is optimal at resolving delay (better than a spread waveform) and is blind to Doppler, or vice versa (cf.~\cite{liu2025cp, liu2026cp}).

\subsubsection{Numerical Results}
\label{subsubsec:pol_results}

We compare the performance of various polarimetry schemes in Fig.~\ref{fig:pol_results} for a monostatic polarimetric radar with frame transmissions of bandwidth $B = 930$ kHz and time $T = 1.2$ ms. Fig.~\ref{fig:pol_results}(\subref{fig:pol_roc}) plots the detection ROC curve for single-target detection in constant-$\gamma$ clutter, where $\gamma = -1.99$ dB is used to model a metropolitan terrain. Fig.~\ref{fig:pol_results}(\subref{fig:pol_roc}) shows ideal target detection performance and greater resilience to clutter with the proposed approach over competing methods based on phase-coding and FMCW transmissions. Fig.~\ref{fig:pol_results}(\subref{fig:pol_mse}) plots the root-mean-squared error (RMSE) for Doppler estimation in the absence of clutter. The Doppler RMSE matches for polarimetry with phase-coding and Zak-OTFS, with $\sim 1.5 \times$ improvement over FMCW as expected. Uni-polarized systems have significantly poorer delay and Doppler RMSE. Note that the RMSEs do not improve beyond a certain threshold due to the inherent resolution limits of the chosen waveforms. 

\subsection{Open Research Problems}
\label{subsec:radar_future}

In this Section, we have taken first steps towards illustrating the radar capabilities using the Zak-OTFS framework. It is an open problem to establish the value that Zak-OTFS brings to traditional radar signal processing, e.g., to clutter rejection, near-field, wideband and MIMO radar. Furthermore, is it possible to adapt the neural receiver architecture described in Section~\ref{sec:neural_rx} to enable radar signal processing without explicit signal processing blocks, such as the matched filter and the cross-ambiguity function?

\section{Conclusion}
\label{sec:conclusion}
We have presented a mathematical framework that provides a foundation for signal processing in the DD domain. At the center of our framework is a finite group (HW) of discrete delay and Doppler shifts, called the discrete Heisenberg-Weyl group. Given a basis of carrier waveforms, we have characterized properties that are useful for sensing and communications in terms of how the basis of carrier waveforms interacts with the discrete group HW. The Zak-OTFS basis is ideal for sensing because each carrier waveform is a pulse in the DD domain, and we have presented a general design framework that encompasses other waveforms proposed for 6G. Our design framework makes it possible to construct mutually unbiased bases of carrier waveforms, and we have described how these bases can be used to accomplish Faster than Nyquist signaling without changing the symbol period. We have described how to use the delay and Doppler degrees of freedom to design equalizers and filters that balance localization and fidelity of communications. We have described how a filter optimized for localization enables grant-free random access for channels subject to both mobility and delay spread. We have designed a discrete radar by repurposing the discrete-time model for communications and shown how mathematical properties of the Zak-OTFS basis reduce the complexity of calculating cross-correlation. At the end of each Section, we have listed open research problems that we hope will attract the attention of the research community.

\section{Acknowledgements}
\label{sec:acks}

We thank Samsung Research America for sharing their insights on coding and shaping, Prof. Krishna Narayanan (Texas A\&M University) for sharing his perspective on coded random access. We also thank Prof. Lingjia Liu, Dr. Karim Said and Mr. Yibin Liang (Virginia Tech) for sharing their insights into the neural receiver, Prof. Tingjun Chen (Duke University) for sharing his expertise on experimental over-the-air implementation, and Prof. Saif Khan Mohammed (IIT Delhi) for numerous discussions.

\bibliographystyle{IEEEtran}
\bibliography{references}
\end{document}